\documentclass[sigconf]{acmart}

\usepackage[ruled,vlined,linesnumbered]{algorithm2e}
\usepackage{xspace}
\usepackage{amsthm,amssymb,amsmath,epsfig,graphics,enumerate,float,subfigure}
\usepackage{epstopdf}
\usepackage{svg}
\usepackage{balance}
\usepackage{color}
\usepackage{setspace}
\usepackage{bm}
\usepackage{booktabs}
\usepackage{multirow}
\usepackage{graphicx}
\usepackage{array}

\SetKwRepeat{Do}{do}{while}

\newcommand{\ignore}[1]{}
\newcommand{\nop}[1]{}
\newcommand{\eat}[1]{}
\newcommand{\kw}[1]{{\ensuremath{\mathsf{#1}}}\xspace}

\newcommand{\stitle}[1]{\par\noindent{\bf #1}}
\long\def\comment#1{}

\newcommand{\sota}{\kw{SOTA}}

\newcommand{\lca}{\kw{LCA}}

\newcommand{\mde}{\kw{MDE}}

\newtheorem{definition}{Definition}
\newtheorem{problem}{Problem}

\newtheorem{example}{Example}

\newcommand{\snap}{\kw{SNAP}}
\newcommand{\dimacs}{\kw{DIMACS}}

\newcommand{\dfs}{\kw{DFS}}

\newcommand{\resistancedistance}{\kw{RD}}
\newcommand{\resistancepath}{\kw{RD}}

\newcommand{\penaltymethod}{\kw{Penalty}}
\newcommand{\plateaumethod}{\kw{Plateau}}

\newcommand{\res}{\kw{res}}
\newcommand{\order}{\kw{order}}
\newcommand{\dfsorder}{\kw{DFSOrder}}
\newcommand{\diagonal}{\kw{Diagonal}}
\newcommand{\subtree}{\kw{SubTree}}

\newcommand{\Parent}{\kw{Parent}}
\newcommand{\rootvector}{\kw{root}}
\newcommand{\column}{\kw{Col}}
\newcommand{\ratio}{\kw{ratio}}
\newcommand{\tw}{\kw{tw}}

\newcommand{\singlepairn}{\kw{SP\textrm{-}\xspace N}}
\newcommand{\length}{\kw{Length}}
\newcommand{\diversity}{\kw{Diversity}}
\newcommand{\robustness}{\kw{Robustness}}

\newcommand{\push}{\kw{push}}
\newcommand{\bipush}{\kw{BiPush}}

\newcommand{\geer}{\kw{GEER}}

\newcommand{\tedi}{\kw{TEDI}}
\newcommand{\multihop}{\kw{MultiHop}}

\newcommand{\lapsolver}{\kw{LapSolver}}
\newcommand{\treeindex}{\kw{TreeIndex}}
\newcommand{\leindex}{\kw{LEIndex}}
\newcommand{\lewalk}{\kw{LEwalk}}
\newcommand{\htwoh}{\kw{H2H}}

\newcommand{\youtube}{\kw{Youtube}}

\newcommand{\dblp}{\kw{DBLP}}
\newcommand{\amazon}{\kw{Amazon}}

\newcommand{\emailenron}{\kw{Email\textrm{-}\xspace enron}}

\newcommand{\roadca}{\kw{Road\textrm{-}CA}}
\newcommand{\roadpa}{\kw{Road\textrm{-}PA}}

\newcommand{\newyork}{\kw{NewYork}}
\newcommand{\roadtx}{\kw{Road\textrm{-}TX}}
\newcommand{\western}{\kw{Western}}
\newcommand{\roadctr}{\kw{Road\textrm{-}CTR}}
\newcommand{\fullusa}{\kw{Full\textrm{-}USA}}

\AtBeginDocument{%
  \providecommand\BibTeX{{%
    \normalfont B\kern-0.5em{\scshape i\kern-0.25em b}\kern-0.8em\TeX}}}
\setcopyright{acmcopyright}
\copyrightyear{2022}
\acmYear{2022}

\acmConference[SIGMOD '26]{SIGMOD '26: International Conference on Management of Data}{June 03--05, 2026}{Woodstock, NY}
\acmPrice{15.00}

\begin{document}
\title{Efficient Exact Resistance Distance Computation on Small-Treewidth Graphs: a Labelling Approach}
\author{Meihao Liao}
\email{mhliao@bit.edu.cn}
\affiliation{%
	\institution{Beijing Institute of Technology}
	\city{Beijing}
	\country{China}}

\author{Yueyang Pan}
\email{yyp@bit.edu.cn}
\affiliation{%
	\institution{Beijing Institute of Technology}
	\city{Beijing}
	\country{China}}

\author{Rong-Hua Li}
\email{lironghuabit@126.com}
\affiliation{%
	\institution{Beijing Institute of Technology}
	\city{Beijing}
	\country{China}}

\author{Guoren Wang}
\email{wanggrbit@126.com}
\affiliation{%
	\institution{Beijing Institute of Technology}
	\city{Beijing}
	\country{China}}

\begin{abstract}
	Resistance distance computation is a fundamental problem in graph analysis, yet existing random walk-based methods are limited to approximate solutions and suffer from poor efficiency on small-treewidth graphs (e.g., road networks). In contrast, shortest-path distance computation achieves remarkable efficiency on such graphs by leveraging cut properties and tree decompositions. Motivated by this disparity, we first analyze the cut property of resistance distance. While a direct generalization proves impractical due to costly matrix operations, we overcome this limitation by integrating tree decompositions, revealing that the resistance distance $r(s,t)$ depends only on labels along the paths from $s$ and $t$ to the root of the decomposition. This insight enables compact labelling structures. Based on this, we propose \treeindex, a novel index method that constructs a resistance distance labelling of size $O(n \cdot h_{\mathcal{G}})$ in $O(n \cdot h_{\mathcal{G}}^2 \cdot d_{\max})$ time, where $h_{\mathcal{G}}$ (tree height) and $d_{\max}$ (maximum degree) behave as small constants in many real-world small-treewidth graphs (e.g., road networks). Our labelling supports exact single-pair queries in $O(h_{\mathcal{G}})$ time and single-source queries in $O(n \cdot h_{\mathcal{G}})$ time. Extensive experiments show that TreeIndex substantially outperforms state-of-the-art approaches. For instance, on the full USA road network, it constructs a $405$ GB labelling in $7$ hours (single-threaded) and answers exact single-pair queries in $10^{-3}$ seconds and single-source queries in $190$ seconds--the first exact method scalable to such large graphs.
\end{abstract}
\maketitle
\section{Introduction}\label{sec:intro}
Resistance distance~\cite{tetali1991random}, recognized for its robustness and smoothness compared to shortest path distance, has recently garnered significant attention in the graph data management community. Its applications span a diverse array of domains, including link prediction in social networks~\cite{linkprediction2018spectralembedding,long-tail-recommendation-vldb12}, graph clustering in geo-spatial networks~\cite{density-clustering-sigmod14,GraphClusteringITCS18}, and robust routing in road networks~\cite{RobustRouting21}. Furthermore, it has found utility in analyzing over-smoothing and over-squashing issues in graph neural networks~\cite{OversquashingICLR22,OversquashingWWW23,GraphVariance,GraphCurvature}. Nevertheless, resistance distance computation remains computationally challenging, primarily because it requires solving a linear system involving the graph Laplacian matrix.

Existing methods for computing resistance distances predominantly rely on random walk-based approximation techniques~\cite{KDDlocal21,22resistance,23resistance,ResistanceYang,SpanningEdgeCentrality}. Although these approaches scale effectively to large graphs, they inherently sacrifice exactness in favor of computational efficiency. Furthermore, random walk-based techniques are highly sensitive to the spectral properties of the underlying graph. Let $\lambda_2$ denote the second smallest eigenvalue of the graph's Laplacian matrix; random walks are known to mix rapidly when $\lambda_2$ is large~\cite{chung1997spectral}. Thus, random walk-based methods have been demonstrated to perform effectively on rapidly mixing graphs, such as scale-free social networks~\cite{KDDlocal21,22resistance,ResistanceYang}. However, many real-world graphs do not exhibit rapid mixing behavior~\cite{mixing2010measuring,mixing18}. Tree-width, a measure quantifying the closeness of a graph's structure to a tree~\cite{TreeWidth84}, is particularly relevant in this context. Road networks, characterized by small tree-width, are known to be easily separable, implying that $\lambda_2$ typically approaches zero according to Cheeger's inequality~\cite{chung1997spectral}. Consequently, random walk-based algorithms suffer significant performance degradation on graphs with small tree-width, including road networks~\cite{23resistance}. For example, our experimental results indicate that even state-of-the-art index-based solutions \leindex \cite{23resistance} for computing resistance distances on large road networks require approximately $1,000$ seconds to achieve an absolute error of merely $10^{-1}$. Such inefficiency significantly limits the practical applicability of resistance distance computations on real-world road networks.

To address this challenge, we leverage the concepts of the \textit{cut property} and \textit{tree decomposition}, which have demonstrated effectiveness in shortest path computations~\cite{TEDISIGMOD10,LiJunChang2012exact,HopLabeling2018hierarchy}. A widely adopted approach for efficient shortest path queries is to construct distance labelling schemes over graphs. Specifically, the \textit{cut property} states that the shortest path distance between two sets of nodes separated by a vertex cut is determined by the minimum sum of distances from each node to the vertex cut. Leveraging this property, \textit{tree decomposition} has been utilized to partition the graph into disconnected components, ensuring that the shortest path distance $d(s,t)$ can be computed solely based on pre-computed distances stored at the least common ancestor (\lca) of $s$ and $t$ in the tree decomposition. The success of distance labelling techniques has enabled shortest path computations to scale effectively to road networks comprising millions of nodes ~\cite{HopLabeling2018hierarchy}, as such graphs typically exhibit small tree-width. A natural question arises: Can we design an analogous labelling scheme for resistance distance?

In this work, we answer this question affirmatively by developing the first efficient resistance distance labelling scheme. Unlike shortest path distance, resistance distance computation involves complex graph matrix operations, posing significant challenges for designing effective labelling strategies. To address this, we first study the cut property of resistance distance and generalize it from individual nodes to node sets. This extension introduces additional complexity, as it necessitates computing the Schur complement. To mitigate this complexity, we utilize the Cholesky decomposition of the inverse Laplacian matrix to provide a simplified version of the cut property. We demonstrate that for each node in separated node sets, it suffices to store only a single label per node in the \textit{vertex cut}, and simple arithmetic operations can accurately recover resistance distances from these labels. Furthermore, by employing tree decomposition and vertex hierarchies, we establish that the resistance distance $r(s,t)$ depends solely on the ancestors of nodes $s$ and $t$ in the tree decomposition. Although fundamentally different, this property closely resembles the cut property of shortest path distances, thereby enabling the design of compact resistance distance labelling schemes.

Leveraging this insight, we propose a compact resistance distance labelling scheme named \treeindex. We demonstrate that the labelling size is bounded by $O(n \cdot h_{\mathcal{G}})$, where $h_{\mathcal{G}}$ denotes the height of the tree decomposition and empirically behaves as a small constant in many real-world small tree-width graphs (e.g., road networks). To efficiently compute the labelling, we develop a bottom-up construction algorithm that builds the labelling in $O(n \cdot h_{\mathcal{G}}^2 \cdot d_{\max})$ time by performing rank-$1$ updates on the inverse Laplacian matrix following a predefined \dfs ordering. Utilizing this labelling, we propose two efficient query processing algorithms, answering single-pair queries in $O(h_{\mathcal{G}})$ time and single-source queries in $O(n \cdot h_{\mathcal{G}})$ time.

We conduct extensive experiments on $10$ real-world large-scale networks, including the entire US road network \fullusa, comprising $23,947,348$ nodes and $28,854,312$ edges. The experimental results demonstrate that the proposed method, \treeindex, achieves more than $3$ orders of magnitude improvement in query efficiency for single-pair queries compared to state-of-the-art approaches, including approximate solutions that yield results with absolute errors up to $10^{-1}$. For single-source queries, our method remains exact while also being an order of magnitude faster than the best available approximate methods. Moreover, \treeindex maintains acceptable label size and construction time. Notably, labels for \fullusa can be constructed within approximately $7$ hours, resulting in a total label size of $405$~GB. With this index, single-pair queries can be answered in approximately $10^{-3}$ seconds, and single-source queries within $190$ seconds. To the best of our knowledge, this represents the first exact approach capable of computing single-source resistance distances on such a large-scale road network. As a practical demonstration of resistance distance computation on large road networks, we also present a case study on robust routing. Our key contributions are summarized as follows:

\stitle{New Theoretical Findings.} 
We discover two new properties of resistance distance: the \textit{cut property} and the \textit{dependency property}. The \textit{cut property} expresses $r(s,t)$ in terms of relative resistances from nodes $s$ and $t$ to a vertex cut, enabling compact storage and efficient recovery of distance labels. The \textit{dependency property} demonstrates that the resistance distance $r(s,t)$ solely depends on labels along the paths from nodes $s$ and $t$ to the root node within the tree decomposition structure.

\stitle{Novel Indexing Algorithms.}
We propose a novel resistance distance labelling, \treeindex, leveraging tree decomposition. The label size is bounded by $O(n\cdot h_{\mathcal{G}})$. We develop a bottom-up algorithm for label construction in $O(n\cdot h_{\mathcal{G}}^2\cdot d_{max})$ time, as well as two query algorithms: one that processes single-pair queries in $O(h_{\mathcal{G}})$ time, and another for single-source queries in $O(n\cdot h_{\mathcal{G}})$ time.

\stitle{Extensive Experiments.}
We conduct extensive evaluations on $10$ large-scale networks, including \fullusa. Our experimental results show that the proposed \treeindex significantly improves query efficiency while guaranteeing exact accuracy, moderate label sizes, and practical label construction time. To the best of our knowledge, this is the first method capable of computing exact single-source resistance distances on graphs with more than $20$ million nodes. We also demonstrate the practical utility of our approach by successfully applying it to robust routing problems on real-life road networks. The source code of our paper is publicly available at \url{https://github.com/mhliao0516/TreeIndex}.

\section{Preliminaries}
\subsection{Problem Definition}
Given an undirected graph $\mathcal{G} = (\mathcal{V}, \mathcal{E})$ with $n$ nodes and $m$ edges, resistance distance \cite{bollobas1998modern} is a distance metric defined by modeling the graph as an electrical network, where each node represents a junction and each edge a resistor. The resistance distance between nodes $s$ and $t$, denoted as $r(s, t)$, is the voltage drop from $s$ to $t$ when a unit current flows into $s$ and out of $t$. According to Kirchhoff's voltage law, the voltage drops are equivalent along any path from $s$ to $t$. Let $\mathbf{f} \in \mathbb{R}^{|\mathcal{E}|}$ denote the electrical flow on each edge $e = (e_1, e_2)$, where $\mathbf{f}(e) > 0$ if the flow is from $e_1$ to $e_2$, and $\mathbf{f}(e) < 0$ otherwise. Let $\mathcal{P}_{st}$ be an arbitrary path from $s$ to $t$. The resistance distance can be represented as: $r(s,t)=\sum_{e\in\mathcal{P}_{st}}\mathbf{f}(e)$. 

Resistance distance is related to the well-known shortest path distance, which is defined as the number of edges in the shortest path from $s$ to $t$. A spanning tree $T$ of $\mathcal{G}$ is a connected subgraph of $\mathcal{G}$ that includes all nodes in $\mathcal{V}$. Let $\mathbf{f}_T$ denote the indicator vector for the shortest path from $s$ to $t$ on spanning tree $T$, where $\mathbf{f}_T(e) = 1$ if edge $e$ is on the shortest path from $s$ to $t$, and $0$ otherwise. There is a unique path from $s$ to $t$ in $T$, which serves as the shortest path within that tree. Let $\mathcal{T}$ denote the set of all spanning trees of $\mathcal{G}$. The electrical flow can be formulated as \cite{tetali1991random}: $\mathbf{f}=\sum_{T\in\mathcal{T}}\frac{1}{|\mathcal{T}|}\mathbf{f}_T$.
Compared to the shortest path distance, resistance distance accounts for all paths between $s$ and $t$, making it more robust.
\begin{example}
	Given an example graph $\mathcal{G}$ illustrated in Fig.~\ref{fig:example-resistance-distance}(a), Fig.~\ref{fig:example-resistance-distance}(b) shows the electrical flow on $\mathcal{G}$ when a unit current flows into $v_2$ and out of $v_4$. Consider the path $\mathcal{P}_{v_2v_4} = (v_2, v_9, v_8, v_4)$; the resistance distance between $v_2$ and $v_4$ can be computed as $r(v_2, v_4) = \mathbf{f}((v_2, v_9)) + \mathbf{f}((v_9, v_8)) + \mathbf{f}((v_8, v_4)) = 0.59 + 0.36 + 0.66 = 1.61$, while the shortest path distance is $d(v_2, v_4) = 3$. Resistance distance exhibits greater robustness compared to shortest path distance. For instance, upon removal of the edge $(v_8, v_9)$, the shortest path distance between $v_2$ and $v_4$ increases to $d(v_2, v_4) = 4$ (a $33\%$ increase), whereas the resistance distance rises to $r(v_2, v_4) = 1.89$ (a $17\%$ increase).
\end{example}
\begin{figure}[t]
	\vspace*{-0.2cm}
	\begin{center}
		\begin{tabular}[t]{c}
			\subfigure[An example graph $\mathcal{G}$]{
				\includegraphics[width=0.42\columnwidth, height=2.2cm]{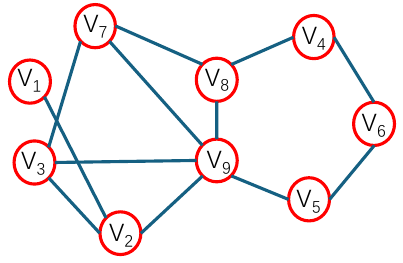}
			}
			\subfigure[Electrical flow on $\mathcal{G}$]{
				\includegraphics[width=0.44\columnwidth, height=2.4cm]{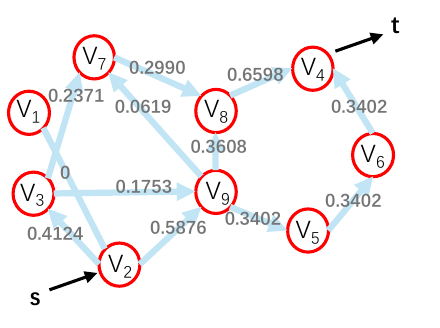}
			}
		\end{tabular}
	\end{center}
	\vspace*{-0.6cm}
	\caption{An illustrative example of resistance distance}
	\vspace*{-0.2cm}
	\label{fig:example-resistance-distance}
\end{figure}
In this paper, we address the problem of exact resistance distance computation, following previous studies \cite{22resistance,23resistance,ResistanceYang,KDDlocal21}, we focus on two types of queries: single-pair and single-source resistance distance queries.
\begin{problem}[Single-pair resistance distance query]
Given a graph $\mathcal{G}=(\mathcal{V},\mathcal{E})$ and a pair of nodes $s, t \in \mathcal{V}$, a single-pair resistance distance query computes the resistance distance $r(s,t)$ between nodes $s$ and $t$.
\end{problem}
\begin{problem}[Single-source resistance distance query]
Given a graph $\mathcal{G}=(\mathcal{V},\mathcal{E})$ and a source node $s \in \mathcal{V}$, a single-source resistance distance query computes the resistance distances from node $s$ to every other node in $\mathcal{V}$.
\end{problem}
Below, we first show that the computation of resistance distance is inherently linked to matrix-based formulations. Then, we review existing methods and discuss their limitations.
\subsection{Resistance Distance Formulations}
According to the definition of resistance distance, it can be expressed using graph-related matrices. Let $\mathbf{A}$ be the adjacency matrix and $\mathbf{D}$ be the degree matrix of graph $\mathcal{G}$; the Laplacian matrix $\mathbf{L}$ is defined as $\mathbf{L} = \mathbf{D} - \mathbf{A}$. Let $\mathbf{x}$ denote the voltage vector at each node when a unit current flows into node $s$ and out of node $t$. The electrical flow on edge $e = (e_1, e_2)$ can be expressed as $\mathbf{f}(e) = \mathbf{x}(e_1) - \mathbf{x}(e_2)$. According to Kirchhoff's voltage law, the voltages at each node satisfy: $\mathbf{L}\mathbf{x} = \mathbf{e}_s - \mathbf{e}_t$, where $\mathbf{e}_s$ is a one-hot vector with a 1 at the index corresponding to $s$ and 0 elsewhere. Since the columns of $\mathbf{L}$ sum to 0, $\mathbf{L}$ has rank $n-1$, so its inverse does not exist. Instead, we use the Moore-Penrose pseudo-inverse. Suppose the eigen-decomposition of $\mathbf{L}$ is $\mathbf{L} = \sum_{i=1}^n \lambda_i \mathbf{u}_i \mathbf{u}_i^T$, where $0 = \lambda_1 \leq \cdots \leq \lambda_n$ are the eigenvalues of $\mathbf{L}$ and $\mathbf{u}_i$ is the corresponding eigenvector for $i = 1$ to $n$. The Moore-Penrose pseudo-inverse of $\mathbf{L}$ is then defined as $\mathbf{L}^\dagger = \sum_{i=2}^n \frac{1}{\lambda_i} \mathbf{u}_i \mathbf{u}_i^T$. Thus, we derive that $\mathbf{x} = \mathbf{L}^\dagger (\mathbf{e}_s - \mathbf{e}_t)$. The resistance distance is therefore:
\begin{equation}
	r(s, t) = \mathbf{x}(s) - \mathbf{x}(t) = (\mathbf{e}_s - \mathbf{e}_t)^T \mathbf{L}^\dagger (\mathbf{e}_s - \mathbf{e}_t).
\end{equation}
Almost all initial methods for computing resistance distance rely on matrix-based definitions. The primary challenge is computing the pseudo-inverse $\mathbf{L}^\dagger$, which requires $O(n^3)$ time for exact calculation. Several formulas have been proposed to avoid computing $\mathbf{L}^\dagger$. For example, \cite{22resistance} focuses on expressing resistance distance via $\mathbf{L}_v^{-1}$, where $\mathbf{L}_v$ is the Laplacian sub-matrix obtained by removing the $v$-th row and column of $\mathbf{L}$, and $v$ is an arbitrary node. Specifically, they provide the exact formula for resistance distance, characterized by:
\begin{equation}
	r(s,v)=\mathbf{e}_s^T\mathbf{L}_v^{-1}\mathbf{e}_s,
\end{equation}
\begin{equation}
	r(s,t)=(\mathbf{e}_s-\mathbf{e}_t)^T\mathbf{L}_v^{-1}(\mathbf{e}_s-\mathbf{e}_t),\quad s,t\neq v,
\end{equation}
Then, \cite{23resistance} proposes a formula that extends the concept from a single node $v$ to a node set $\mathcal{V}_i$. Suppose that $\mathcal{U}_i$ and $\mathcal{V}_i$ form a partition of $\mathcal{V}$ such that $\mathcal{V} = \mathcal{U}_i \cup \mathcal{V}_i$. Then, $\mathbf{L}_{\mathcal{U}_i\mathcal{U}_i}$ is the matrix obtained by removing the rows and columns indexed by $\mathcal{V}_i$ from $\mathbf{L}$. The Schur complement $\mathbf{L}/\mathcal{V}_i$ is defined as:
\begin{equation}
	\mathbf{L}/\mathcal{V}_i = \mathbf{L}_{\mathcal{V}_i\mathcal{V}_i} - \mathbf{L}_{\mathcal{V}_i\mathcal{U}_i} \mathbf{L}_{\mathcal{U}_i\mathcal{U}_i}^{-1} \mathbf{L}_{\mathcal{U}_i\mathcal{V}_i}.
\end{equation}
The resistance distance can be computed using the Schur complement as:
\begin{theorem}\label{theorem:formula-prob}\cite{23resistance}
              Let $\mathbf{p}_{u}$ be the $u$-th row of the matrix $\mathbf{L}_{\mathcal{U}_i\mathcal{U}_i}^{-1}\mathbf{L}_{\mathcal{U}_i\mathcal{V}_i}$ for $u\in\mathcal{U}$. The resistance distance can be formulated as: 
  \begin{enumerate}
\item  For $u_1,u_2\in\mathcal{U}_i$, we have
    \begin{equation}\label{equation:r-p-uu}
      \begin{split}
          r(u_1,u_2)&=(\mathbf{e}_{u_1}-\mathbf{e}_{u_2})^T(\mathbf{L}_{\mathcal{U}_i\mathcal{U}_i}^{-1})(\mathbf{e}_{u_1}-\mathbf{e}_{u_2})\\
          &\quad +(\mathbf{p}_{u_1}-\mathbf{p}_{u_2})^T(\mathbf{L}/\mathcal{V}_i)^\dagger(\mathbf{p}_{u_1}-\mathbf{p}_{u_2});
      \end{split}
  \end{equation}

  \item For $u\in\mathcal{U}_i, v\in\mathcal{V}_i$, we have
        \begin{equation}\label{equation:r-p-uv}
      \begin{split}
          r(u,v)&=\mathbf{e}_{u}^T\mathbf{L}_{\mathcal{U}_i\mathcal{U}_i}^{-1}\mathbf{e}_{u}+(\mathbf{p}_{u}-\mathbf{e}_{v})^T(\mathbf{L}/\mathcal{V}_i)^\dagger(\mathbf{p}_{u}-\mathbf{e}_{v});
      \end{split}
  \end{equation}

  \item For $v_1,v_2\in\mathcal{V}_i$, we have
      \begin{equation}\label{equation:r-p-vv}
      \begin{split}
          r(v_1,v_2)&=(\mathbf{e}_{v_1}-\mathbf{e}_{v_2})^T(\mathbf{L}/\mathcal{V}_i)^\dagger(\mathbf{e}_{v_1}-\mathbf{e}_{v_2}).
      \end{split}
  \end{equation}
\end{enumerate}
\end{theorem}
\begin{figure}[t]
	\centering
	\begin{minipage}{0.48\linewidth}
	  \centering
	  \scalebox{0.58}{$\begin{bmatrix}
	   1 & -1 & 0 & 0 & 0 & 0 & 0 & 0 & 0 \\
	  -1 & 3 & -1 & 0 & 0 & 0 & 0 & 0 & -1 \\
	   0 & -1 & 3 & 0 & 0 & 0 & -1 & 0 & -1 \\
	   0 & 0 & 0 & 2 & -1 & -1 & 0 & 0 & 0 \\
	   0 & 0 & 0 & -1 & 2 & -1 & 0 & 0 & -1 \\
	   0 & 0 & 0 & -1 & -1 & 2 & 0 & 0 & 0 \\
	   0 & 0 & -1 & 0 & 0 & 0 & 3 & -1 & -1 \\
	   0 & 0 & 0 & -1 & 0 & 0 & -1 & 3 & -1 \\
	   0 & -1 & -1 & 0 & -1 & 0 & -1 & -1 & 5
	  \end{bmatrix}$}
	  \vspace{-1.0em}
	  \caption*{(a) Laplacian matrix $\mathbf{L}$}
	\end{minipage}
	\hspace{0.0\linewidth}
	\begin{minipage}{0.48\linewidth}
	  \centering
	  \scalebox{0.62}{$\begin{bmatrix}
	  1.62 & 0.62 & 0.24 & 0.03 & 0.01 & 0.02 & 0.09 & 0.04 \\
	  0.62 & 0.62 & 0.24 & 0.03 & 0.01 & 0.02 & 0.09 & 0.04 \\
	  0.24 & 0.24 & 0.47 & 0.06 & 0.02 & 0.04 & 0.19 & 0.08 \\
	  0.03 & 0.03 & 0.06 & 1.05 & 0.35 & 0.70 & 0.15 & 0.40 \\
	  0.01 & 0.01 & 0.02 & 0.35 & 0.78 & 0.57 & 0.05 & 0.13 \\
	  0.02 & 0.02 & 0.04 & 0.70 & 0.57 & 1.13 & 0.10 & 0.27 \\
	  0.09 & 0.09 & 0.19 & 0.15 & 0.05 & 0.10 & 0.46 & 0.21 \\
	  0.04 & 0.04 & 0.08 & 0.40 & 0.13 & 0.27 & 0.21 & 0.54
	  \end{bmatrix}$}
	  \vspace{-1.0em}
	  \caption*{(b) $\mathbf{L}_v^{-1}$}
	\end{minipage}
	
	\vspace{0.8em}
	\begin{minipage}{0.33\linewidth}
	  \centering
	  \scalebox{0.6}{$\begin{bmatrix}
	  1.60 & 0.60 & 0.20 & 0 & 0 & 0 \\
	  0.60 & 0.60 & 0.20 & 0 & 0 & 0 \\
	  0.20 & 0.20 & 0.40 & 0 & 0 & 0 \\
	  0 & 0 & 0 & 0.75 & 0.25 & 0.50 \\
	  0 & 0 & 0 & 0.25 & 0.75 & 0.50 \\
	  0 & 0 & 0 & 0.50 & 0.50 & 1.00
	  \end{bmatrix}$}
	  \vspace{-1.0em}
	  \caption*{(c) $\mathbf{L}_{\mathcal{U}_i\mathcal{U}_i}^{-1}$}
	\end{minipage}
	\hfill
	\begin{minipage}{0.31\linewidth}
	  \centering
	  \scalebox{0.65}{$\begin{bmatrix}
	  0.20 & 0 & 0.80 \\
	  0.20 & 0 & 0.80 \\
	  0.40 & 0 & 0.60 \\
	  0 & 0.75 & 0.25 \\
	  0 & 0.25 & 0.75 \\
	  0 & 0.50 & 0.50
	  \end{bmatrix}$}
	  \vspace{-1.0em}
	  \caption*{(d) $\mathbf{L}_{\mathcal{U}_i\mathcal{U}_i}^{-1}\mathbf{L}_{\mathcal{U}_i\mathcal{V}_i}$}
	\end{minipage}
	\hspace{-0.02\linewidth}
	\begin{minipage}{0.31\linewidth}
	  \centering
	  \scalebox{0.75}{$\begin{bmatrix}
	   0.17 & -0.11 & -0.01 \\
	  -0.11 &  0.20 & -0.09 \\
	  -0.01 & -0.09 &  0.16
	  \end{bmatrix}$}
	  \vspace{-1.0em}
	  \caption*{(e) $(\mathbf{L}/\mathcal{V}_i)^{\dagger}$}
	\end{minipage}
	
	\caption{An illustrative example of graph-related matrices of $\mathcal{G}$. (a) Laplacian matrix $\mathbf{L}$; (b) $\mathbf{L}_v^{-1}$, $v$ is selected as $v_9$; (c) $\mathbf{L}_{\mathcal{U}_i\mathcal{U}_i}^{-1}$, $\mathcal{U}_i=\{v_1,v_2,v_3,v_4,v_5,v_6\}$; (d) $\mathbf{L}_{\mathcal{U}_i\mathcal{U}_i}^{-1}\mathbf{L}_{\mathcal{U}_i\mathcal{V}_i}$; (e) $(\mathbf{L}/\mathcal{V}_i)^\dagger$.}
	\label{fig:example-graph-matrix}
\end{figure}
\begin{example}
	Fig.~\ref{fig:example-graph-matrix} illustrates several graph matrices that can be used to represent resistance distance. Fig.~\ref{fig:example-graph-matrix}(a) shows the Laplacian matrix $\mathbf{L}$ of the graph $\mathcal{G}$ from Fig.~\ref{fig:example-resistance-distance}(a). Fig.~\ref{fig:example-graph-matrix}(b) displays the matrix $\mathbf{L}_v^{-1}$, where $v$ is set to $v_9$. The resistance distance $r(v_1, v_9)$ is the element at the $v_1$-th row and column of $\mathbf{L}_v^{-1}$, which equals 1.62. Similarly, $r(v_2, v_4)$ can be computed as $r(v_2, v_4) = \mathbf{e}_2^T \mathbf{L}_v^{-1} \mathbf{e}_2 + \mathbf{e}_4^T \mathbf{L}_v^{-1} \mathbf{e}_4 - 2 \mathbf{e}_2^T \mathbf{L}_v^{-1} \mathbf{e}_4 = 1.61$. Fig.~\ref{fig:example-graph-matrix}(c)-(e) display the matrices $\mathbf{L}_{\mathcal{U}_i\mathcal{U}_i}^{-1}$, $\mathbf{L}_{\mathcal{U}_i\mathcal{U}_i}^{-1} \mathbf{L}_{\mathcal{U}_i\mathcal{V}_i}$, and $(\mathbf{L}/\mathcal{V}_i)^\dagger$, where $\mathcal{U}_i = \{v_1, v_2, v_3, v_4, v_5, v_6\}$. The resistance distance between $v_2$ and $v_4$ can be computed as $r(v_2, v_4) = (\mathbf{e}_2 - \mathbf{e}_4)^T \mathbf{L}_{\mathcal{U}_i\mathcal{U}_i}^{-1} (\mathbf{e}_2 - \mathbf{e}_4) + (\mathbf{p}_2 - \mathbf{p}_4)^T (\mathbf{L}/\mathcal{V}_i)^\dagger (\mathbf{p}_2 - \mathbf{p}_4) = 1.61$.
\end{example}
\subsection{Existing Solutions and Their Limitations}
The matrix-based formulations give flexible ways to express resistance distance. However, explicit use of numerical solvers can not scale to large graphs. Therefore, most existing methods focus on designing graph-based methods.
\stitle{Random Walk-based Approximate Methods.} Recently, a series of methods focuses on sampling random walks to approximate resistance distance \cite{KDDlocal21,ResistanceYang,22resistance,23resistance}. The basic idea is to develop a random walk-based estimator for $r(s,t)$ and use the Monte Carlo method to approximate the expectation. Specifically, \geer \cite{ResistanceYang} focuses on designing random walk algorithms to approximate $r(s,t)$, enabling guaranteed approximation results without accessing the entire graph. \bipush \cite{22resistance} utilizes variance-reduced random walk sampling to approximate elements of $\mathbf{L}_v^{-1}$. It heuristically selects $v$ as an easy-to-hit node, allowing the random walk to terminate quickly. On some graphs, finding a single suitable node $v$ is challenging. \leindex \cite{23resistance} extends the approach to express resistance distance via $\mathbf{L}_{\mathcal{U}_i\mathcal{U}_i}^{-1}$. \leindex is an index-based method that employs random walk and random spanning forest sampling to approximate $(\mathbf{L}/\mathcal{V}_i)^\dagger$ and $\mathbf{L}_{\mathcal{U}_i\mathcal{U}_i}^{-1} \mathbf{L}_{\mathcal{U}_i\mathcal{V}_i}$. It then stores the $|\mathcal{V}_i| \times |\mathcal{V}_i|$ and $|\mathcal{U}_i| \times |\mathcal{V}_i|$ matrices as an index. For queries, it computes resistance distance by only calculating elements of $\mathbf{L}_{\mathcal{U}_i\mathcal{U}_i}^{-1}$. These random walk-based methods scale resistance distance computation to large-scale networks. However, they are limited to approximate solutions and are sensitive to the spectral properties of the graph, often resulting in slow query times and large estimation errors on small tree-width graphs. 

\stitle{Laplacian Solver-based Exact Methods.} Another line of methods directly employs a Laplacian solver to compute resistance distance. Similar to random walk-based methods, exact numerical methods also suffer on small tree-width graphs due to their large condition numbers \cite{CoreTree14}. The basic idea of a Laplacian solver is to first construct a preconditioner to reduce the condition number and then apply traditional iterative methods like conjugate gradient. In theory, Laplacian solvers have achieved a near-linear complexity of $\widetilde{O}(m)$ \cite{GraphSparsificationEff08}. Several attempts have been made to make Laplacian solvers practical \cite{LaplacianSolver,RobustandPracticalLaplacian23,RCHOL21}. However, the $\widetilde{O}(m)$ complexity, along with the hidden large constant factor, still limits the query efficiency of resistance distance on large-scale networks.
\comment{
\stitle{Tree decomposition-based shortest path distance labelling.} Large road networks are typically sparse and have small tree widths, making them easily decomposable into disconnected pieces by removing a small number of nodes. Tree decomposition, a technique that breaks a graph into such disconnected components, is one of the most advanced methods for shortest path distance computation \cite{TEDISIGMOD10,LiJunChang2012exact,ProjectedVertexSeparator2021,HopLabeling2018hierarchy}. A representative method is \htwoh \cite{HopLabeling2018hierarchy}, which leverages the cut property of shortest path distances and tree decomposition to create a distance labelling over the graph. For a query $d(s,t)$, the distance can be quickly recovered by accessing the labels of the ancestors of $s$ and $t$ in the tree decomposition structure. \htwoh efficiently processes shortest path distance queries on large road networks. However, it is designed for shortest path computation, whereas resistance distance computation presents a fundamentally different problem.}
\subsection{Challenges of applying tree decomposition}
For the problem of shortest path distance computation, tree decomposition has been successfully applied to obtain superior performance on small treewidth graphs \cite{TEDISIGMOD10,LiJunChang2012exact,HopLabeling2018hierarchy}. Given the limitations of existing resistance distance computation approaches for graphs with small treewidth and the success of tree decomposition in shortest path distance computation, in this paper, we focus on applying tree decomposition to resistance distance computation. However, extending the tree decomposition-based method to resistance distance computation poses several primary challenges:

(1) $d(s,t)$ is only related to the shortest path between $s$ and $t$, while $r(s,t)$ is related to all paths between $s$ and $t$. Thus, the cut property of shortest path distance only needs a simple minimum operation. However, resistance distance computation indeed solves a Laplacian linear system, which requires a lot of complex matrix operations. Immediate matrix operation results should be stored as labels. It is non-trival to design such a new cut property, which requires an in-depth understanding of resistance distance. 

(2) To integrate the tree decomposition and vertex hierarchy property, given the proposed cut properties, the challenge is to recognize the relationship between the non-zero structure of matrix decomposition and the structure of tree decomposition. A closed form matrix-based formula of $r(s,t)$ in terms of the tree decomposition must be provided. It is also non-trivial to ensure that such a formula relates to only a small part of the tree decomposition.

(3) Compared with the labels of shortest path distance laeblling structure representing distance values, the labels of resistance distance laeblling scheme are immediate matrix computation results, which are hard to compute. Thus, it is challenging to construct the labels exactly, by leveraging the non-zero structure corresponding to the tree decomposition.

In the subsequent sections, we address the first challenge by expressing and simplifying the cut property of resistance distance using the operations of vector outer products (Section~\ref{subsec:cut-property}). For the second challenge, we employ Cholesky decomposition to the inverse Laplacian matrix, ensuring that the non-zero structure of the labelling corresponds exactly to the tree decomposition (Section~\ref{subsec:tree-decomposition-vertex-hierarchy}). For the third challenge, we compute the labels using incremental rank-1 updates, facilitating a bottom-up construction of the resistance distance labelling. By integrating these techniques, we develop a resistance distance labelling scheme with a time complexity comparable to that of tree decomposition-based shortest path distance labelling (Section~\ref{sec:resistance-distance-labelling}), marking a substantial advancement over the previous state-of-the-art index-based method \leindex \cite{23resistance}.

\comment{However, resistance distance poses two primary challenges compared to shortest path distance: (i) The cut property of resistance distance is not well-defined, and direct generalization necessitates computing the Schur complement, which is significantly less efficient than the minimization operation inherent in the cut property of shortest path distance; (ii) Unlike shortest path distance labelling, which comprises simple distance values, each label in the proposed resistance distance labelling represents an element of the inverse of a Laplacian submatrix, requiring the exact solution of a Laplacian linear system, a computationally intensive task. In the subsequent sections, we address the first challenge by employing Cholesky decomposition of the inverse Laplacian matrix to present the cut property of resistance distance, thereby enabling compact storage and efficient recovery (Section~\ref{sec:resistance-distance-property}). For the second challenge, we compute the labels using incremental rank-1 updates based on a node ordering, facilitating a bottom-up construction of the resistance distance labelling. By integrating these techniques, we develop a resistance distance labelling scheme with a time complexity comparable to that of tree decomposition-based shortest path distance labelling (Section~\ref{sec:resistance-distance-labelling}), marking a substantial advancement over the previous state-of-the-art index-based method \leindex \cite{23resistance}.}

\section{Resistance distance Properties}\label{sec:resistance-distance-property}
In this section, we first establish the cut property of resistance distance in Section~\ref{subsec:cut-property}. Then, we introduce the concept of tree decomposition and vertex hierarchy in Section~\ref{subsec:tree-decomposition-vertex-hierarchy}, showing that a resistance distance labelling can be built such that resistance distance $r(s,t)$ only relies on the labels of the ancestors of nodes $s$ and $t$ in the tree decomposition.
\subsection{Cut Property of Resistance Distance}\label{subsec:cut-property}
Given a graph $\mathcal{G} = (\mathcal{V},\mathcal{E})$, a vertex set $\mathcal{V}_{cut} \subset \mathcal{V}$ is called a \textit{vertex cut} if its removal from $\mathcal{G}$ results in multiple connected components. Suppose that $s \in \mathcal{V}_1$ and $t \in \mathcal{V}_2$, where $\mathcal{V}_1$ and $\mathcal{V}_2$ are two disconnected node sets obtained by deleting $\mathcal{V}_{cut}$. The cut property for the shortest path distance \cite{HopLabeling2018hierarchy} implies that
\begin{equation}
    d(s,t) = \min_{v \in \mathcal{V}_{cut}} \Big[ d(s,v) + d(v,t) \Big].
\end{equation}
According to this property, when the vertex cut contains only a few nodes, existing distance labelling methods \cite{LiJunChang2012exact,TEDISIGMOD10,HopLabeling2018hierarchy,ProjectedVertexSeparator2021} need only store the distances between nodes in $\mathcal{V}_1$ (or $\mathcal{V}_2$) and each $v \in \mathcal{V}_{cut}$. Consequently, a query for $d(s,t)$ can be quickly resolved by taking the minimum over all $v \in \mathcal{V}_{cut}$. 

\stitle{Warm up.} For resistance distance, we observe that the cut property holds as well if the vertex cut consists of only a single node $v$.
\begin{lemma}\label{lemma:cut-property-single}
  Let $s,t,v\in\mathcal{V}$. If $v$ is a cut vertex that the removal of $v$ splits $s$ and $t$ into different connected components. The resistance distance $r(s,t)$ satisfies: $r(s,t)=r(s,v)+r(v,t)$.
  \begin{proof}
    According to the resistance computation formula, we have:
    \begin{equation*}
      \begin{split}
        r(s,t)&=(\mathbf{e}_s-\mathbf{e}_t)^T\mathbf{L}_v^{-1}(\mathbf{e}_s-\mathbf{e}_t)\\
        &=\mathbf{e}_s^T\mathbf{L}_v^{-1}\mathbf{e}_s+\mathbf{e}_t^T\mathbf{L}_v^{-1}\mathbf{e}_t-2\mathbf{e}_s^T\mathbf{L}_v^{-1}\mathbf{e}_t\\
        &=\mathbf{e}_s^T\mathbf{L}_v^{-1}\mathbf{e}_s+\mathbf{e}_t^T\mathbf{L}_v^{-1}\mathbf{e}_t\\
        &=r(s,v)+r(v,t).
      \end{split}
    \end{equation*}
    Here, the third equality holds because $\mathbf{e}_s^T\mathbf{L}_v^{-1}\mathbf{e}_t=0$. This is because $\mathbf{e}_s^T\mathbf{L}_v^{-1}\mathbf{e}_t=\frac{\tau_v[s,t]}{d_t}$ \cite{22resistance}, where $\tau_v[s,t]$ is the expected number of passes to $t$ in a random walk starts from $s$ and terminates when it hits $v$. Since $v$ is a vertex cut of $s$ and $t$, it is impossible for a random walk from $s$ to pass $t$ before it hits $v$. Thus, $\tau_v[s,t]=0$. The Lemma is established.
  \end{proof}
\end{lemma}
However, when generalizing the cut property to a cut set $\mathcal{V}_{cut}$, the resistance distance can no longer be expressed solely in terms of the resistance distances from the vertex cut. As an alternative, inspired by the formulas presented in Theorem~\ref{theorem:formula-prob}, we generalize the cut property of resistance distance to a \textit{vertex cut} by introducing the concept of a \textit{contraction graph}. Formally, we have:
\begin{definition}[Contraction graph]
	Given a graph $\mathcal{G}$ and a node set $\mathcal{V}_1 \subset \mathcal{V}$, a \textit{contraction graph} $\mathcal{G}_{\mathcal{V}_1}$ is defined as the graph obtained by contracting $\mathcal{V}_1$ into a single node, such that all edges from nodes in $\mathcal{V}_1$ to nodes outside $\mathcal{V}_1$ are redirected to this new node.
\end{definition}
\begin{example}
	Fig.~\ref{fig:example-cut-property} illustrates examples of contraction graphs. Given the graph $\mathcal{G}$ in Fig.~\ref{fig:example-cut-property}(a), Fig.~\ref{fig:example-cut-property}(b) shows the contraction graph $\mathcal{G}_{\mathcal{V}_1}$ for $\mathcal{V}_1=\{v_1,v_2,v_3\}$, obtained by contracting $\mathcal{V}\setminus\mathcal{V}_1$ into a single node $\Delta_1$. Fig.~\ref{fig:example-cut-property}(c) shows the contraction graph $\mathcal{G}_{\mathcal{V}_2}$ for $\mathcal{V}_2=\{v_7,v_8,v_9\}$, obtained by contracting $\mathcal{V}\setminus\mathcal{V}_2$ into a single node $\Delta_2$.
\end{example}
Given a graph a vertex cut $\mathcal{V}_{cut}$, resistance distance between two nodes $s \in \mathcal{V}_1$ and $t \in \mathcal{V}_2$ can be expressed in terms of $r_{\mathcal{G}_{\mathcal{V}_1}}(s, \Delta_1)$, which denotes the resistance distance between $s$ and $\Delta_1$ in $\mathcal{G}_{\mathcal{V}_1}$, and $r_{\mathcal{G}_{\mathcal{V}_2}}(t, \Delta_2)$, which denotes the resistance distance between $t$ and $\Delta_2$ in $\mathcal{G}_{\mathcal{V}_2}$.
\begin{lemma}\label{lemma:cut-property} Let $\mathcal{U}=\mathcal{V}\setminus\mathcal{V}_{cut}$, $\mathbf{p}_s$ and $\mathbf{p}_t$ be the $s$-th and the $t$-th row of the matrix $-\mathbf{L}_{\mathcal{U}\mathcal{U}}^{-1}\mathbf{L}_{\mathcal{U}\mathcal{V}_{cut}}$. Then, we have: $\small r(s,t) = r_{\mathcal{G}_{\mathcal{V}_1}}(s,\Delta_1)+r_{\mathcal{G}_{\mathcal{V}_2}}(t,\Delta_2)+(\mathbf{p}_{s}-\mathbf{p}_{t})^T(\mathbf{L}/\mathcal{V}_{cut})^\dagger(\mathbf{p}_{s}-\mathbf{p}_{t})$.
\end{lemma}
\begin{proof}
    According to the resistance distance formula in Theorem~\ref{theorem:formula-prob}. We have:
    \begin{equation*}
      \begin{split}
        \small r(s,t) &= \mathbf{e}_{s}^T\mathbf{L}_{\mathcal{U}\mathcal{U}}^{-1}\mathbf{e}_{s}+\mathbf{e}_{t}^T\mathbf{L}_{\mathcal{U}\mathcal{U}}^{-1}\mathbf{e}_{t}-2\mathbf{e}_{s}^T\mathbf{L}_{\mathcal{U}\mathcal{U}}^{-1}\mathbf{e}_{t}\\
        &\quad+(\mathbf{p}_{s}-\mathbf{p}_{t})^T(\mathbf{L}/\mathcal{V}_{cut})^\dagger(\mathbf{p}_{s}-\mathbf{p}_{t})\\
        &= \mathbf{e}_{s}^T\mathbf{L}_{\mathcal{U}\mathcal{U}}^{-1}\mathbf{e}_{s}+\mathbf{e}_{t}^T\mathbf{L}_{\mathcal{U}\mathcal{U}}^{-1}\mathbf{e}_{t}\\
        &\quad+(\mathbf{p}_{s}-\mathbf{p}_{t})^T(\mathbf{L}/\mathcal{V}_{cut})^\dagger(\mathbf{p}_{s}-\mathbf{p}_{t})\\
        &= \mathbf{e}_{s}^T\mathbf{L}_{\mathcal{V}_1\mathcal{V}_1}^{-1}\mathbf{e}_{s}+\mathbf{e}_{t}^T\mathbf{L}_{\mathcal{V}_2\mathcal{V}_2}^{-1}\mathbf{e}_{t}\\
        &\quad+(\mathbf{p}_{s}-\mathbf{p}_{t})^T(\mathbf{L}/\mathcal{V}_{cut})^\dagger(\mathbf{p}_{s}-\mathbf{p}_{t})\\
        &=r_{\mathcal{G}_{\mathcal{V}_1}}(s,\Delta_1)+r_{\mathcal{G}_{\mathcal{V}_2}}(t,\Delta_2)\\
        &\quad+(\mathbf{p}_{s}-\mathbf{p}_{t})^T(\mathbf{L}/\mathcal{V}_{cut})^\dagger(\mathbf{p}_{s}-\mathbf{p}_{t}).
      \end{split}
  \end{equation*}
  Similar to the proof of Lemma~\ref{lemma:cut-property-single}, we can obtain that $\mathbf{e}_{s}^T\mathbf{L}_{\mathcal{U}\mathcal{U}}^{-1}\mathbf{e}_{t}=0$ for $s$ and $t$ in different connected components. $\mathbf{L}_{\mathcal{U}\mathcal{U}}^{-1}$ has a structure: $\small\left[\begin{matrix}
    \mathbf{L}_{\mathcal{V}_1\mathcal{V}_1}^{-1} & & \\
    & \mathbf{L}_{\mathcal{V}_2\mathcal{V}_2}^{-1} & \\
    & & \cdots
  \end{matrix}\right]$. Thus, the third and the fourth equality holds. The Lemma is established.
\end{proof}
\begin{figure}[t]
	\vspace*{-0.2cm}
	\begin{center}
		\begin{tabular}[t]{c}
			\subfigure[A vertex cut $\{v_7,v_8,v_9\}$]{
				\includegraphics[width=0.3\columnwidth, height=1.8cm]{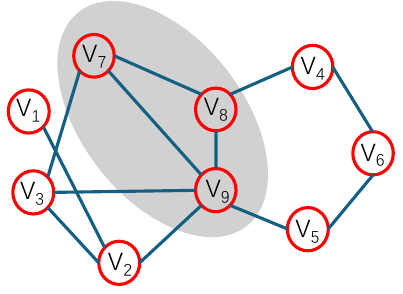}
			}
			\subfigure[$\mathcal{G}_{\mathcal{V}_1}$]{
				\raisebox{0.4cm}{\includegraphics[width=0.3\columnwidth, height=1.1cm]{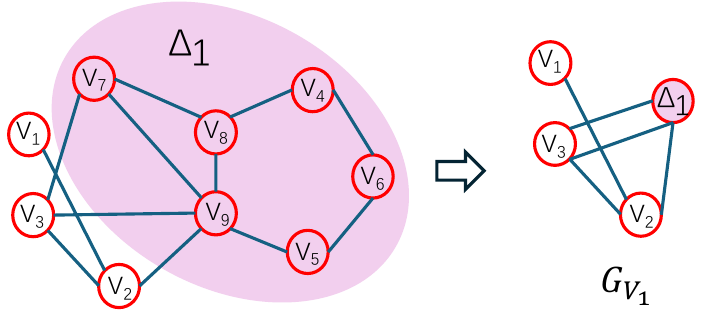}}
			}
      \subfigure[$\mathcal{G}_{\mathcal{V}_2}$]{
				\raisebox{0.4cm}{\includegraphics[width=0.3\columnwidth, height=1.1cm]{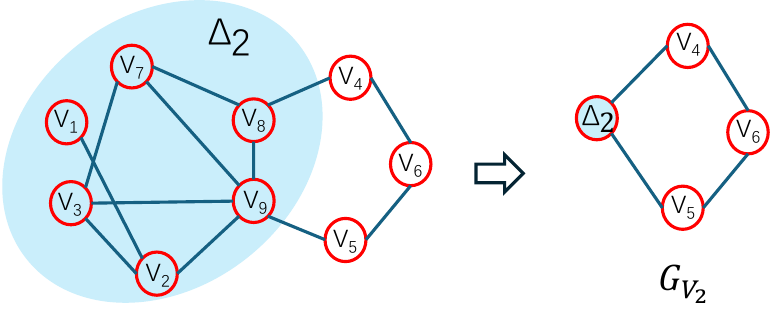}}
			}
		\end{tabular}
	\end{center}
	\vspace*{-0.6cm}
	\caption{An illustrative example of the cut property of resistance distance}
	\vspace*{-0.2cm}
	\label{fig:example-cut-property}
\end{figure}
\begin{example}\label{example:cut-property}
	An example illustrating the cut property of resistance distance is shown in Fig.~\ref{fig:example-cut-property}. Given the graph $\mathcal{G}$ in Fig.~\ref{fig:example-resistance-distance}(a), Fig.~\ref{fig:example-cut-property}(a) depicts a \textit{vertex cut} $\mathcal{V}_{cut} = \{v_7, v_8, v_9\}$ that separates $\mathcal{G}$ into two connected components. According to the cut property of shortest path distance, $d(v_2, v_4) = \min_{v \in \mathcal{V}_{cut}} d(v_2, v) + d(v, v_4) = 3$. Similarly, the cut property of resistance distance states that $r(v_2, v_4) = r(v_2, \Delta_1) + r(v_4, \Delta_2) + (\mathbf{p}_{v_2} - \mathbf{p}_{v_4})^T (\mathbf{L}/\mathcal{V}_{cut})^\dagger (\mathbf{p}_{v_2} - \mathbf{p}_{v_4}) = 1.61$.
\end{example}
According to Lemma~\ref{lemma:cut-property}, for a cut set $\mathcal{V}_{cut}$, the computation of the resistance distance $r(s, t)$ can be reduced to $r(s, \Delta_1)$ and $r(t, \Delta_2)$, both of which can be computed independently. However, Lemma~\ref{lemma:cut-property} also introduces additional complexity in computing $-\mathbf{L}_{\mathcal{U}\mathcal{U}}^{-1}\mathbf{L}_{\mathcal{U}\mathcal{V}_{cut}}$ and $(\mathbf{L}/\mathcal{V}_{cut})^\dagger$. The computation of the Schur complement and its pseudo-inverse is computationally expensive. To circumvent this, we propose a novel method to represent the cut property of resistance distance in terms of the \textit{Cholesky decomposition}. 

\stitle{Simplification by Cholesky decomposition.} According to the matrix-based formulations of resistance distance, $\mathbf{L}_v^{-1}$ encodes the resistance distances in $\mathcal{G}$, while $\mathbf{L}_{\mathcal{V}_1\mathcal{V}_1}^{-1}$ and $\mathbf{L}_{\mathcal{V}_2\mathcal{V}_2}^{-1}$ encode the resistance distances in $\mathcal{G}_{\mathcal{V}_1}$ and $\mathcal{G}_{\mathcal{V}_2}$, respectively. This leads to the following observation:
\begin{lemma}\label{lemma:matrix-difference} Let $\mathcal{U}=\mathcal{V}\setminus\mathcal{V}_{cut}$, we have:
  \begin{equation*}
    r(s,t)-r(s,\Delta_1)-r(t,\Delta_2)=(\mathbf{e}_{s}-\mathbf{e}_{t})^T(\mathbf{L}_v^{-1}-\left[\begin{matrix}
		\mathbf{L}_{\mathcal{U}\mathcal{U}}^{-1} & 0 \\
		0 & 0
	  \end{matrix}\right])(\mathbf{e}_{s}-\mathbf{e}_{t}).
  \end{equation*}
  \begin{proof}
    According to the definition of resistance distance, we can obtain that $r(s,t)=(\mathbf{e}_{s}-\mathbf{e}_{t})^T\mathbf{L}_v^{-1}(\mathbf{e}_{s}-\mathbf{e}_{t})$, $r(s,\Delta_1)=\mathbf{e}_{s}^T\mathbf{L}_{\mathcal{V}_1\mathcal{V}_1}^{-1}\mathbf{e}_{s}$, $r(t,\Delta_2)=\mathbf{e}_{t}^T\mathbf{L}_{\mathcal{V}_2\mathcal{V}_2}^{-1}\mathbf{e}_{t}$. Moreover, $\mathbf{L}_{\mathcal{U}\mathcal{U}}^{-1}$ has the block structure: $\small\left[\begin{matrix}
		\mathbf{L}_{\mathcal{V}_1\mathcal{V}_1}^{-1} & & \\
		& \mathbf{L}_{\mathcal{V}_2\mathcal{V}_2}^{-1} & \\
		& & \cdots
	  \end{matrix}\right]$, which establishes the lemma.
  \end{proof}
\end{lemma}
Based on Lemma~\ref{lemma:matrix-difference}, the key challenge in establishing the cut property of resistance distance is to represent the difference between $\mathbf{L}_v^{-1}$ and $\left[\begin{matrix}
	\mathbf{L}_{\mathcal{U}\mathcal{U}}^{-1} & 0 \\
	0 & 0
  \end{matrix}\right]$. This motivates us to introduce the concepts of Cholesky decomposition and the Schur complement.

First, we present an important property of the Schur complement of $\mathbf{L}_v^{-1}$. While existing studies~\cite{23resistance,LaplacianSolver} typically consider the Schur complement of $\mathbf{L}$, we instead focus on the Schur complement of $\mathbf{L}_v^{-1}$. We observe that the inverse of any Laplacian submatrix can be expressed as the Schur complement of the inverse of a Laplacian submatrix associated with a larger node set.
\begin{lemma}\label{lemma:schur-complement-inverse}
  Let $\mathcal{U}_1\subset\mathcal{U}_2\subset\mathcal{V}$ be two subsets of the nodes of $\mathcal{G}$, then $\mathbf{L}^{-1}_{\mathcal{U}_1\mathcal{U}_1}$ is the Schur complement of $\mathbf{L}^{-1}_{\mathcal{U}_2\mathcal{U}_2}$ with respect to the node set $\mathcal{U}_1$.
  \begin{proof}
	According to the block matrix decomposition formula, Let $S=A-BD^{-1}C$ be the Schur complement of $D$ with respect to $A$, we have:
	\begin{equation*}
		\small\begin{bmatrix}
			A & B \\
			C & D
		  \end{bmatrix}^{-1}=\begin{bmatrix}
		  S^{-1} & -S^{-1}BD^{-1} \\
		  -D^{-1}CS^{-1} & D^{-1}+D^{-1}CS^{-1}BD^{-1}
		  \end{bmatrix}.
	\end{equation*}
    We can prove this lemma using block matrix decomposition. First, let's partition $\mathbf{L}_{\mathcal{U}_2\mathcal{U}_2}$ according to $\mathcal{U}_1$ and $\mathcal{U}_2\setminus\mathcal{U}_1$:
    \begin{equation*}
      \small\mathbf{L}_{\mathcal{U}_2\mathcal{U}_2} = \begin{bmatrix}
        \mathbf{L}_{\mathcal{U}_1\mathcal{U}_1} & \mathbf{L}_{\mathcal{U}_1(\mathcal{U}_2\setminus\mathcal{U}_1)} \\
        \mathbf{L}_{(\mathcal{U}_2\setminus\mathcal{U}_1)\mathcal{U}_1} & \mathbf{L}_{(\mathcal{U}_2\setminus\mathcal{U}_1)(\mathcal{U}_2\setminus\mathcal{U}_1)}
      \end{bmatrix}.
    \end{equation*}
    
    Using the formula for the inverse of a block matrix, we have:
    \begin{equation*}
      \small\left(\mathbf{L}_{\mathcal{U}_2\mathcal{U}_2}^{-1}\right)^{-1} = \begin{bmatrix}
        A & B \\
        C & D
      \end{bmatrix}^{-1}
	  = \begin{bmatrix}
        \mathbf{L}_{\mathcal{U}_1\mathcal{U}_1} & \mathbf{L}_{\mathcal{U}_1(\mathcal{U}_2\setminus\mathcal{U}_1)} \\
        \mathbf{L}_{(\mathcal{U}_2\setminus\mathcal{U}_1)\mathcal{U}_1} & \mathbf{L}_{(\mathcal{U}_2\setminus\mathcal{U}_1)(\mathcal{U}_2\setminus\mathcal{U}_1)}
      \end{bmatrix}.
    \end{equation*}
    
    The top-left block is precisely the inverse of the Schur complement of $D$ with respect to $A$. Therefore, we have:
	\begin{align*}
		\tiny S=\mathbf{L}_{\mathcal{U}_1\mathcal{U}_1}^{-1} = &\left(\mathbf{L}_{\mathcal{U}_2\mathcal{U}_2}^{-1}\right)_{\mathcal{U}_1\mathcal{U}_1} -\left(\mathbf{L}_{\mathcal{U}_2\mathcal{U}_2}^{-1}\right)_{\mathcal{U}_1(\mathcal{U}_2\setminus\mathcal{U}_1)}\\
		&\tiny\left(\mathbf{L}_{\mathcal{U}_2\mathcal{U}_2}^{-1}\right)_{(\mathcal{U}_2\setminus\mathcal{U}_1)(\mathcal{U}_2\setminus\mathcal{U}_1)}^{-1}\left(\mathbf{L}_{\mathcal{U}_2\mathcal{U}_2}^{-1}\right)_{(\mathcal{U}_2\setminus\mathcal{U}_1)\mathcal{U}_1}.
	  \end{align*}
This is exactly the Schur complement of $\mathbf{L}_{\mathcal{U}_2\mathcal{U}_2}^{-1}$ with respect to the node set $\mathcal{U}_1$, which completes the proof.
  \end{proof}
\end{lemma}
Motivated by Lemma~\ref{lemma:schur-complement-inverse}, the problem of computing the matrix difference can be reduced to computing the Schur complement of $\mathbf{L}_{\mathcal{U}_2\mathcal{U}_2}^{-1}$. In numerical linear algebra, Gaussian elimination \cite{golub2013matrix} serves as a standard approach for computing such Schur complements. This method systematically transforms the matrix through a sequence of elementary row operations to eliminate specific elements. To implement this approach, we initialize the Schur complement matrix as $\widetilde{\mathcal{S}}_0=\mathbf{L}^{-1}_{\mathcal{U}_2\mathcal{U}_2}$. Without loss of generality, we establish an ordering $\{v_1,v_2,\ldots,v_{|\mathcal{U}_2\setminus\mathcal{U}_1|}\}$ for the nodes in $\mathcal{U}_2\setminus\mathcal{U}_1$. The Gaussian elimination algorithm then iteratively applies the following update procedure for each step $i\geq1$:

\begin{equation}
	\small\widetilde{\mathcal{S}}_i=\widetilde{\mathcal{S}}_{i-1}-\frac{\widetilde{\mathcal{S}}_{i-1}[:,v_i]\widetilde{\mathcal{S}}_{i-1}[:,v_i]^T}{\widetilde{\mathcal{S}}_{i-1}[v_i,v_i]},
\end{equation}
where $\widetilde{\mathcal{S}}_{i-1}[:,v_i]$ denotes the $v_i$-th column of $\widetilde{\mathcal{S}}_{i-1}$. While we defined a specific ordering above, it is worth noting that the elimination sequence can be arbitrary without affecting the final result. After completing the elimination of all nodes in $\mathcal{U}_2\setminus\mathcal{U}_1$, we obtain:
\begin{lemma}
	$\widetilde{\mathcal{S}}_{|\mathcal{U}_2\setminus\mathcal{U}_1|}=\left[\begin{matrix}
	\mathbf{L}_{\mathcal{U}_1\mathcal{U}_1}^{-1} & 0 \\ 0 & 0
	\end{matrix}\right]$.
	\begin{proof}
	  We will prove this by mathematical induction on the number of eliminated nodes.
	  
	  Base case: When no nodes have been eliminated ($i=0$), we have $\widetilde{\mathcal{S}}_0 = \mathbf{L}^{-1}_{\mathcal{U}_2\mathcal{U}_2}$.
	  
	  Inductive hypothesis: Assume that after eliminating $k$ nodes, the resulting matrix $\widetilde{\mathcal{S}}_k$ has the form where all rows and columns corresponding to the eliminated nodes are zero, and the submatrix corresponding to the remaining nodes correctly represents their Schur complement.
	  
	  Inductive step: Consider the elimination of node $v_{k+1}$. Let's partition the matrix $\widetilde{\mathcal{S}}_k$ as $\small\widetilde{\mathcal{S}}_k = \begin{bmatrix}
		A & b \\
		b^T & c
		\end{bmatrix}$, where $c = \widetilde{\mathcal{S}}_k[v_{k+1},v_{k+1}]$ is a scalar, $b = \widetilde{\mathcal{S}}_k[:,v_{k+1}]$ excluding the element $c$, and $A$ is the remaining submatrix.
	  
	  The elimination step gives:
	  \begin{equation*}
	  \small\widetilde{\mathcal{S}}_{k+1} = \widetilde{\mathcal{S}}_k - \frac{\widetilde{\mathcal{S}}_k[:,v_{k+1}]\widetilde{\mathcal{S}}_k[:,v_{k+1}]^T}{\widetilde{\mathcal{S}}_k[v_{k+1},v_{k+1}]} = \begin{bmatrix}
	  A - \frac{bb^T}{c} & 0 \\
	  0 & 0
	  \end{bmatrix}.
	  \end{equation*}
	  This is precisely the Schur complement operation. According to the block matrix inversion formula, if we have a matrix $M = \begin{bmatrix} A & B \\ C & D \end{bmatrix}$ and its inverse $M^{-1} = \begin{bmatrix} E & F \\ G & H \end{bmatrix}$, then $E = (A - BD^{-1}C)^{-1}$, which is the inverse of the Schur complement of $D$ in $M$.
	  
	  In our case, we're performing the elimination in the inverse matrix $\mathbf{L}^{-1}_{\mathcal{U}_2\mathcal{U}_2}$, and each elimination step corresponds to computing the Schur complement with respect to one node.
	  
	  After eliminating all nodes in $\mathcal{U}_2\setminus\mathcal{U}_1$, by Lemma~\ref{lemma:schur-complement-inverse}, the remaining submatrix corresponding to $\mathcal{U}_1$ is exactly $\mathbf{L}_{\mathcal{U}_1\mathcal{U}_1}^{-1}$, and all other elements are zero.
	  
	  Therefore, $\widetilde{\mathcal{S}}_{|\mathcal{U}_2\setminus\mathcal{U}_1|} = \begin{bmatrix} \mathbf{L}_{\mathcal{U}_1\mathcal{U}_1}^{-1} & 0 \\ 0 & 0 \end{bmatrix}$.
	\end{proof}
\end{lemma}
\comment{
	\begin{equation*}
		\tiny\widetilde{\mathcal{S}}_1 = \begin{bmatrix}
		1.62 & 0.62 & 0.23 & 0 & 0 & 0 & 0.08 & 0 \\
		0.62 & 0.62 & 0.23 & 0 & 0 & 0 & 0.08 & 0 \\
		0.23 & 0.23 & 0.46 & 0 & 0 & 0 & 0.15 & 0 \\
		0 & 0 & 0 & 0.75 & 0.25 & 0.50 & 0 & 0 \\
		0 & 0 & 0 & 0.25 & 0.75 & 0.50 & 0 & 0 \\
		0 & 0 & 0 & 0.50 & 0.50 & 1.00 & 0 & 0 \\
		0.08 & 0.08 & 0.15 & 0 & 0 & 0 & 0.38 & 0 \\
		0 & 0 & 0 & 0 & 0 & 0 & 0 & 0
		\end{bmatrix}.
		\begin{bmatrix}
			1.60 & 0.60 & 0.20 & 0 & 0 & 0 & 0 & 0 \\
			0.60 & 0.60 & 0.20 & 0 & 0 & 0 & 0 & 0 \\
			0.20 & 0.20 & 0.40 & 0 & 0 & 0 & 0 & 0 \\
			0 & 0 & 0 & 0.75 & 0.25 & 0.50 & 0 & 0 \\
			0 & 0 & 0 & 0.25 & 0.75 & 0.50 & 0 & 0 \\
			0 & 0 & 0 & 0.50 & 0.50 & 1.00 & 0 & 0 \\
			0 & 0 & 0 & 0 & 0 & 0 & 0 & 0 \\
			0 & 0 & 0 & 0 & 0 & 0 & 0 & 0
		\end{bmatrix}.
		\end{equation*}
}
\begin{example}\label{example:cholesky-decomposition}
	Consider the graph illustrated in Fig.\ref{fig:example-resistance-distance}(a). Let $\widetilde{\mathcal{S}}_0=\mathbf{L}_{\mathcal{U}_2\mathcal{U}_2}^{-1}=\mathbf{L}_v^{-1}$ with $\mathcal{U}_2=\mathcal{V}\setminus\{v_9\}$, as illustrated in Fig.~\ref{fig:example-graph-matrix}(b). We show the process to compute $\small\begin{bmatrix} \mathbf{L}_{\mathcal{U}_1\mathcal{U}_1}^{-1} & 0 \\ 0 & 0 \end{bmatrix}$ from $\mathbf{L}_{\mathcal{U}_2\mathcal{U}_2}^{-1}$ with $\mathcal{U}_1=\mathcal{V}\setminus\{v_8,v_9\}$. By applying $\small\widetilde{\mathcal{S}}_1 = \widetilde{\mathcal{S}}_0 - \frac{\widetilde{\mathcal{S}}_0[:,v_8]\widetilde{\mathcal{S}}_0[:,v_8]^T}{\widetilde{\mathcal{S}}_0[v_8,v_8]}$, we obtain:
	\begin{equation*}
	\tiny\widetilde{\mathcal{S}}_1 
	= \mathbf{L}_{\mathcal{U}_1\mathcal{U}_1}^{-1} 
	= \begin{bmatrix}
		1.62 & 0.62 & 0.23 & 0 & 0 & 0 & 0.08 & 0 \\
		0.62 & 0.62 & 0.23 & 0 & 0 & 0 & 0.08 & 0 \\
		0.23 & 0.23 & 0.46 & 0 & 0 & 0 & 0.15 & 0 \\
		0 & 0 & 0 & 0.75 & 0.25 & 0.50 & 0 & 0 \\
		0 & 0 & 0 & 0.25 & 0.75 & 0.50 & 0 & 0 \\
		0 & 0 & 0 & 0.50 & 0.50 & 1.00 & 0 & 0 \\
		0.08 & 0.08 & 0.15 & 0 & 0 & 0 & 0.38 & 0 \\
		0 & 0 & 0 & 0 & 0 & 0 & 0 & 0
	\end{bmatrix},
	\end{equation*}
	After eliminating $\mathcal{V}_{cut}=\{v_8,v_9\}$, we can observe the non-zero block structure that reflects two distinct connected components $\{v_1,v_2,v_3,v_7\}$ and $\{v_4,v_5,v_6\}$.
\end{example}

After applying \textit{Gaussian elimination}, we can observe a fundamental connection to the \textit{Cholesky decomposition} \cite{golub2013matrix} of $\mathbf{L}_v^{-1}$. The elimination process iteratively removes rank-1 updates from the matrix, each in the form of a vector outer product $\frac{\widetilde{\mathcal{S}}_i[:,v_k]\widetilde{\mathcal{S}}_i[:,v_k]^T}{\widetilde{\mathcal{S}}_i[v_k,v_k]}$. This reveals a crucial property: any Laplacian submatrix can be expressed as the sum of $n_k$ rank-1 matrices, where each rank-1 matrix is formed by the outer product of a column vector, and $n_k$ is the dimension of the matrix.
\begin{lemma}\label{lemma:cholesky-decomposition}
	Define $\mathcal{U}_i$ as the set of nodes that remain uneliminated at the point when node $v_i$ is being eliminated. Suppose that $\mathcal{S}[:,v_k]$ is the $v_k$-th column of $\widetilde{\mathcal{S}}_i$ at the moment of $v_i$'s elimination, then we have: $\small\left[\begin{matrix}
		\mathbf{L}_{\mathcal{U}_{i}\mathcal{U}_{i}}^{-1} & 0 \\ 0 & 0
		\end{matrix}\right]=\sum_{k=1}^i{\frac{\mathcal{S}[:,v_k]\mathcal{S}[:,v_k]^T}{\mathcal{S}[v_k,v_k]}}$. Specifically, we have: $\mathbf{L}_{\mathcal{U}_{2}\mathcal{U}_{2}}^{-1}-\left[\begin{matrix}
			\mathbf{L}_{\mathcal{U}_{1}\mathcal{U}_{1}}^{-1} & 0 \\ 0 & 0
			\end{matrix}\right]=\sum_{k=|\mathcal{U}_1|+1}^{|\mathcal{U}_2|}{\frac{\mathcal{S}[:,v_k]\mathcal{S}[:,v_k]^T}{\mathcal{S}[v_k,v_k]}}$.
	\begin{proof}
	The proof follows directly from the Gaussian elimination process. When we eliminate a node $v_i$, we perform the operation $\widetilde{\mathcal{S}}_{i} = \widetilde{\mathcal{S}}_{i-1} - \frac{\widetilde{\mathcal{S}}_{i-1}[:,v_i]\widetilde{\mathcal{S}}_{i-1}[:,v_i]^T}{\widetilde{\mathcal{S}}_{i-1}[v_i,v_i]}$. This means we're subtracting a rank-1 matrix from $\widetilde{\mathcal{S}}_{i-1}$. Each elimination step removes exactly one rank-1 matrix of the form $\frac{\mathcal{S}[:,v_k]\mathcal{S}[:,v_k]^T}{\mathcal{S}[v_k,v_k]}$. Since we start with $\mathbf{L}_{\mathcal{U}_2\mathcal{U}_2}^{-1}$ and end with $\left[\begin{matrix} \mathbf{L}_{\mathcal{U}_{1}\mathcal{U}_{1}}^{-1} & 0 \\ 0 & 0 \end{matrix}\right]$ after eliminating all nodes in $\mathcal{U}_2 \setminus \mathcal{U}_1$, the difference between these matrices must be the sum of all the rank-1 matrices we subtracted during elimination. Therefore, $\mathbf{L}_{\mathcal{U}_{2}\mathcal{U}_{2}}^{-1}-\left[\begin{matrix} \mathbf{L}_{\mathcal{U}_{1}\mathcal{U}_{1}}^{-1} & 0 \\ 0 & 0 \end{matrix}\right]=\sum_{k=|\mathcal{U}_1|+1}^{|\mathcal{U}_2|}{\frac{\mathcal{S}[:,v_k]\mathcal{S}[:,v_k]^T}{\mathcal{S}[v_k,v_k]}}$.
	\end{proof}
\end{lemma}
Combined the above results, we can derive a simplified version of the cut property of resistance distance.
\begin{lemma}[Cut property of resistance distance]\label{lemma:simplified-cut-property}
  Let $\mathcal{S}[v_i,s]$ be the $s$-th element of $\mathbf{L}_{\mathcal{U}_i\mathcal{U}_i}^{-1}$. We have:
  \begin{equation*}
      \small r(s,t)=r_{\mathcal{G}_{\mathcal{V}_1}}(s,\Delta_1)+r_{\mathcal{G}_{\mathcal{V}_2}}(t,\Delta_2)+\sum_{v_i\in\mathcal{V}_{cut}}\frac{(\mathcal{S}[v_i,s]-\mathcal{S}[v_i,t])^2}{\mathcal{S}[v_i,v_i]}.
  \end{equation*}
\end{lemma}
\begin{proof}
  By applying Lemma~\ref{lemma:cholesky-decomposition}, we can express the difference between the inverse Laplacian matrices. Since $s \in \mathcal{V}_1$ and $t \in \mathcal{V}_2$, and considering that the elements of the matrix are zero outside their respective blocks, we can decompose the resistance distance calculation as follows:
  \begin{equation*}
   \small \mathbf{e}_s^T\mathbf{L}_{\mathcal{U}\mathcal{U}}^{-1}\mathbf{e}_t=\mathbf{e}_s^T\mathbf{L}_{\mathcal{V}_1\mathcal{V}_1}^{-1}\mathbf{e}_t+\mathbf{e}_s^T\mathbf{L}_{\mathcal{V}_2\mathcal{V}_2}^{-1}\mathbf{e}_t+\sum_{v\in\mathcal{V}_{cut}}\frac{\mathcal{S}[v,s]\mathcal{S}[v,t]}{\mathcal{S}[v,v]}.
\end{equation*}
The Lemma is established since $r_{\mathcal{G}_{\mathcal{V}}}(s,t)=\mathbf{e}_s^T\mathbf{L}_{\mathcal{U}\mathcal{U}}^{-1}\mathbf{e}_s+\mathbf{e}_t^T\mathbf{L}_{\mathcal{U}\mathcal{U}}^{-1}\mathbf{e}_t-2\mathbf{e}_s^T\mathbf{L}_{\mathcal{U}\mathcal{U}}^{-1}\mathbf{e}_t$ and Lemma~\ref{lemma:matrix-difference}.
\end{proof}
\begin{example}
	Consider the graph $\mathcal{G}$ illustrated in Fig.~\ref{fig:example-resistance-distance}(a), which has a vertex cut $\mathcal{V}_{cut}=\{v_7,v_8,v_9\}$. From Fig.~\ref{fig:example-graph-matrix}(b), we observe that when vertex $v_8$ is eliminated, the resulting Schur complement values are $\mathcal{S}[v_8,v_2]=0.04$, $\mathcal{S}[v_8,v_4]=0.40$ and $\mathcal{S}[v_8,v_8]=0.54$. Similarly, Example~\ref{example:cholesky-decomposition} shows that when vertex $v_7$ is eliminated, we obtain $\mathcal{S}[v_7,v_2]=0.08$, $\mathcal{S}[v_7,v_4]=0$, and $\mathcal{S}[v_7,v_7]=0.38$. Applying Lemma~\ref{lemma:simplified-cut-property}, we compute:
$\small r(v_2,v_4) = r_{\mathcal{G}_1}(v_2,\Delta_1) + r_{\mathcal{G}_2}(v_4,\Delta_2) + \sum_{v_i \in \{v_7,v_8\}} \frac{(\mathcal{S}[v_i,v_2] - \mathcal{S}[v_i,v_4])^2}{\mathcal{S}[v_i,v_i]}=1.61$.
\end{example}
The advantage of the simplified cut property is twofold: (i) \textit{compact storage}. For each node $u$ in the sets separated by the vertex cut, we only need to store a single element $\mathcal{S}[v,u]$ for $v\in\mathcal{V}_{cut}$; (ii) \textit{efficient recovery}. When recovering the resistance distance, it suffices to compute the squared differences between the corresponding elements of $\mathcal{S}$ for the nodes in the vertex cut, which exhibits complexity similar to that of a min operation.
\subsection{Dependency Property of Resistance Distance}\label{subsec:tree-decomposition-vertex-hierarchy}
Tree decomposition~\cite{TreeWidth84} is a widely used technique in algorithm design that transforms any graph into a tree structure, thereby imposing a natural hierarchy among its nodes. While Lemma~\ref{lemma:simplified-cut-property} establishes the cut property of resistance for a single \textit{vertex cut}, in this subsection, we generalize this property to the entire graph by introducing the concepts of tree decomposition and vertex hierarchy. Formally, tree decomposition is defiend as:
\begin{definition}(Tree decomposition)
A tree decomposition of a graph $\mathcal{G} = (\mathcal{V}, \mathcal{E})$ consists of a set of subsets (called bags) $\mathcal{X}_{\mathcal{G}} = {\mathcal{X}_1, \mathcal{X}_2, \dots, \mathcal{X}_{|\mathcal{X}_{\mathcal{G}}|}}$ of the node set $\mathcal{V}$, and a tree $\mathcal{T}_{\mathcal{X}_{\mathcal{G}}}$ with node set $\mathcal{X}_{\mathcal{G}}$, satisfying the following three properties:
	(i)	Every node $v \in \mathcal{V}$ appears in at least one bag, i.e., $\forall v \in \mathcal{V}, \exists \mathcal{X}_i \in \mathcal{X}_{\mathcal{G}}$ such that $v \in \mathcal{X}_i$;
	(ii)	For every edge $(u, v) \in \mathcal{E}$, there exists a bag $\mathcal{X}_i \in \mathcal{X}_{\mathcal{G}}$ such that $u, v \in \mathcal{X}_i$;
	(iii)	For every node $v \in \mathcal{V}$, the bags containing $v$ form a connected subtree in $\mathcal{T}_{\mathcal{X}_{\mathcal{G}}}$.
\end{definition}
\comment{
\begin{algorithm}[t]
  \small
  \caption{\mde heuristic tree decomposition \cite{Treewidth06}}
  \label{algo:tree-decomposition}
  \LinesNumbered
  \KwIn{Graph $\mathcal{G} = (\mathcal{V}, \mathcal{E})$}
  \KwOut{An approximate tree decomposition $\mathcal{T}_{min}$}
  Initialize $\mathcal{T}_{min}$ as an empty tree;\\
  \For{$i=1$ \KwTo $|\mathcal{V}|$}{
    Select node $v_i$ with minimum degree in $\mathcal{G}$;\\
    Create a bag $\mathcal{X}_i = \{v_i\} \cup \mathcal{N}_{\mathcal{G}}(v_i)$ and add it to $\mathcal{T}_{min}$;\\
    Remove $v_i$ from $\mathcal{G}$;\\
    Add edges to form a clique among all neighbors $\mathcal{N}_{\mathcal{G}}(v_i)$ in $\mathcal{G}$;
  }
  \For{each node $v_i \in \mathcal{V}$}{
    Find $v_j\in\mathcal{X}_i\setminus\{v_i\}$ with the minimum index $j$;\\
    Set $\mathcal{X}_i.\Parent() \leftarrow \mathcal{X}_j$;
  }
  \Return $\mathcal{T}_{min}$;
\end{algorithm}}
The width of a tree decomposition is defined as $ \max_{i} |\mathcal{X}_i| - 1 $, and the tree height $h_{\mathcal{G}}$ is the maximum distance from each node to the root node of the tree decomposition. The tree-width $ \tw(\mathcal{G}) $ of a graph $ \mathcal{G} $ is the minimum width among all possible tree decompositions of $\mathcal{G}$. Computing an exact tree decomposition is known to be NP-complete~\cite{TreeWidth84}; however, many efficient heuristic algorithms have been developed. Following previous studies~\cite{HopLabeling2018hierarchy,LiJunChang2012exact,ProjectedVertexSeparator2021}, in this paper, we focus on a specific approximate tree decomposition constructed using the \mde (minimum degree) heuristic, which is introduced in \cite{Treewidth06} and performs exceptionally well on real-world networks, exhibiting stronger vertex cut properties. The \mde heuristic algorithm computes a tree decomposition $\mathcal{T}_{min}$ with each tree node corresponding to a node in the graph $\mathcal{G}$ in $O(n(\tw(\mathcal{G})^2+\log n))$ time \cite{LiJunChang2012exact}. Due to space limits, the details of the algorithm can be found in the full version of this paper \cite{fullversion}.

\comment{
Algorithm~\ref{algo:tree-decomposition} presents the pseudo-code of the \mde heuristic tree decomposition algorithm. It iteratively identifies and processes nodes with the minimum degree in $\mathcal{G}$. For each selected node, the algorithm eliminates it from the graph and updates the structure by forming a clique among its neighbors. Simultaneously, the node and its neighbors are added as a bag to the tree decomposition $\mathcal{T}_{min}$ (Lines 2-6). After all nodes have been processed, the algorithm constructs the hierarchical tree structure by assigning each node a parent, specifically selecting the tree node with the minimum index among its potential parents (Lines 7-9). Note that each tree node $\mathcal{X}_i$ corresponds to a node $v_i \in \mathcal{V}$, which is not necessarily the case in a general tree decomposition. In the remainder of this paper, we assume that the tree decomposition refers to $\mathcal{T}_{min}$ and refer to its node $\mathcal{X}_i$ simply as $v_i$. The time complexity of Algorithm~\ref{algo:tree-decomposition} is $O(n(\tw(\mathcal{G})^2+\log n))$ \cite{LiJunChang2012exact}.}
\begin{example}
	Fig.~\ref{fig:example-RDL}(a) illustrates an example of a tree decomposition of the graph $\mathcal{G}$ in Fig.~\ref{fig:example-resistance-distance}(a), constructed using the \mde heuristic. The treewidth of $\mathcal{T}_{min}$ is $2$, and the tree height $h_\mathcal{G}$ is $6$.
\end{example}
The \mde heuristic tree decomposition possesses stronger vertex hierarchy properties than a general tree decomposition.
\begin{lemma}[Vertex hierarchy property of tree decomposition \cite{LiJunChang2012exact}]\label{lemma:vertex-hierarchy-property}
	For a tree decomposition $\mathcal{T}_{min}$ obtained using the \mde heuristic, we can derive: For any bag $\mathcal{X}_u$ in $\mathcal{T}_{min}$, all nodes in $\mathcal{X}_u$ except $u$ itself are ancestors of $u$ in $\mathcal{T}_{min}$. Consequently, for any two nodes $s$ and $t$, their lowest common ancestor (\lca) and its ancestor nodes in $\mathcal{T}_{min}$ form a vertex cut that partitions the graph into distinct connected components containing $s$ and $t$, respectively.
\end{lemma}
\begin{example}
	Consider the tree decomposition $\mathcal{T}_{min}$ illustrated in Fig.~\ref{fig:example-RDL}(a). It can be observed that the node $v_8$ and its ancestor $v_9$ together form a vertex cut that partitions $\{v_1, v_2, v_3, v_7\}$ and $\{v_4, v_5, v_6\}$.
\end{example}
Combined with the cut property of resistance distance, the vertex hierarchy property of tree decomposition provides a compact approach for storing distance labels. We formally define resistance distance labelling and illustrate the non-zero structure that arises when it is integrated with the tree decomposition.
\begin{figure}[t!]
	\vspace*{-0.2cm}
	\begin{center}
		\begin{tabular}[t]{c}
			\subfigure[Tree decomposition of $\mathcal{G}$]{
				\raisebox{0.8cm}{\includegraphics[width=0.48\columnwidth]{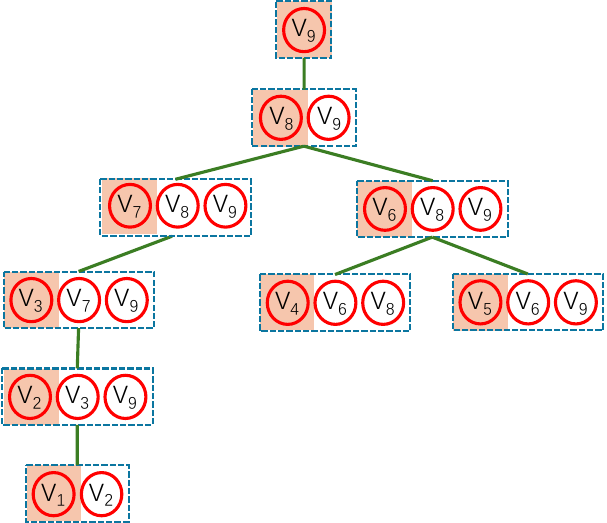}}
			}
			\subfigure[Resistance distance labelling]{
				\raisebox{0.1cm}{\includegraphics[width=0.48\columnwidth]{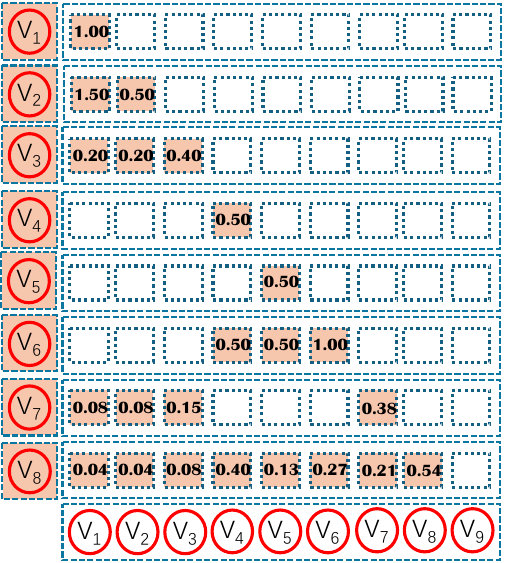}}
			}
		\end{tabular}
	\end{center}
	\vspace*{-0.6cm}
	\caption{An illustrative example of tree decomposition and resistance distance labelling}
	\vspace*{-0.2cm}
	\label{fig:example-RDL}
\end{figure}

\begin{definition}[Resistance distance labelling]\label{def:resistance-distance-labelling} Given a graph $\mathcal{G}$ and a tree decomposition $\mathcal{T}_{min}$, suppose that we apply \textit{Gaussian elimination} following the reverse ordering of the nodes processed in the \mde heuristic tree decomposition. The resistance distance labelling can be represented as $\mathcal{S}[v,u]$, which stores the $u$-th element of the $v$-th column of $\mathbf{L}^{-1}_{\mathcal{U}\mathcal{U}}$, $\mathcal{U}$ is the remaining set of nodes when $v$ is eliminated.
\end{definition}
\begin{lemma}[Non-zero structure of resistance distance labelling]
	The resistance distance labelling $\mathcal{S}[v,u]$ has the following non-zero structures: (i) For each node $v$, all nodes $u$ in the subtree of $v$ are non-zero; (ii) For each node $u$, all nodes $v$ in the path from $u$ to the root node are non-zero.
\begin{proof}
The non-zero structure of the resistance distance labelling can be derived from the vertex hierarchy property of the tree decomposition (Lemma~\ref{lemma:vertex-hierarchy-property}) and the cut property of resistance distance (Lemma~\ref{lemma:simplified-cut-property}).

For property (i), consider a node $v$ and any node $u$ in its subtree. When $v$ is eliminated during the \mde process, all nodes in its subtree (including $u$) are still present in the remaining graph. According to Definition~\ref{def:resistance-distance-labelling}, $\mathcal{S}[v,u]$ represents the $u$-th element of the $v$-th column of $\mathbf{L}^{-1}_{\mathcal{U}\mathcal{U}}$, where $\mathcal{U}$ is the set of remaining nodes after $v$ is eliminated. Since $u \in \mathcal{U}$, the corresponding entry in the inverse Laplacian is non-zero due to the connectivity between $v$ and nodes in its subtree. All other entries in the resistance distance labelling are zero because they correspond to pairs of nodes that are separated by vertex cut formed by ancestors of $v$ in $\mathcal{T}_{min}$, as established by the tree decomposition structure and the cut property of resistance distance.

For property (ii), it essentially provides the reverse perspective of property (i). Since $\mathcal{S}[v,u]$ is non-zero for all $u$ in the subtree of $v$, if $v$ lies on the path from $u$ to the root, we can conclude that $\mathcal{S}[v,u]$ is non-zero for all nodes $v$ that appear on the path from $u$ to the root in the tree decomposition. Conversely, for any node $v$ that is not on the path from $u$ to the root, $\mathcal{S}[v,u]$ will be zero. This includes two categories: (1) nodes in the subtree rooted at $u$ (i.e., children of $u$ and their descendants), and (2) nodes in other branches of the tree. For children of $u$ and their descendants, when these nodes are eliminated, $u$ has already been eliminated earlier according to the reverse ordering of the \mde process, so $u$ is not in the remaining set $\mathcal{U}$ when computing $\mathbf{L}^{-1}_{\mathcal{U}\mathcal{U}}$, resulting in zero entries. For nodes in other branches, the vertex hierarchy property ensures that the lowest common ancestor of $u$ and such nodes (along with its ancestors) forms a vertex cut that separates them in the original graph. By the cut property of resistance distance, this separation leads to zero entries in the resistance distance labelling.
\end{proof}
\end{lemma}
\comment{
\begin{proof}[Proof sketch]
	\textcolor{red}{The proof relies on the vertex hierarchy property and the cut property of resistance distance. For property (i), when a node $v$ is eliminated, any node $u$ in its subtree remains connected in the graph, resulting in a non-zero $\mathcal{S}[v,u]$. Property (ii) is the dual of (i): if $v$ is an ancestor of $u$, then $u$ is in $v$'s subtree, making $\mathcal{S}[v,u]$ non-zero. All other entries are zero because ancestral nodes form vertex cuts that separate unrelated node pairs.}
\end{proof}}
\begin{example}
Fig.~\ref{fig:example-RDL}(b) illustrates the non-zero structure of resistance distance labelling corresponding to the tree decomposition in Fig.~\ref{fig:example-RDL}(a). For node $v_2$, $\mathcal{S}[v_2,u]$ is non-zero for nodes in its subtree $\{v_1,v_2\}$, while $\mathcal{S}[v,v_2]$ is non-zero for nodes on the path to root $\{v_2,v_3,v_7,v_8\}$.
\end{example}
Given a tree decomposition and the corresponding resistance distance labelling, we can now demonstrate that the resistance distance $r(s,t)$ depends solely on the labels stored along the paths from nodes $s$ and $t$ to the root of the tree.
\begin{lemma}[Dependence property of resistance distance]\label{lemma:dependence-property}
	In $\mathcal{T}_{min}$ obtained from the \mde heuristic tree decomposition, the resistance distance $r(s,t)$ depends only on the labels stored along the paths from nodes $s$ and $t$ to the root of $\mathcal{T}_{min}$. Specifically,
  \begin{align*}
	  \tiny r(s,t) &= \sum_{v\in\mathcal{P}_{s\leadsto\lca(s,t)}}\frac{(\mathcal{S}[v,s])^2}{\mathcal{S}[v,v]}+\sum_{v\in\mathcal{P}_{t\leadsto\lca(s,t)}}\frac{(\mathcal{S}[v,t])^2}{\mathcal{S}[v,v]}\\
		&+\sum_{v\in\mathcal{P}_{\lca(s,t)\leadsto root}}\frac{(\mathcal{S}[v,s]-\mathcal{S}[v,t])^2}{\mathcal{S}[v,v]}.
  \end{align*}
  \begin{proof}
According to Lemma~\ref{lemma:cholesky-decomposition}, we have $\mathbf{L}_v^{-1} = \sum_{v}{\frac{\mathcal{S}[:,v]\mathcal{S}[:,v]^T}{\mathcal{S}[v,v]}}$. Substituting this into the resistance distance formula, we obtain:
\begin{align*}
    \small r(s,t) &= (\mathbf{e}_s-\mathbf{e}_t)^T\left(\sum_{v}{\frac{\mathcal{S}[:,v]\mathcal{S}[:,v]^T}{\mathcal{S}[v,v]}}\right)(\mathbf{e}_s-\mathbf{e}_t) \\
    &= \sum_{v}{\frac{(\mathbf{e}_s-\mathbf{e}_t)^T\mathcal{S}[:,v]\mathcal{S}[:,v]^T(\mathbf{e}_s-\mathbf{e}_t)}{\mathcal{S}[v,v]}} \\
    &= \sum_{v}{\frac{(\mathcal{S}[v,s]-\mathcal{S}[v,t])^2}{\mathcal{S}[v,v]}}.
\end{align*}

Based on the sparsity structure of the resistance distance labelling and the properties of the tree decomposition, we can partition this sum into three parts: (i) For $v \in \mathcal{P}_{s\leadsto\lca(s,t)}$, $\mathcal{S}[v,t] = 0$, so $\frac{(\mathcal{S}[v,s]-\mathcal{S}[v,t])^2}{\mathcal{S}[v,v]} = \frac{(\mathcal{S}[v,s])^2}{\mathcal{S}[v,v]}$.
(ii) For $v \in \mathcal{P}_{t\leadsto\lca(s,t)}$, $\mathcal{S}[v,s] = 0$, so $\frac{(\mathcal{S}[v,s]-\mathcal{S}[v,t])^2}{\mathcal{S}[v,v]} = \frac{(\mathcal{S}[v,t])^2}{\mathcal{S}[v,v]}$.
(iii) For $v \in \mathcal{P}_{\lca(s,t)\leadsto \text{root}}$, both $\mathcal{S}[v,s]$ and $\mathcal{S}[v,t]$ may be nonzero. Therefore, the sum of the three parts is equal to the resistance distance $r(s,t)$.
\end{proof}
  \end{lemma}
\begin{example}
    Consider computing $r(v_2,v_4)$ in Fig.~\ref{fig:example-RDL}. The path from $v_2$ to the root is $(v_2,v_3,v_7,v_8)$, and from $v_4$ to the root is $(v_4,v_6,v_8)$, with $\lca(v_2,v_4)=v_8$. Using Lemma~\ref{lemma:dependence-property}, we have: 
    \begin{align*}
    \small r(v_2,v_4) &= \sum_{v\in\{v_2,v_3,v_7\}}\frac{(\mathcal{S}[v,v_2])^2}{\mathcal{S}[v,v]} + \sum_{v\in\{v_4,v_6\}}\frac{(\mathcal{S}[v,v_4])^2}{\mathcal{S}[v,v]} \\
    &+ \sum_{v\in\{v_8\}}\frac{(\mathcal{S}[v,v_2]-\mathcal{S}[v,v_4])^2}{\mathcal{S}[v,v]}=1.61.
    \end{align*}
\end{example}

\section{The Proposed Labelling Scheme}\label{sec:resistance-distance-labelling}
Building upon the resistance distance labelling introduced in Section~\ref{sec:resistance-distance-property}, we now address the challenge of efficient implementation. Two key questions remain: (i) how can the resistance distance labelling be stored in a space-efficient manner, and (ii) how can this labelling be computed efficiently? In this section, we propose a resistance labelling scheme, \treeindex. We first describe the structure of \treeindex, followed by efficient algorithms for label construction. Finally, we present efficient query processing algorithms for both single-pair and single-source queries based on \treeindex.
\subsection{Labelling Structure}
To efficiently store the resistance distance labelling described in Definition~\ref{def:resistance-distance-labelling}, it is essential to first identify the non-zero structure of the labelling. This is achieved by performing a \dfs traversal on $\mathcal{T}_{min}$. The resulting labelling framework comprises two main components: (i) the tree decomposition $\mathcal{T}_{min}$, and (ii) the resistance distance labelling $\mathcal{S}$. For the first part, the tree decomposition $\mathcal{T}_{min}$ is stored as a tree structure. $\mathcal{T}_{min}.\Parent()$ maintains the parent for each node in the tree. We then perform a \dfs traversal on $\mathcal{T}_{min}$ to determine the position of each node $u$ in the subtree structure, where $\mathcal{T}_{min}.\dfsorder[u]$ records the order of each node in the traversal sequence. For the second part, for each node $u\in\mathcal{V}$, the resistance distance labelling $\mathcal{S}[u].\res$ stores the label values for all nodes in the subtree rooted at $u$, organized according to the \dfs ordering scheme. Specifically, we have:
\begin{lemma}
	$\mathcal{S}[u].\res$ contains exactly $|\mathcal{T}_{min}.\subtree[u]|$ elements, where $\mathcal{T}_{min}.\subtree[u]$ denotes the subtree rooted at $u$ in $\mathcal{T}_{min}$. $\mathcal{S}[v,u]$ can be visited via $\mathcal{S}[u].\res[\mathcal{T}_{min}.\dfsorder[u] - \mathcal{T}_{min}.\dfsorder[v]]$.
	\begin{proof}
		This result follows from the properties of \dfs traversal on trees. During a \dfs traversal, all nodes in the subtree rooted at any node $u$ are visited consecutively before the traversal backtracks to nodes outside this subtree. When performing a \dfs traversal on $\mathcal{T}_{min}$, each node $v$ is assigned a position $\mathcal{T}_{min}.\dfsorder[v]$ in the traversal sequence. For any node $v$ in the subtree of $u$, the values $\mathcal{T}_{min}.\dfsorder[v]$ form a contiguous range starting from $\mathcal{T}_{min}.\dfsorder[u]$. $\mathcal{T}_{min}.\dfsorder[u] - \mathcal{T}_{min}.\dfsorder[v]$ indicates the relative position of node $u$ within the subtree rooted at $v$, and is unique for each node $v$ in this subtree. In our index structure, $\mathcal{S}[v].\res$ stores the resistance distance labelling for all nodes in the subtree of $u$, arranged according to their relative positions in the \dfs ordering. Therefore, $\mathcal{S}[v].\res[\mathcal{T}_{min}.\dfsorder[u] - \mathcal{T}_{min}.\dfsorder[v]]$ directly retrieves the value of $\mathcal{S}[v,u]$ for any node $u$ in the subtree of $v$.
	\end{proof}
\end{lemma}
\begin{figure}[t!]
	\vspace*{-0.2cm}
	\begin{center}
		\begin{tabular}[t]{c}
			\subfigure[Tree decomposition processed by a \dfs ordering]{
				\raisebox{0.4cm}{\includegraphics[width=0.3\columnwidth, height=3.2cm]{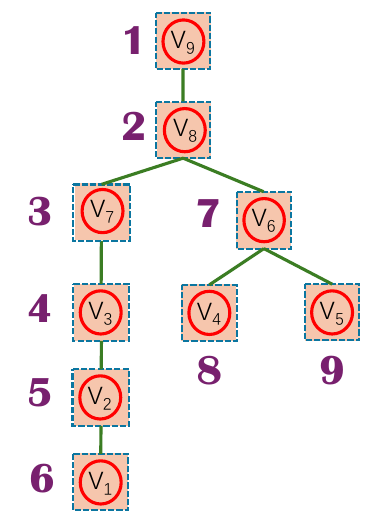}}
			}
			\subfigure[Resistance distance labelling]{
				\raisebox{0.1cm}{\includegraphics[width=0.52\columnwidth, height=3.6cm]{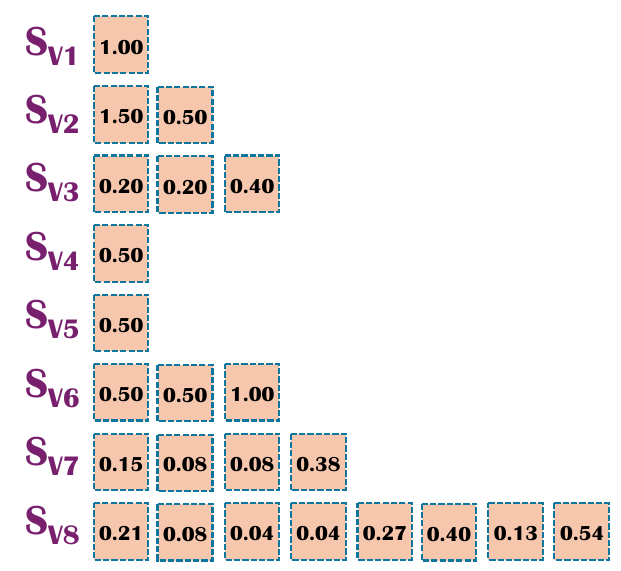}}
			}
		\end{tabular}
	\end{center}
	\vspace*{-0.6cm}
	\caption{An example of the labelling structure}
	\vspace*{-0.2cm}
	\label{fig:example-index-structure}
\end{figure}
\begin{example}
	Fig.~\ref{fig:example-index-structure} illustrates the labelling structure of \treeindex. Fig.~\ref{fig:example-index-structure}(a) depicts a \dfs ordering of $\mathcal{T}_{min}$ starting from the root node $v_9$, while Fig.~\ref{fig:example-index-structure}(b) presents the resistance distance labelling $\mathcal{S}$ rearranged according to this \dfs ordering. To access $\mathcal{S}[v_8, v_4]$, we compute $\mathcal{T}_{min}.\dfsorder[v_4] - \mathcal{T}_{min}.\dfsorder[v_8] = 8 - 2 = 6$, and then retrieve $\mathcal{S}[v_8].\res[6]$, which equals $0.40$ as shown in Fig.~\ref{fig:example-index-structure}(b).
\end{example}
Given the labelling structure, the following lemma gives an upper bound on the label size:
\begin{lemma}\label{lemma:index-size}
  The label size of \treeindex is $O(n\cdot h_{\mathcal{G}})$.
  \begin{proof}
    For each node $u$ in the graph, we need to store:
    (i) The parent pointer in $\mathcal{T}_{min}$, which requires $O(1)$ space per node;
    (ii) The \dfs ordering information, which also requires $O(1)$ space per node;
    (iii) The resistance distance labelling $\mathcal{S}[u].\res$, which stores values for all nodes in the path from $u$ to the root of $\mathcal{T}_{min}$. Since the height of $\mathcal{T}_{min}$ is at most $h_{\mathcal{G}}$, the path from any node to the root contains at most $h_{\mathcal{G}}$ nodes. Therefore, the size of $\mathcal{S}[u].\res$ is bounded by $O(h_{\mathcal{G}})$ for each node $u$. With $n$ nodes in total, the overall space complexity is $O(n) + O(n) + O(n \cdot h_{\mathcal{G}}) = O(n \cdot h_{\mathcal{G}})$.
  \end{proof}
\end{lemma}
\comment{
\begin{proof}[Proof sketch]
	\textcolor{red}{The total label size is the sum of storage required for all nodes. For each node $u$, due to its non-zero structure, only needs to store values for the path from $u$ to the root. Since this path has a length at most the tree height $h_{\mathcal{G}}$, the total size is $O(n\cdot h_{\mathcal{G}})$.}
\end{proof}}
In practice, the tree height $h_{\mathcal{G}}$ is typically not very large, and the actual label size is often much smaller than the theoretical upper bound (see Section~\ref{sec:experiment}, Table~\ref{tab:datasets}). For example, in the full USA road network, the tree height is $3,976$. The actual number of non-zero labels is approximately $2,268 \times n$, corresponding to less than $405$~GB, which can be entirely loaded into memory for a commodity server..

\subsection{Label Construction Algorithm}
Based on our analysis, the first component $\mathcal{T}_{min}$ of the labelling structure can be computed during the tree decomposition process and a simple \dfs traversal. The main challenge lies in computing the second component, $\mathcal{S}$, which consists of elements from the inverse Laplacian submatrix $\mathbf{L}_{\mathcal{U}\mathcal{U}}^{-1}$ for various sets $\mathcal{U}$. A straightforward approach is to solve a linear system, resulting in a time complexity of $\widetilde{O}(m)$ (nearly linear in the number of edges), according to the state-of-the-art Laplacian solver~\cite{LaplacianSolver}. However, this method incurs a significant hidden constant factor in the $\widetilde{O}()$ notation and is therefore not efficient in practice. To overcome this limitation, we propose an efficient incremental algorithm that iteratively applies rank-$1$ updates to the inverse Laplacian matrix. This method enables us to compute the resistance distance labelling in a bottom-up manner, following the reverse \dfs order of $\mathcal{T}_{min}$, thereby significantly improving computational efficiency.

Recall that $\mathcal{S}[v, u]$ denotes the $(u, v)$-th element of $\mathbf{L}_{\mathcal{U}_i\mathcal{U}_i}^{-1}$ when node $v$ is eliminated. When $\mathcal{U}_{i+1}$ differs from $\mathcal{U}_i$ by only a single node, the matrix can be efficiently updated using a rank-1 update. The key question is as follows: suppose $\mathcal{U}_{i+1} = \mathcal{U}_i \cup \{v_{i+1}\}$; given $\mathbf{L}_{\mathcal{U}_i\mathcal{U}_i}^{-1}$, how can we compute a column of $\mathbf{L}_{\mathcal{U}_{i+1}\mathcal{U}_{i+1}}^{-1}$? We address this by leveraging the rank-$1$ update formula.
\begin{lemma}\label{lemma:rank-1-update}
Suppose $\mathcal{U}_{i+1} = \mathcal{U}_i \cup \{v_{i+1}\}$, the inverse Laplacian submatrix $\mathbf{L}_{\mathcal{U}_{i+1}\mathcal{U}_{i+1}}^{-1}$ and $\mathbf{L}_{\mathcal{U}_i\mathcal{U}_i}^{-1}$ satisfy:
\begin{equation*}
	\mathbf{e}_{v_{i+1}}\mathbf{L}_{\mathcal{U}_{i+1}\mathcal{U}_{i+1}}^{-1}\mathbf{e}_{v_{i+1}}=\frac{1}{d_{v_{i+1}}-\mathbf{1}^T_{\mathcal{N}(v_{i+1})\cap\mathcal{U}_i}\mathbf{L}_{\mathcal{U}_i\mathcal{U}_i}^{-1}\mathbf{1}_{\mathcal{N}(v_{i+1})\cap\mathcal{U}_i}},
\end{equation*}
\begin{equation*}
	\mathbf{L}_{\mathcal{U}_{i+1}\mathcal{U}_{i+1}}^{-1}\mathbf{e}_{v_{i+1}}=\frac{\mathbf{L}_{\mathcal{U}_i\mathcal{U}_i}^{-1}\mathbf{1}_{\mathcal{N}(v_{i+1})\cap\mathcal{U}_i}}{d_{v_{i+1}}-\mathbf{1}^T_{\mathcal{N}(v_{i+1})\cap\mathcal{U}_i}\mathbf{L}_{\mathcal{U}_i\mathcal{U}_i}^{-1}\mathbf{1}_{\mathcal{N}(v_{i+1})\cap\mathcal{U}_i}}.
\end{equation*}
\begin{proof}
	According to the block matrix inverse formula, we have:
	\begin{equation*}
		\left[\begin{matrix}
			A & B \\ C & D
		\end{matrix}\right]^{-1}= \left[\begin{matrix}
			A^{-1}+A^{-1}BS^{-1}CA^{-1} & -A^{-1}BS^{-1} \\ -S^{-1}CA^{-1} & S^{-1}
		\end{matrix}\right]
	\end{equation*}
	where $S = D - CA^{-1}B$ is the Schur complement. Applying this formula to our case with $A=\mathbf{L}_{\mathcal{U}_i\mathcal{U}_i}$ and $B=-\mathbf{1}_{\mathcal{N}(v_{i+1})\cap\mathcal{U}_i}$ yields the desired result.
\end{proof}
\end{lemma}
\begin{algorithm}[t]
	\caption{Label construction algorithm}
	\label{algo:index-building}
	\LinesNumbered
	\KwIn{Graph $\mathcal{G}$, tree decomposition $\mathcal{T}_{min}$}
	\KwOut{Resistance distance labelling $\mathcal{S}$}
	Initialize $\mathcal{S}$, $\mathcal{U} \leftarrow \emptyset$;\\
	$\order\leftarrow$ reverse of $\mathcal{T}_{min}.\dfsorder()$;\\
	\For{$i=1:n$}{
	    $v_i \leftarrow \order[i]$, $\mathcal{U} \leftarrow \mathcal{U} \cup \{v_i\}$;\\
	    \For{each $w\in\mathcal{N}(v_i)\cap\mathcal{U}$}{
	        $v_k\leftarrow w$;\\
	        \While{$v_k\neq v_i$}{
	          $\ratio \leftarrow \frac{\mathcal{S}[v_k, v_w]}{\mathcal{S}[v_k, v_k]}$;\\
	          \For{each node $u$ in $\mathcal{T}_{min}.\subtree(v_k)$}{
	            Update $\mathcal{S}[v_i, u]$ by adding $\mathcal{S}[v_k, u] \times \ratio$;
	          }
	          $v_k \leftarrow \mathcal{T}_{min}.\Parent(v_k)$;
	        }
	    }
	    $\mathcal{S}[v_i,v_i] \leftarrow \frac{1}{d_{v_i}-\sum_{w\in\mathcal{N}(v_i)\cap\mathcal{U}}\mathcal{S}[v_i,w]}$;\\
    	\For{each node $u$ in $\mathcal{T}_{min}.\subtree(v_i)$}{
    	    $\mathcal{S}[v_i,u] \leftarrow \frac{\mathcal{S}[v_i,u]}{d_{v_i}-\sum_{w\in\mathcal{N}(v_i)\cap\mathcal{U}}\mathcal{S}[v_i,w]}$;\\
			$\mathcal{S}.\diagonal[u] \leftarrow \left(\mathcal{S}[v_i,u]\right)^2$;
    	}
	}
	\Return $\mathcal{S}$;
\end{algorithm}
Notice that during construction, each column of $\mathbf{L}_{\mathcal{U}_i\mathcal{U}_i}^{-1}$ can be computed based on the current state. The pseudo-code of the algorithm is presented in Algorithm~\ref{algo:index-building}. The algorithm computes the resistance distance labelling by processing nodes according to the reverse \dfs order of the tree decomposition (line~2). For each node $v_i$ in this order (lines~3-4), the algorithm first adds $v_i$ to the processed set $\mathcal{U}$. Then, for each neighbor $w$ of $v_i$ that has already been processed (line~5), it traverses up the tree from $w$ to $v_i$ (lines~6-12). For each node $v_k$ in this path, the algorithm computes a ratio based on the current labelling (line~8) and uses it to update the labelling for all nodes in the subtree rooted at $v_k$ (lines~9-11). After processing all neighbors, the algorithm computes the diagonal element for $v_i$ (line~13) and normalizes the entries for all nodes in the subtree rooted at $v_i$ (line~14). Furthermore, during the index-building process, we can additionally include the diagonal of $\mathbf{L}_v^{-1}$ to support single-source queries (line~15). During the algorithm, $\mathcal{S}[v_k, u]$ is accessed as $\mathcal{S}.\res[\mathcal{T}_{min}.\dfsorder[u] - \mathcal{T}_{min}.\dfsorder[v_k]]$. This approach efficiently propagates information through the tree structure to compute the complete resistance distance labelling.

We now analyze the correctness and time complexity of the algorithm. We first need the following lemma to ensure the while loop in Algorithm~\ref{algo:index-building} terminates.
\begin{lemma}\label{while-loop-termination}
  For each node $w\in\mathcal{N}(v_i)\cup\mathcal{U}$, $v_i$ must be the parent of $w$ in $\mathcal{T}_{min}$.
  \begin{proof}
  Consider any node $w \in \mathcal{N}(v_i) \cap \mathcal{U}$. By the algorithm's processing order, $w$ has been processed before $v_i$ since $w \in \mathcal{U}$ when we process $v_i$. In a tree decomposition, for any edge $(u,v)$ in the original graph, there must exist at least one tree node containing both $u$ and $v$. Since $w$ and $v_i$ are neighbors in $\mathcal{G}$, they must share at least one tree node in $\mathcal{T}_{min}$. Given that we process nodes according to the reverse \dfs order, when we process $v_i$, all nodes in its subtree (including $w$) have already been processed. The tree structure ensures that $v_i$ is the parent of $w$ in $\mathcal{T}_{min}$, as this is the only arrangement that maintains the required connectivity property while preserving the properties of the tree decomposition. Therefore, for each node $w \in \mathcal{N}(v_i) \cap \mathcal{U}$, $v_i$ must be the parent of $w$ in $\mathcal{T}_{min}$.
  \end{proof}
\end{lemma}
\begin{lemma}[Correctness of Algorithm~\ref{algo:index-building}]
  Algorithm~\ref{algo:index-building} computes the resistance distance labelling of the graph $\mathcal{G}$.
  \begin{proof}
The correctness follows from the rank-1 update formula presented in Lemma~\ref{lemma:rank-1-update}. For each node $v_i$, the algorithm implements this formula by:
(i) Computing $\mathbf{L}_{\mathcal{U}_i\mathcal{U}_i}^{-1}\mathbf{1}_{\mathcal{N}(v_{i+1})\cap\mathcal{U}_i}$ through the traversal and update process (Lines 5-12);
(ii) Computing the denominator $d_{v_i}-\mathbf{1}^T_{\mathcal{N}(v_{i+1})\cap\mathcal{U}_i}\mathbf{L}_{\mathcal{U}_i\mathcal{U}_i}^{-1}\mathbf{1}_{\mathcal{N}(v_{i+1})\cap\mathcal{U}_i}$ (Line 13);
(iii) Normalizing by this denominator (Lines 14-16). Line 15 specifically stores the diagonal elements of $\mathbf{L}_v^{-1}$ to support single-source queries. According to Lemma~\ref{while-loop-termination}, the while loop will terminate as $v_i$ is guaranteed to be an ancestor of each $w \in \mathcal{N}(v_i) \cap \mathcal{U}$ in the tree. By processing nodes according to the \dfs order of the tree decomposition, the algorithm correctly builds the complete resistance distance labelling.
\end{proof}
\end{lemma}
\comment{
\begin{proof}[Proof sketch]
	\textcolor{red}{The correctness hinges on the iterative application of the rank-1 update formula for the inverse Laplacian from Lemma~\ref{lemma:rank-1-update}. For each node $v_i$, the nested loops correctly compute the vector and scalar components of this formula. By processing nodes in the specified \dfs order, these incremental updates correctly build the final resistance distance labeling for the entire graph.}
\end{proof}}
\begin{lemma}[Time Complexity Analysis]
  Algorithm~\ref{algo:index-building} has a time complexity of $O(n\cdot d_{max}\cdot h_{\mathcal{G}}^2)$.
  \begin{proof}
    We analyze the time complexity by examining each component of the algorithm: First, initialization and DFS order computation (lines 1-2) require $O(n)$ time. The main outer loop (line 3) executes exactly $n$ iterations, processing each node $v_i$ sequentially. For each node $v_i$, the inner loop (line 5) iterates through its already processed neighbors, bounded by $|\mathcal{N}(v_i)\cap\mathcal{U}| \leq d_{v_i} \leq d_{max}$. Within this loop, the while loop (lines 7-12) traverses from each processed neighbor $v_k$ up to $v_i$ in the tree, updating the labelling for all nodes in the subtree rooted at $v_k$. This operation takes at most $O(|\mathcal{T}_{min}.\subtree(v_k)|)$ steps for each $v_k$. Since Lemma~\ref{while-loop-termination} establishes that $v_i$ is an ancestor of each $v_k$, we can bound the total work: for each node $v_i$, the combined size of all subtrees rooted at its processed neighbors cannot exceed $v_i$'s own subtree size multiplied by $d_{max}$. This gives us an upper bound of $O(d_{max} \cdot n \cdot h_{\mathcal{G}}^2)$ operations across all nodes, since $\sum_{i=1}^n |\mathcal{T}_{min}.\subtree(v_i)| = O(n \cdot h_{\mathcal{G}})$. The final update operations (lines 13-16) require $O(|\mathcal{T}_{min}.\subtree(v_i)|)$ time for each node, summing to $O(n \cdot h_{\mathcal{G}})$ across all nodes. Combining these analyses, we conclude that the overall time complexity of Algorithm~\ref{algo:index-building} is $O(n \cdot d_{max} \cdot h_{\mathcal{G}}^2)$.
  \end{proof}
\end{lemma}
\comment{
\begin{proof}[Proof sketch]
	\textcolor{red}{The algorithm's time complexity is dominated by the main loop iterating through $n$ nodes. For each node $v_i$, the algorithm processes its neighbors (at most $d_{max}$) and performs updates on their respective subtrees. The total work for these updates, when summed across all nodes, is proportional to the sum of all subtree sizes in the decomposition tree, which is bounded by $O(n\cdot h_{\mathcal{G}})$, resulting in a total complexity of $O(n\cdot d_{max}\cdot h_{\mathcal{G}})$.}
\end{proof}}
In practice, we observe that the maximum degree $d_{max}$ is typically small in large road networks (generally less than 10 across most datasets). Therefore, our proposed algorithm demonstrates high efficiency when applied to large-scale road network analysis.
\begin{algorithm}[t]
	\caption{Single-pair query processing algorithm}
	\label{algo:single-pair-query-processing}
	\LinesNumbered
	\KwIn{Graph $\mathcal{G}$, resistance labelling $\mathcal{S}$, query nodes $s$ and $t$}
	\KwOut{Resistance distance $r(s,t)$}
	Initialize $r(s,t) \leftarrow 0$;\\
	$\lca(s,t) \leftarrow$ least common ancestor of $s$ and $t$ in $\mathcal{T}_{min}$;\\
	$w \leftarrow s$;\\
	\While{$w \neq \lca(s,t)$}{
		$r(s,t) \leftarrow r(s,t) + \frac{(\mathcal{S}[w, s])^2}{\mathcal{S}[w, w]}$;\\
		$w \leftarrow \mathcal{T}_{min}[w].\Parent()$;\\
	}
	$w \leftarrow t$;\\
	\While{$w \neq \lca(s,t)$}{
		$r(s,t) \leftarrow r(s,t) + \frac{(\mathcal{S}[w, t])^2}{\mathcal{S}[w, w]}$;\\
		$w \leftarrow \mathcal{T}_{min}[w].\Parent()$;\\
	}
	$w \leftarrow \lca(s,t)$;\\
	\While{$w \neq \mathcal{T}_{min}.\rootvector()$}{
		$r(s,t) \leftarrow r(s,t) + \frac{(\mathcal{S}[w, s] - \mathcal{S}[w, t])^2}{\mathcal{S}[w, w]}$;\\
		$w \leftarrow \mathcal{T}_{min}[w].\Parent()$;\\
	}
	\Return $r(s,t)$;
\end{algorithm}
\subsection{Query Processing Algorithms}
Based on the constructed resistance distance labelling, we present efficient algorithms for processing both single-pair and single-source resistance distance queries.

\stitle{Single-pair query processing.} For a query $r(s,t)$, we compute the resistance distance between nodes $s$ and $t$ by leveraging the tree structure of $\mathcal{T}_{\min}$ and the precomputed resistance distance labelling. Algorithm~\ref{algo:single-pair-query-processing} presents our approach, which consists of three main phases. First, we identify the least common ancestor (\lca) of nodes $s$ and $t$ in the tree decomposition (Line 2). Next, we traverse upward from node $s$ to the \lca, accumulating resistance contributions along this path (Lines 3-6). Similarly, we traverse from node $t$ to the \lca (Lines 7-10). Finally, we continue from the \lca to the root of the tree, adding the squared differences of the resistance labels (Lines 11-14). This approach correctly computes $r(s,t)$ by exploiting the dependency property of resistance distances, which allows us to decompose the calculation into contributions from each ancestor in the tree decomposition.
\begin{lemma}[Time Complexity of Single-Pair Queries]
	Algorithm~\ref{algo:single-pair-query-processing} has a time complexity of $|\mathcal{P}_{s\leadsto root}|+|\mathcal{P}_{t\leadsto root}|$, where $\mathcal{P}_{s\leadsto root}$ and $\mathcal{P}_{t\leadsto root}$ are the paths from $s$ and $t$ to the root in the tree decomposition, respectively. This complexity can further be bounded by $O(h_{\mathcal{G}})$.
  \begin{proof}
    The time complexity of Algorithm~\ref{algo:single-pair-query-processing} is determined by the number of nodes visited during the three traversal phases. In the first phase, we traverse from node $s$ to $\lca(s,t)$, which requires $|\mathcal{P}_{s\leadsto \lca(s,t)}|$ steps. Similarly, the second phase traverses from $t$ to $\lca(s,t)$, requiring $|\mathcal{P}_{t\leadsto \lca(s,t)}|$ steps. The final phase traverses from $\lca(s,t)$ to the root, taking $|\mathcal{P}_{\lca(s,t)\leadsto root}|$ steps. Since each step performs only constant-time operations (arithmetic calculations using precomputed labels), the total time complexity is $|\mathcal{P}_{s\leadsto root}| + |\mathcal{P}_{t\leadsto root}|$. Since the length of any path from a node to the root is bounded by the height of the tree $h_{\mathcal{G}}$, the overall time complexity is $O(h_{\mathcal{G}})$.
  \end{proof}
\end{lemma}

\begin{algorithm}[t]
	\caption{Single-source query processing algorithm}
	\label{algo:single-source-query-processing}
	\LinesNumbered
	\KwIn{Graph $\mathcal{G}$, resistance labelling $\mathcal{S}$, query node $s$}
	\KwOut{Resistance distance $r(s,u)$ for all $u\in\mathcal{V}$}
	$\column[u] \leftarrow 0$ for all $u\in\mathcal{V}$;\\
	$w \leftarrow s$;\\
	\While{$w \neq \mathcal{T}_{min}.\rootvector()$}{
		$\ratio \leftarrow \frac{\mathcal{S}[w, s]}{\mathcal{S}[w, w]}$;\\
		\For{each node $u\in\mathcal{T}_{min}.\subtree(w)$}{
			$\column[u] \leftarrow \column[u] + \mathcal{S}[w, u] \cdot \ratio$;\\
		}
		$w \leftarrow \mathcal{T}_{min}[w].\Parent()$;\\
	}
	\For{each node $u\in\mathcal{V}$}{
		$r(s,u) \leftarrow \mathcal{S}.\diagonal[s] + \mathcal{S}.\diagonal[u] - 2 \cdot \column[u]$;\\
	}
	\Return $r(s,u)$ for all $u\in\mathcal{V}$;
\end{algorithm}
\stitle{Single-source query processing.} Given a query source node $s$, a remarkable advantage of our method is its ability to efficiently handle exact single-source resistance distance queries. In contrast, existing exact methods \cite{LaplacianSolver} can only accomplish this task by solving $n-1$ separate linear systems, which is computationally expensive.

Algorithm~\ref{algo:single-source-query-processing} details our efficient approach for single-source queries. The key insight is that we only need to compute the $s$-th column of $\mathbf{L}_v^{-1}$ to determine resistance distances from source node $s$ to all other nodes in the graph. This approach leverages the tree decomposition structure to efficiently traverse from the source node to the root (lines 3-7), accumulating resistance contributions along the way. After computing the column vector, we calculate the final resistance distances for all nodes (lines 8-9) using the diagonal values and the computed column. By computing this single column using the distance labelling, we dramatically reduce the computational complexity compared to traditional methods.
\begin{lemma}[Time Complexity of Single-Source Queries]
	Algorithm~\ref{algo:single-source-query-processing} has a time complexity of $n+\sum_{u\in\mathcal{P}_{s\leadsto root}}|\subtree(u)|$, where $\mathcal{P}_{s\leadsto root}$ is the path from $s$ to the root and $|\subtree(u)|$ is the size of the subtree rooted at $u$ in the tree decomposition. This complexity can further be bounded by $O(n\cdot h_{\mathcal{G}})$.
  \begin{proof}
    The time complexity of Algorithm~\ref{algo:single-source-query-processing} can be derived by analyzing the algorithm's operations: First, the algorithm traverses the path from node $s$ to the root, which has length at most $h_{\mathcal{G}}$. At each node $w$ along this path, it updates values for all nodes in $\subtree(w)$, requiring $|\subtree(w)|$ operations. Thus, the total cost for these updates is $\sum_{u\in\mathcal{P}_{s\leadsto root}}|\subtree(u)|$. Additionally, the final loop (Lines 8-9) computes resistance distances for all $n$ nodes, contributing an $O(n)$ term. In the worst case, when the tree is highly unbalanced, each subtree could contain up to $O(n)$ nodes, and the path length could be $O(h_{\mathcal{G}})$, resulting in an upper bound of $O(n \cdot h_{\mathcal{G}})$ for the time complexity.
  \end{proof}
\end{lemma}
Note that although invoking the single-pair query algorithm (Algorithm~\ref{algo:single-pair-query-processing}) for each node in $\mathcal{V}$ to answer a single-source query also yields a complexity of $O(n \cdot h_{\mathcal{G}})$, its precise query complexity is $\sum_{u\in\mathcal{V}}|\mathcal{P}_{s\leadsto root}| = \sum_{u\in\mathcal{V}}|\subtree(u)|$ (due to variations in counting subtree sizes within a tree decomposition). This is strictly greater than the complexity $\sum_{u\in\mathcal{P}_{s\leadsto root}}|\subtree(u)|$ of Algorithm~\ref{algo:single-source-query-processing}. Experimental results demonstrate that the actual query time of Algorithm~\ref{algo:single-source-query-processing} is at least an order of magnitude faster than executing Algorithm~\ref{algo:single-pair-query-processing} for $n$ times (see Section~\ref{sec:experiment}).
\begin{table*}[]
	\vspace{-0.4cm}
	  \caption{Comparison of the complexities with existing methods for resistance distance computation}
	  \vspace{-0.3cm}
	  \centering
	  \scalebox{0.8}{
	  \begin{tabular}{|l|l|l|l|l|l|l|}
	  \hline
	  \textbf{Method} & \textbf{Category} & \textbf{Quality} & \textbf{Index building time} & \textbf{Index space} & \textbf{Single-pair query time} & \textbf{Single-source query time} \\
	  \hline
	  \lapsolver \cite{LaplacianSolver} & online & exact & - & - & $\widetilde{O}(m)$ & $\widetilde{O}(n\cdot m)$ \\
	  \hline
	  \geer \cite{ResistanceYang} & online & absolute error $\epsilon$ & - & - & $O(\frac{1}{\epsilon^2}\sigma^3)$, $\sigma=\log\left(\frac{1/(\epsilon(1-\lambda))}{1/\lambda}\right)$ & $O(\frac{n}{\epsilon^2}\sigma^3)$ \\
	  \hline
	  \bipush \cite{22resistance} & online & absolute error $\epsilon$ & - & - & $O(\frac{1}{\epsilon^2}\sigma_v^3)$, $\sigma_v=\log\left(\frac{1/(\epsilon(1-\lambda_v))}{1/\lambda_v}\right)$ & $O(\frac{n}{\epsilon^2}\sigma_v^3)$ \\
	  \hline
	  \leindex \cite{23resistance} & index-based & absolute error $\epsilon$ & $\widetilde{O}(\frac{1}{\epsilon^2}n)$ with assumptions & $O(n\cdot|\mathcal{V}_l|)$ & $O(\frac{1}{\epsilon^2}\sigma_{\mathcal{V}_l}^3)$, $\sigma_{\mathcal{V}_l}=\log\left(\frac{1/(\epsilon(1-\lambda_{\mathcal{V}_l}))}{1/\lambda_{\mathcal{V}_l}}\right)$ & $O(n+\frac{1}{\epsilon^2}\sigma_{\mathcal{V}_l}^3)$ \\
	  \hline
	  \treeindex (ours) & index-based & exact & $O(n\cdot h_{\mathcal{G}}^2\cdot d_{max})$ & $O(n \cdot h_{\mathcal{G}})$ & $O(h_{\mathcal{G}})$ & $O(n\cdot h_{\mathcal{G}})$ \\
	  \hline
	  \end{tabular}}
	  \label{tab:complexity}
	\end{table*}
	  
	\begin{table*}[]
	\vspace{-0.1cm}
	  \caption{Comparison of the complexities with tree decomposition-based shortest path distance computation method}
	  \vspace{-0.3cm}
	  \centering
	  \scalebox{0.8}{
	  \begin{tabular}{|l|l|l|l|l|}
	  \hline
	  \textbf{Problem} & \textbf{Method} & \textbf{Index building time} & \textbf{Index size} & \textbf{Query time} \\
	  \hline
	  \multirow{3}{*}{shortest path distance} & \tedi \cite{TEDISIGMOD10} & $O(n^2+n\cdot m)$ & $O(n\cdot(tw(\mathcal{G}))^2)$ & $O((tw(\mathcal{G}))^2\cdot h_{\mathcal{G}})$ \\
	  \cline{2-5}
	  & \multihop \cite{LiJunChang2012exact} & $O(n^2+n\cdot m)$ & $O(\sum_{X_i\in\mathcal{X}}|X_i|)$, $|X_i|$ the size of $i$-th tree node & $O(tw(\mathcal{G})\cdot h_{\mathcal{G}})$ \\
	  \cline{2-5}
	  & \htwoh \cite{HopLabeling2018hierarchy} & $O(n\cdot h_{\mathcal{G}}\cdot tw(\mathcal{G}))$ & $O(n\cdot h_{\mathcal{G}})$ & $O(tw(\mathcal{G}))$ \\
	  \hline
	  resistance distance & \treeindex (ours) & $O(n\cdot h_{\mathcal{G}}^2\cdot d_{max})$ & $O(n\cdot h_{\mathcal{G}})$ & $O(h_{\mathcal{G}})$ \\
	  \hline
	  \end{tabular}}
	  \vspace{-0.3cm}
	  \label{tab:sp_vs_resistance}
	\end{table*}

\stitle{Extension \treeindex to dynamic graphs.} In this paper, we focus on static computation of resistance distance. Extending the proposed \treeindex to support graph updates is a challenging problem. Like many existing solutions for dynamic tree decomposition-based shortest path distance computation, such as \cite{RelativeSubboundednessSIGMOD22,DynamicHubLabelingICDE21}, for the graph structure, it is possible to recognize the nodes infected by the update and only update the labels of the infected nodes. However, unlike the labels for shortest path distance which are weights of edges, the labels of \treeindex are elements of the matrix $\mathbf{L}_{\mathcal{U}_i\mathcal{U}_i}^{-1}$ for different node sets $\mathcal{U}_i$. It is non-trivial to update such labels. Carefully designed matrix-based update formulas should be provided for efficiency, which is a promising direction for future work.

\subsection{Comparison with Existing methods}
\stitle{Comparison with resistance distance computation methods.} Here, we first compare the complexities of existing solutions and the proposed \treeindex for the problem of resistance distance computation, which are listed in Table \ref{tab:complexity}. The complexities of the existing methods are obtained from the original paper \cite{LaplacianSolver,ResistanceYang,22resistance,23resistance}. It can be seen that except the Laplacian solver-based methods which have a large query time that is near-linear to the number of edges, the query time of all other methods \geer, \bipush and \leindex are bounded by different parameters that characterize the property of graphs. Specifically, $\lambda$ ($\lambda_v$, $\lambda_{\mathcal{V}_l})$ is the spectral radius of the probability transition matrix $\mathbf{P}=\mathbf{D}^{-1}\mathbf{A}$ ($\mathbf{P}_v=\mathbf{D}_v^{-1}\mathbf{A}_v$, $\mathbf{P}_{\mathcal{V}_l}=\mathbf{D}_{\mathcal{V}_l}^{-1}\mathbf{A}_{\mathcal{V}_l}$), $\mathbf{A}_v$, $\mathbf{D}_v$ ($\mathbf{A}_{\mathcal{V}_l}$, $\mathbf{D}_{\mathcal{V}_l}$) denote the matrix obtained by removing the $v$-th row and $v$-th column (rows and columns corresponding to $\mathcal{V}_l$) from $\mathbf{A}$, $\mathbf{D}$. To the best of our knowledge, there is no exact way to formulate the relationship between treewidth and the functions of eigenvalues. Intuitively, they are opposite to each other on the same graph. The graphs with small treewidth are easily separable, meaning that the spectral radius becomes small (according to the well-known Cheeger inequality \cite{chung1997spectral}). Consequently, the random walk will mix slowly ($\sigma$ is large) as it takes longer to overcome the bottlenecks and become uniformly distributed. As a result, existing random walk-based methods, which often rely on fast mixing assumptions (the near-linear index building time of \leindex is also obtained under such assumptions), perform poorly on graphs with small treewidth. This is also observed in previous study \cite{23resistance}. The proposed approach \treeindex, whose complexities are characterized by tree height $h_{\mathcal{G}}$ (which are observed small in real-life road networks), performs significantly faster on such graphs.

\stitle{Comparison with shortest path distance computation methods.} As \treeindex applies the tree decomposition technique for the problem of resistance distance computation, we also summarize the complexities of existing tree decomposition-based shortest path computation methods in Table \ref{tab:sp_vs_resistance}, followed by \treeindex. We select three representative tree decomposition-based shortest path computation methods \tedi \cite{TEDISIGMOD10}, \multihop \cite{LiJunChang2012exact}, and \htwoh \cite{HopLabeling2018hierarchy}. A short introduction of these methods can be found in the related work (Section~\ref{sec:related-work}). It can be seen that their complexities are also bounded by the functions of $h_{\mathcal{G}}$, a constant that is observed small on real-life small treewidth graphs. \treeindex has similar complexities with the \sota tree decomposition-based shortest path computation method \htwoh, with a remarkable difference being the query time--$tw(\mathcal{G})$ is strictly smaller than $h_{\mathcal{G}}$. This gap in query complexity is due to the different properties of resistance distance and shortest path distance. For shortest path distance, $d(s,t)$ is only related to the distances between $s,t$ and the cut vertices. However, resistance distance $r(s,t)$ is related to all nodes on the paths from $s,t$ to the root in the tree decomposition. As a result, the query process must traverse from $s$ and $t$ to the root (bounded by $O(h_{\mathcal{G}})$), while \htwoh only needs to query labels stored in the tree node $\lca(s,t)$ (bounded by $O(tw(\mathcal{G}))$). As a result, whether there is a $O(tw(\mathcal{G}))$ query algorithm for resistance distance is a challenging open problem.

\stitle{Contributions of this paper compared to existing methods.} Compared to the existing methods for resistance distance computation and tree decomposition-based shortest path distance computation, the main novelty of this paper is summarized as follows:

(1) To the best of our knowledge, this work is the first to successfully leverage tree decomposition for efficient resistance distance computation, in contrast to existing solutions that primarily rely on Laplacian solvers or random walk-based approaches.

(2) This work presents a novel cut property for resistance distance with combination of Schur complement and Cholesky decomposition. Non-zero structure of the inverse Cholesky decomposition is connected with a specific tree decomposition. The proposed dependency property of resistance distance (a new closed form formula of $r(s,t)$) shows that $r(s,t)$ is only dependent on the labels stored in the paths from $s$ and $t$ to the root of the tree decomposition. Efficient label construction algorithm is proposed by integrating the rank-1 update formula and the tree decomposition structure.

(3) \treeindex is the first study capable of computing single-source resistance distance exactly on the \fullusa dataset (see Section~\ref{sec:experiment}). Such application is not possible with existing methods.

\section{Application: Robust Routing on Road Networks	}\label{sec:application}
The proposed \treeindex approach demonstrates significant performance advantages on graphs with small treewidth, particularly road networks. In this section, we illustrate an important concrete application in geo-spatial database: robust routing on road networks. Given two query nodes $s$ and $t$, robust routing aims to provide $k$ alternative paths that satisfy two criteria: (i) each alternative path should be relatively short, and (ii) the set of alternative paths should be robust, meaning that the inaccessibility of any single edge should not disrupt all provided paths. For example, when an accident occurs on a highway segment along the shortest path, navigation systems should offer viable alternative routes. Electrical flow has been shown to perform effectively in this context \cite{RobustRouting21}. The \resistancedistance (resistance distance)-based robust routing method first computes the electrical flow between a source node $s$ and a target node $t$, and then generates alternative paths according to the electrical flow.

We first show that how the proposed \treeindex can efficiently compute the electrical flow between a source node $s$ and a target node $t$. Recall that Algorithm~\ref{algo:single-source-query-processing} addresses single-source queries originating from node $s$ by computing the $s$-th column of $\mathbf{L}_v^{-1}$ using the resistance distance labelling. By a straightforward modification of Algorithm~\ref{algo:single-source-query-processing}, we can compute the electrical flow between a source node $s$ and a target node $t$ by calculating both the $s$-th and $t$-th columns of $\mathbf{L}_v^{-1}$, as formalized in the following lemma:
\begin{lemma}
	The eletrical flow, when a unit current comes in through node $s$ and comes out through $t$, can be represented as:
	\begin{equation}
		\nabla \mathbf{f}=\mathbf{L}^\dagger(\mathbf{e}_s-\mathbf{e}_t)=\left[\begin{matrix}
		  \mathbf{L}_v^{-1} & 0 \\ 0 & 0
		\end{matrix}\right](\mathbf{e}_s-\mathbf{e}_t),
	\end{equation}
	where the eletrical flow along edge $(e_1,e_2)$ is $\nabla\mathbf{f}_{e_1}-\nabla\mathbf{f}_{e_2}$.
	\begin{proof}
	According to \cite{22resistance}, $\left[\begin{matrix}
		\mathbf{L}_v^{-1} & 0 \\ 0 & 0
	  \end{matrix}\right]$ is a g-inverse of $\mathbf{L}$. The lemma can be established according to the property of g-inverse \cite{22resistance} and given that $\mathbf{e}_s-\mathbf{e}_t$ is orthogonal to $\vec{1}$.
	\end{proof}
\end{lemma}
\begin{figure}[t!]
	\vspace*{0cm}
	\begin{center}
		\includegraphics[width=0.72\columnwidth, height=3.6cm]{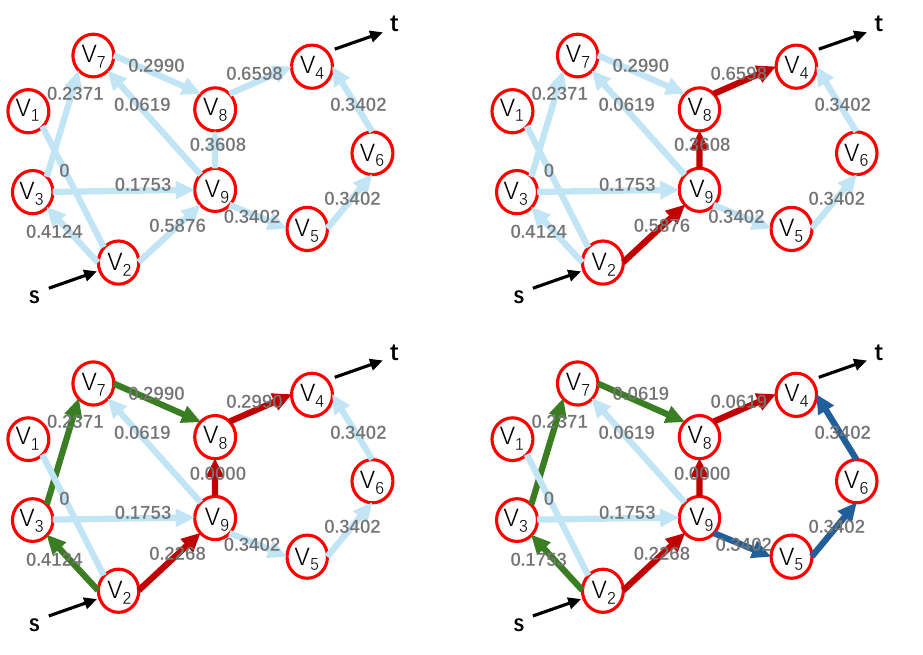}
	\end{center}
	\vspace*{-0.2cm}
	\caption{An illustrative example of robust routing on road networks using \resistancedistance-based method}
	\vspace*{-0.2cm}
	\label{fig:example-robust-routing}
\end{figure}
After computing the electrical flow, we generate $k$ alternative paths between nodes $s$ and $t$ as follows: we iteratively identify a path that maximizes the minimum flow from $s$ to $t$, remove the corresponding flow weights from the graph, and repeat this process to obtain additional paths. Fig.~\ref{fig:example-robust-routing} illustrates this process with an example. Initially, we identify the red path, which has the minimum flow value of $0.2608$, maximizing the flow across all possible paths from $s$ to $t$. Subsequently, we select the green and blue paths.

The \resistancedistance-based routing method can outperform the state-of-the-art robust routing methods \cite{PenaltyMethod,PlateauMethod} in terms of both routing time and path quality, as illustrated in our experiments (see Section~\ref{sec:experiment}). Thus, \treeindex has the potential to be applied to real-life applications.

\begin{table}[t]\vspace*{-0cm}
	\small
  \centering
  \caption{\small Datasets ($n$: number of nodes; $m$: number of edges; $d_{max}$: maximum degree; $h_{\mathcal{G}}$: tree height; $tw(\mathcal{G})$: tree width; $\frac{\# nnz}{n}$: average number of non-zero labels per node)} \label{tab:datasets}
  \vspace{-0.2cm}
  \resizebox{\columnwidth}{!}{
  \begin{tabular}{cccccccc}
	\toprule
	\bf Type & \bf Dataset & $n$ & $m$ & $d_{max}$ & $h_{\mathcal{G}}$ & $tw(\mathcal{G})$ & $\frac{\# nnz}{n}$ \cr\midrule
	\multirow{3}{*}{Social} & \emailenron & 33,696 & 180,811 & 1,383 & 2,397 & 2,258 & 1,895 \cr
	& \amazon & 334,863 & 925,872 & 549 & 21,394 & 18,462 & 18,770 \cr
	& \dblp & 317,080 & 1,049,866 & 343 & 20,109 & 19,322 & 17,582 \cr\midrule
	\multirow{7}{*}{Road} & \newyork & 264,346 & 365,050 & 8 & 767 & 113 & 346 \cr
	& \roadpa & 1,090,920 & 1,541,898 & 9 & 2,038 & 396 & 1,085 \cr
	& \roadtx & 1,393,383 & 1,921,660 & 12 & 1,389 & 308 & 760 \cr
	& \roadca & 1,971,281 & 2,766,607 & 12 & 1,821 & 470 & 1,237 \cr
	& \western & 6,262,104 & 7,559,642 & 9 & 1,489 & 263 & 1,016 \cr
	& \roadctr & 14,081,816 & 16,933,413 & 8 & 2,982 & 602 & 1,898 \cr
	& \fullusa & 23,947,348 & 28,854,312 & 9 & 3,976 & 642 & 2,268 \cr
	\bottomrule
  \end{tabular}}
  \vspace*{-0.4cm}
  \end{table}
\section{Experiments}\label{sec:experiment}
\subsection{Experimental Setup}
\stitle{Datasets.} We carefully selected 10 large-scale networks with diverse properties. Among them, 3 are social or collaboration networks exhibiting fast mixing properties, while the remaining 7 are road networks characterized by small treewidth. All datasets are publicly available from the \snap repository \cite{snapnets} and the \dimacs road network challenge \cite{dimacs}. We remove the node weights and the duplicate edges to build undirected, unweighted graphs, since previous methods \cite{ResistanceYang,22resistance,23resistance} are designed for such graphs. However, our methods also support weighted graphs, as we show in the case study in Section~\ref{sec:application}. Table~\ref{tab:datasets} presents comprehensive statistics of these datasets, including the number of nodes $n$, edges $m$, maximum degree $d_{max}$, tree height $h_{\mathcal{G}}$, treewidth $\tw(\mathcal{G})$ (computed using the minimum degree heuristic), and the average number of non-zero labels per node $\frac{\# nnz}{n}$. It can be seen that, on road networks with $20$ million edges, the tree height is relatively small (within $4,000$) and the real number of non-zero labels per node is even smaller. However, on social netwroks, $h_{\mathcal{G}}$ can exceed $20,000$ in a graph with around $1$ million edges. For our experimental evaluation, we generated $1,000$ random node pairs for single-pair query. Similarly, for single-source queries, we selected $100$ random source nodes per dataset. For all query sets, we establish the ``ground truth'' using the proposed \treeindex method. It should be noted that, while our method is theoretically exact, minor discrepancies may arise due to floating-point precision errors during computation. We assess the precision-related issues of \treeindex in Section~\ref{subsec:performance} (see \textbf{Exp III}).

\stitle{Comparison Methods.} We implement Algorithm~\ref{algo:index-building} to build our proposed \treeindex, which supports both single-pair (Algorithm~\ref{algo:single-pair-query-processing}) and single-source (Algorithm~\ref{algo:single-source-query-processing}) queries. We compare \treeindex against state-of-the-art exact and approximate approaches. For \textit{exact methods}, we benchmark against \lapsolver \cite{LaplacianSolver}, which employs advanced Laplacian solver techniques to compute resistance distances. This method first constructs an approximate Cholesky decomposition of the Laplacian matrix $\mathbf{L}$ as a preconditioner, then applies preconditioned conjugate gradient methods to solve $\mathbf{L}\mathbf{x}=\mathbf{e}_s-\mathbf{e}_t$, after which the resistance distance is computed as $\mathbf{x}_s-\mathbf{x}_t$. We set a tolerance of $10^{-9}$ to ensure precise results. For exact single-source queries, existing methods require solving $n$ separate linear systems. We also include a baseline \singlepairn by invoking our single-pair query algorithm for $n$ times. For \textit{approximate methods}, we compare with leading \textit{online} approaches \bipush \cite{22resistance} and \geer \cite{ResistanceYang} for single-pair queries and \lewalk \cite{22resistance} for single-source queries, as well as the state-of-the-art \textit{index-based} method \leindex \cite{23resistance}. Both \bipush and \geer utilize random walk sampling, while \lewalk employs loop-erased walk sampling. \leindex constructs an index by using spanning forest sampling to create Schur complement approximations with respect to a landmark node set $\mathcal{V}_l$, and employs \bipush (\push) to compute $\mathbf{L}_{\mathcal{U}\mathcal{U}}^{-1}$ during single-pair (single-source) query processing. Following previous work \cite{23resistance}, we set $|\mathcal{V}_l|=100$ and select landmarks using the highest degree heuristic. For all approximate methods, we use parameter $\epsilon=0.1$ to control accuracy by default, following previous studies \cite{23resistance,ResistanceYang,22resistance}.

\stitle{Experimental Environment.} All experiments are conducted on a Linux server with Intel Xeon E5-2680 v4 CPU and 512GB memory. All the algorithms are implemented in C++ and compiled with GCC 7.5.0. For comparison methods, we use the original C++ implementations provided by the authors \cite{ResistanceYang,RCHOL21,22resistance,23resistance}.
\subsection{Performance Evaluation}\label{subsec:performance}
\begin{figure}[t]
	\vspace*{0cm}
	\begin{center}
	\includegraphics[width=0.95\columnwidth]{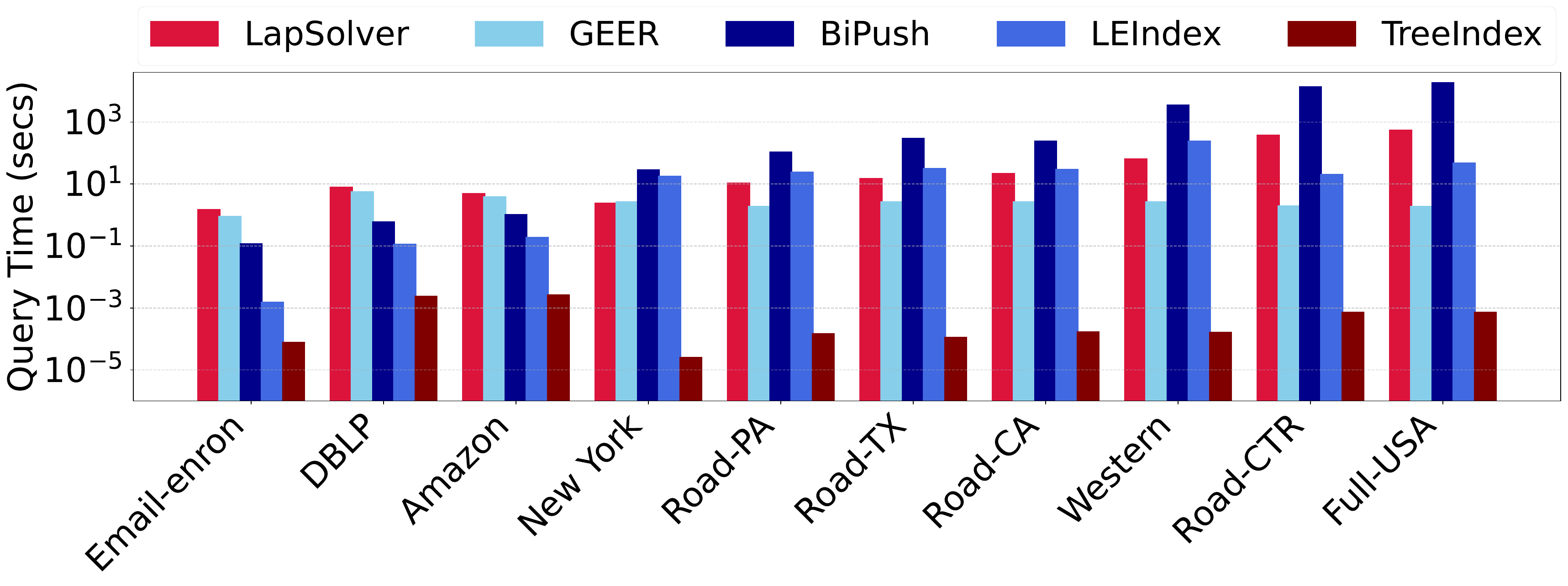}
	\end{center}
	\vspace*{-0.5cm}
	\caption{Processing time of single-pair query}
	\vspace*{-0.4cm}
	\label{fig:single-pair-query-time}
\end{figure}
\begin{figure}[t]
	\vspace*{0cm}
	\begin{center}
	\includegraphics[width=0.95\columnwidth]{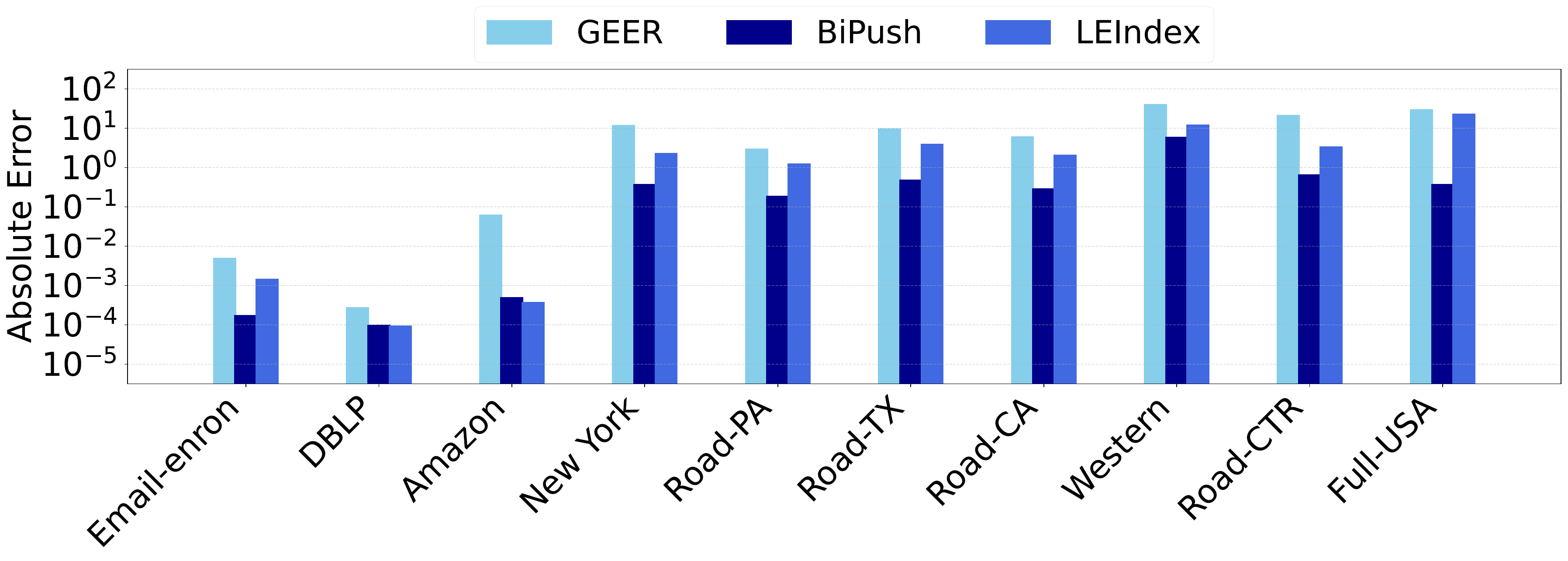}
	\end{center}
	\vspace*{-0.5cm}
	\caption{Absolute error of single-pair query for approximate methods}
	\vspace*{-0.4cm}
	\label{fig:single-pair-query-absolute-error}
\end{figure}
\stitle{Exp I: Overall Query Processing Performance.} We first evaluate the query processing performance of our method in comparison with state-of-the-art exact and approximate methods. The results for single-pair query processing are illustrated in Fig.~\ref{fig:single-pair-query-time}. Additionally, for the approximate methods \geer, \bipush, and \leindex, we present the average absolute error results in Fig.~\ref{fig:single-pair-query-absolute-error}. Letting $\hat{r}(s,t)$ denote the query result from approximate methods, we define the absolute error as $|\hat{r}(s,t)-r(s,t)|$ to measure query accuracy. The results clearly show that \treeindex achieves the fastest query times across all datasets. For the \emailenron, \amazon, and \dblp networks, which exhibit relatively large tree widths, the query efficiency order from fastest to slowest is: \treeindex, followed by the index-based approximate method \leindex, the online approximate method \bipush, \geer, and the exact method \lapsolver. Remarkably, despite providing exact results, \treeindex is at least one order of magnitude faster than \leindex. On road networks, random walk-based methods (\geer, \bipush, \leindex) exhibit even poorer performance than the exact method \lapsolver, characterized by prolonged query times or high relative errors due to the significant mixing times inherent to road networks. In contrast, \treeindex remains more than $3$ orders of magnitude faster than \lapsolver. Notably, on the largest road network, \fullusa, \treeindex achieves query times of approximately $7\times 10^{-4}$ seconds, while \lapsolver requires $525$ seconds. Furthermore, approximate methods require over $5,405$ seconds to achieve an absolute error of $10^{-1}$. This superior performance is primarily attributable to the small tree height of road networks, which enables our method to effectively leverage tree width for highly efficient query processing.

For single-source queries, we employ the average relative error, defined as $\frac{1}{n}\sum_{u\in\mathcal{V}}{|\hat{r}(s,u)-r(s,u)|}$, to measure the query accuracy of the approximate methods \lewalk and \leindex. We exclude any queries exceeding $10$ hours. The corresponding results are illustrated in Fig.~\ref{fig:single-source-query-time} and Fig.~\ref{fig:single-source-query-error}. Across all datasets, \lapsolver can compute results only for \emailenron within $10$ hours. \treeindex is at least an order of magnitude faster than \singlepairn. As observed, on road networks, \treeindex remains the fastest method, followed by the index-based approximate method \leindex, the online approximate method \lewalk, and finally the exact method \lapsolver. Notably, the average absolute error of \leindex is significantly high, primarily due to the substantial variance inherent in loop-erased walk sampling on road networks. Despite this, \treeindex still achieves query times at least an order of magnitude faster than \leindex and \lewalk. For example, on \fullusa, \treeindex has an average query time of approximately $190$ seconds, whereas \leindex requires around $3,776$ seconds to achieve an absolute error of only $4.2$.

\begin{figure}[t]
	\vspace*{-0.1cm}
	\begin{center}
	\includegraphics[width=0.95\columnwidth]{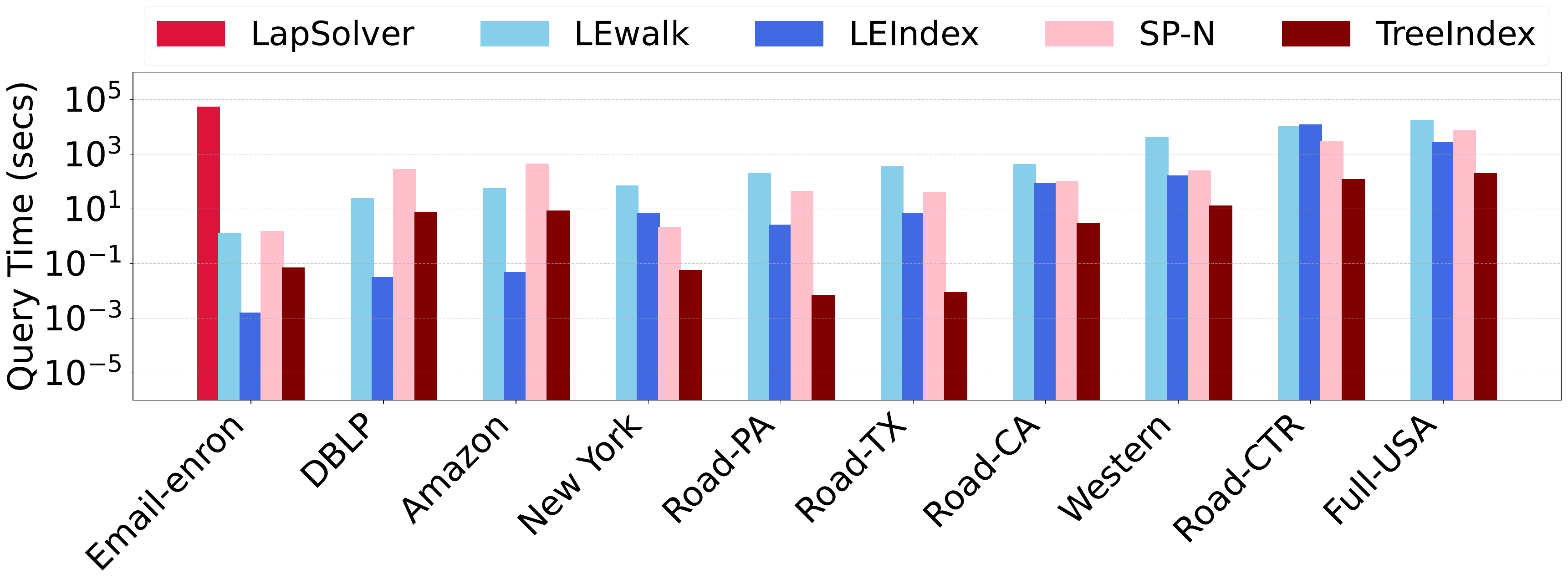}
	\end{center}
	\vspace*{-0.5cm}
	\caption{Processing time of single-source query}
	\vspace*{-0.4cm}
	\label{fig:single-source-query-time}
\end{figure}
\begin{figure}[t]
	\vspace*{-0.1cm}
	\begin{center}
	\includegraphics[width=0.95\columnwidth]{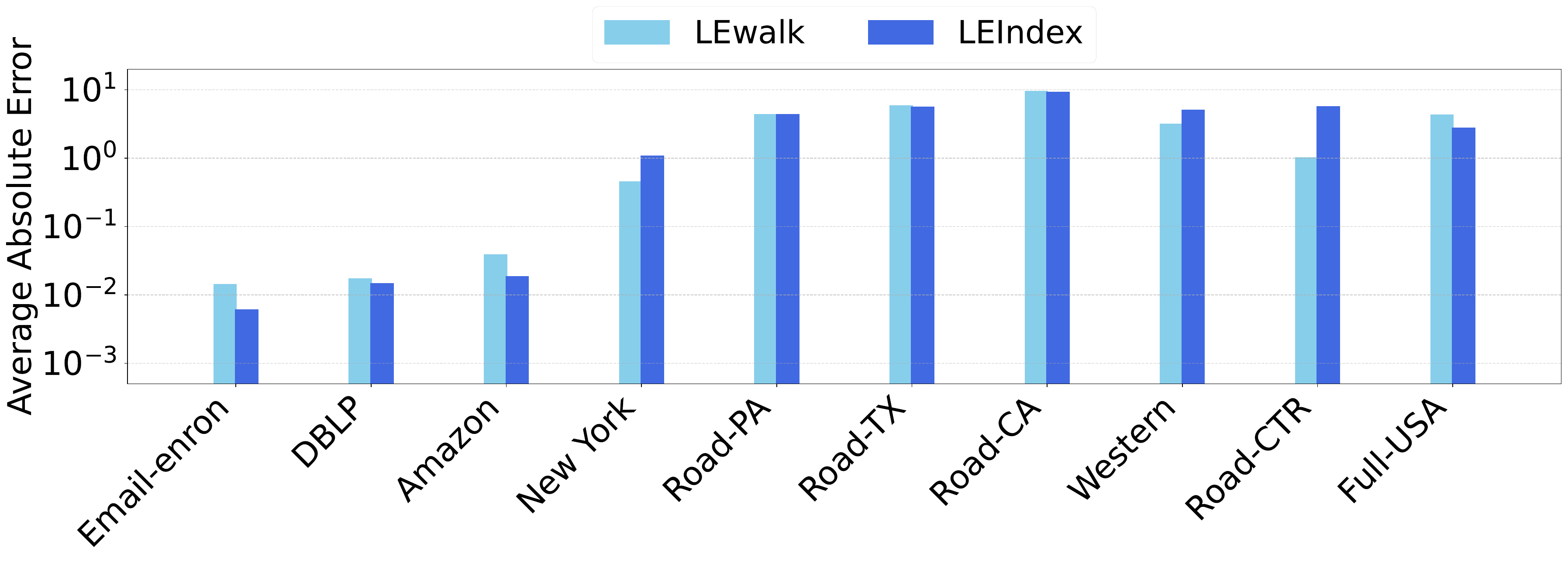}
	\end{center}
	\vspace*{-0.5cm}
	\caption{Average absolute error of single-source query for approximate methods}
	\vspace*{-0.2cm}
	\label{fig:single-source-query-error}
\end{figure}
\begin{table}[t]
	\centering
	\caption{Indexing performance analysis on all datasets}
	\label{tab:index-analysis}
	\vspace*{-0.2cm}
	\resizebox{\columnwidth}{!}{
	\begin{tabular}{lccccc}
	\toprule
	\multirow{2}{*}{\textbf{Dataset}} & \multicolumn{3}{c}{\textbf{Index Size (MB)}} & \multicolumn{2}{c}{\textbf{Construction Time (secs)}} \\
	\cmidrule(lr){2-4} \cmidrule(lr){5-6}
	& Graph size & \textbf{\treeindex} & \textbf{\leindex} & \textbf{\treeindex} & \textbf{\leindex} \\
	\midrule
	\emailenron & 2 & 487 $(244\times)$ & 5 $(3\times)$ & 1,136 & 2 \\
	\amazon & 12 & 47,952 $(3,996\times)$ & 51 $(4\times)$ & 363,099 & 243 \\
	\dblp & 13 & 42,536 $(3,272\times)$ & 48 $(4\times)$ & 521,969 & 154 \\
	\newyork & 5 & 698 $(140\times)$ & 40 $(8\times)$ & 7 & 515 \\
	\roadpa & 20 & 9,505 $(475\times)$ & 166 $(8\times)$ & 452 & 3,323 \\
	\roadtx & 26 & 8,929 $(343\times)$ & 206 $(8\times)$ & 257 & 4,960 \\
	\roadca & 39 & 17,418 $(447\times)$ & 299 $(8\times)$ & 814 & 7,583 \\
	\western & 113 & 48,538 $(430\times)$ & 955 $(8\times)$ & 1,567 & 65,392 \\
	\roadctr & 266 & 203,858 $(766\times)$ & 2,176 $(8\times)$ & 44,614 & 68,390 \\
	\fullusa & 470 & 414,392 $(882\times)$ & 3,735 $(8\times)$ & 77,323 & 115,523 \\
	\bottomrule
	\end{tabular}
	}
	\vspace*{-0.3cm}
\end{table}
\stitle{Exp II: Indexing Performance Anaylsis.} For the two index-based methods, \leindex and \treeindex, we analyze both their index sizes and index construction times. Note that both methods can support single-source queries by additionally storing diagonal entries of $\mathbf{L}_v^{-1}$ during index construction. We compare the index sizes of these methods, measured both in absolute terms and relative to the graph size. The results are presented in Table~\ref{tab:index-analysis}. It can be observed that the index size of \treeindex depends significantly on the graph's structure, specifically on its tree height $h_{\mathcal{G}}$. In contrast, the index size of \leindex is independent of graph structure, as it explicitly stores an $n\times|\mathcal{V}_l|$ matrix, where $\mathcal{V}_l$ is the landmark node set. Consequently, the index size of \leindex is approximately $8\times$ the graph size, whereas the index size of \treeindex can scale up to $4000\times$ the graph size on social networks and several hundred times the graph size on road networks. Despite this significant difference, the index sizes remain manageable in practice. For example, on the largest road network \fullusa, the index size reaches around $405$ GB, which can still be efficiently loaded onto a commodity server.

We also evaluate and compare the index construction times, with the results summarized in Table~\ref{tab:index-analysis}. The results reveal a significant performance variation between graphs with relatively large tree width and road networks. Specifically, the index construction time of \treeindex is substantially longer compared to \leindex on networks with large tree width. For example, constructing the resistance distance labelling on \dblp takes approximately $145$ hours (over $6$ days). Conversely, the index construction process is notably faster on road networks, even surpassing the performance of \leindex, which provides only approximate solutions. For instance, \treeindex constructs the index for \fullusa within $7$ hours, whereas \leindex requires more than $32$ hours.
\comment{
\begin{figure}[t]
	\vspace*{-0cm}
	\begin{center}
 \begin{tabular}[t]{c}
   \subfigure[\dblp]{
				\includegraphics[width=0.40\columnwidth]{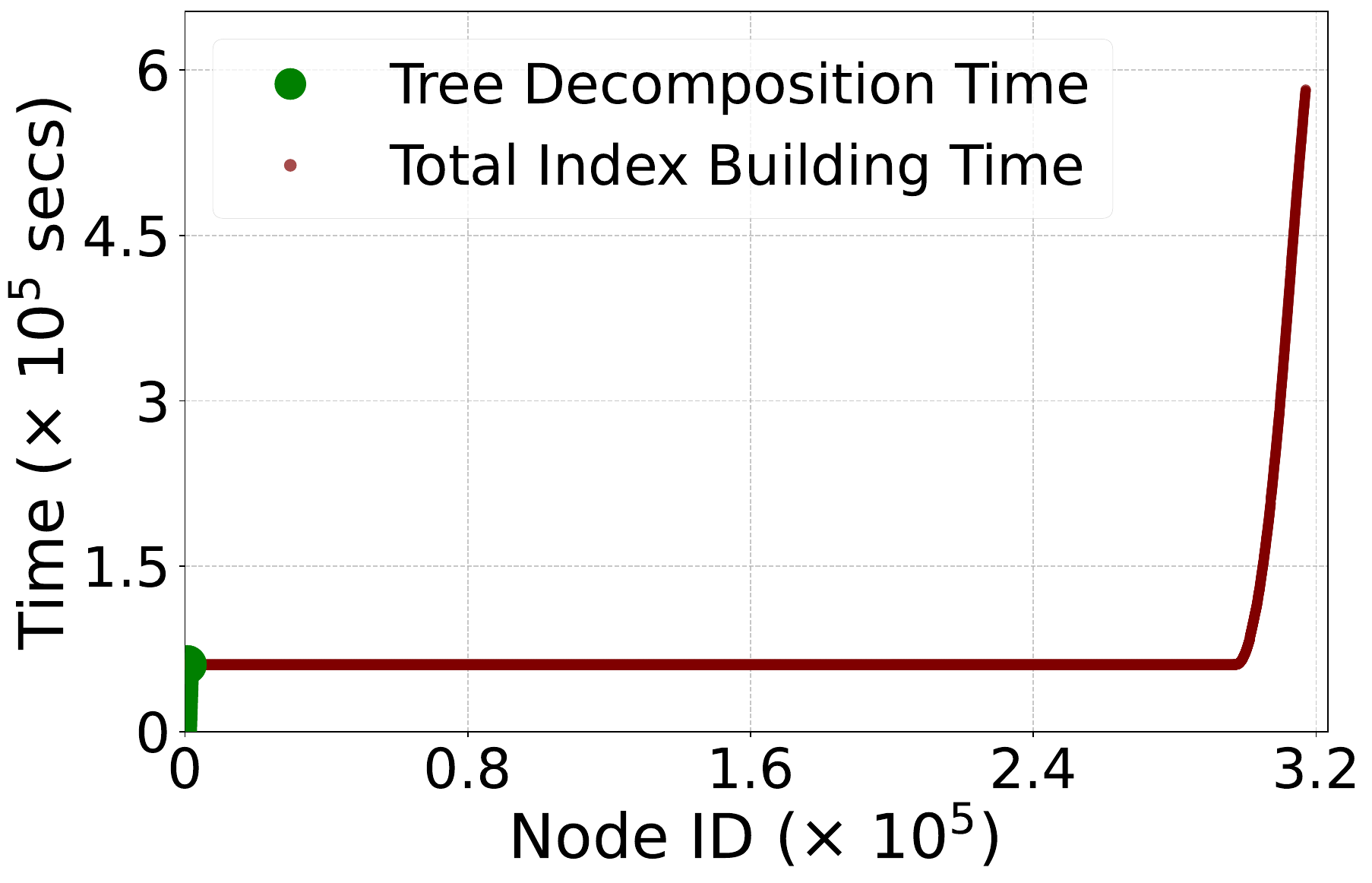}
			}
   \subfigure[\roadpa]{
				\includegraphics[width=0.38\columnwidth]{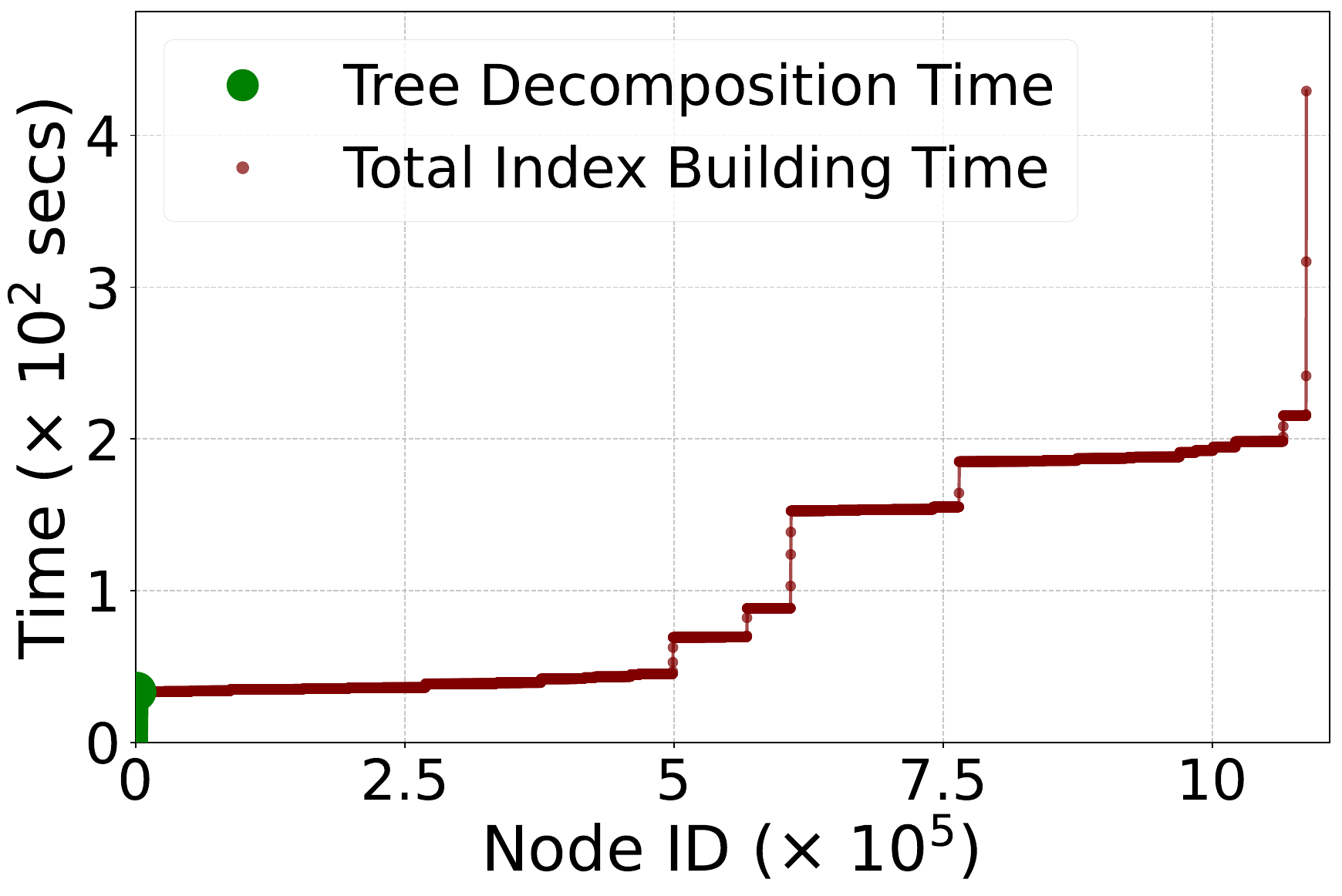}
			}
		\end{tabular}
	\end{center}
	\vspace*{-0.6cm}
	\caption{Distribution of index building time of \treeindex}
	\vspace*{-0.2cm}
	\label{fig:index-building-process}
\end{figure}}
\comment{
\stitle{Exp III: Distribution of Index Construction Time.} Recall that the index construction process (Algorithm~\ref{algo:index-building}) of \treeindex comprises two stages: (i) tree decomposition and preprocessing, and (ii) iterative bottom-up construction of the index. In this experiment, we further analyze the distribution of index construction time in \treeindex. We plot the cumulative index construction time against the number of nodes processed according to the \dfs ordering. The initial vertical rise observed in the plots is attributed to tree decomposition and preprocessing. The results for datasets \dblp and \roadpa are shown in Fig.~\ref{fig:index-building-process}; similar trends are observed on other datasets. We can observe that the time required for tree decomposition and preprocessing is relatively small compared to the total index construction time. Moreover, an interesting phenomenon can be observed: the index construction is initially very fast but dramatically slows for the last set of nodes. This sharp increase is particularly pronounced on social networks. Notably, although the entire construction process for \dblp requires approximately $145$ hours, constructing the first $90\%$ of nodes takes only about $1$ hour. This behavior arises from the bottom-up nature of our construction approach, wherein nodes with higher degrees are processed later.}

\begin{figure}[t]
	\vspace*{-0cm}
	\begin{center}
 \begin{tabular}[t]{c}
   \subfigure[\newyork]{
				\includegraphics[width=0.38\columnwidth]{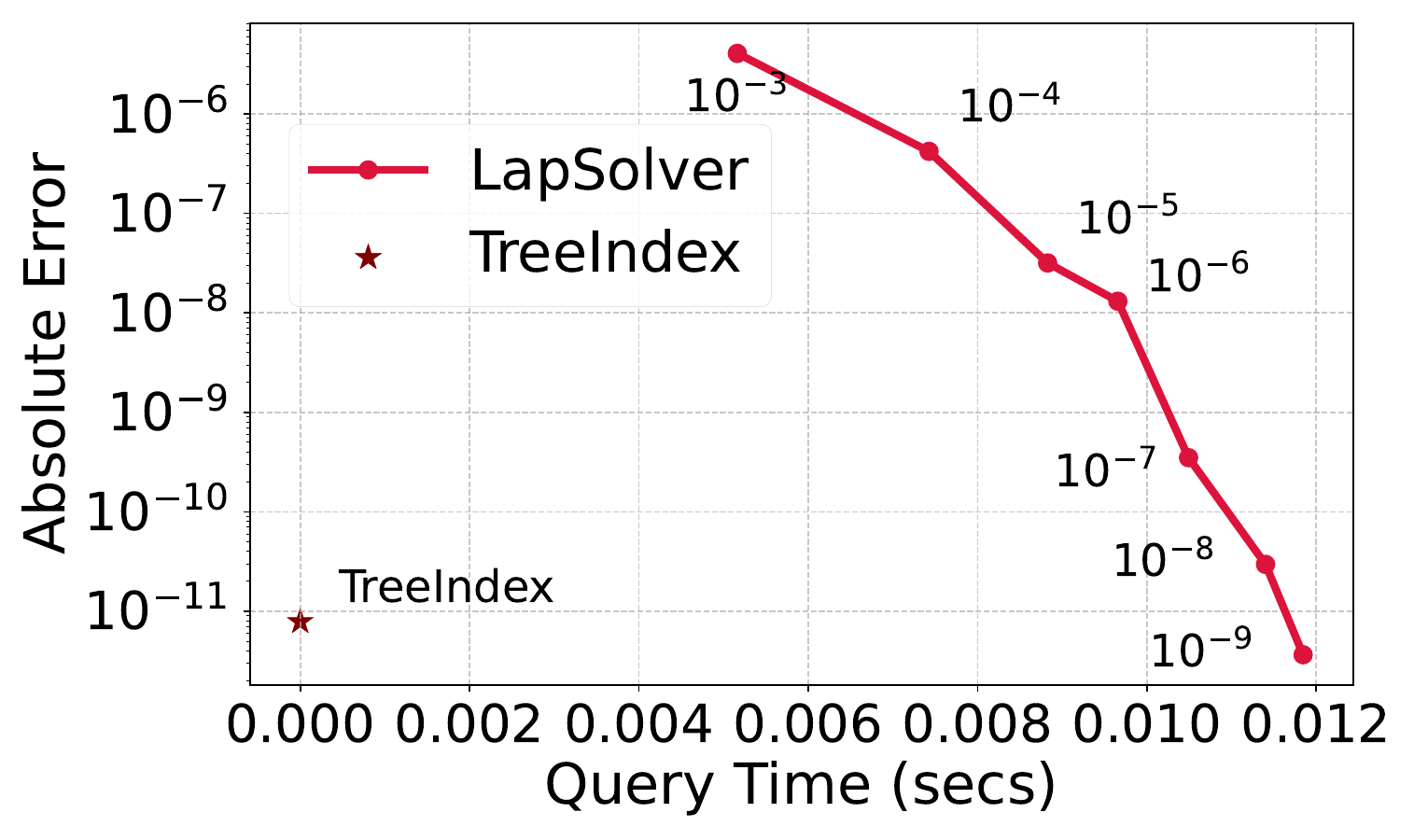}
			}
   \subfigure[\roadpa]{
				\includegraphics[width=0.38\columnwidth]{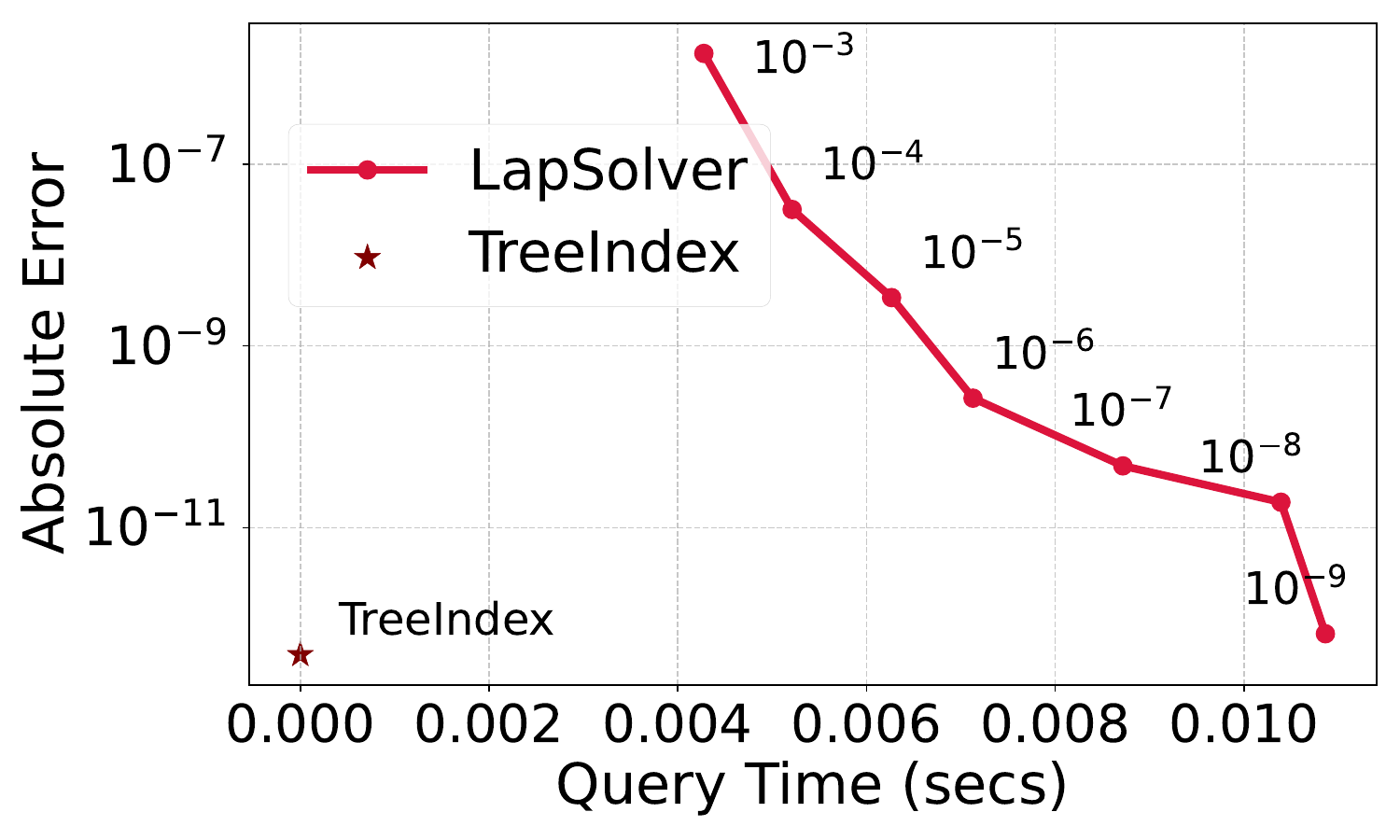}
			}
		\end{tabular}
	\end{center}
	\vspace*{-0.6cm}
	\caption{Precision analysis of \treeindex and \lapsolver}
	\vspace*{-0.2cm}
	\label{fig:precision-evaluation}
\end{figure}
\stitle{Exp III: Precision Analysis.} It should be noted that although \treeindex is theoretically exact, practical computations can still introduce minor errors due to floating-point precision limitations. In this experiment, we further analyze the numerical precision of \treeindex compared to \lapsolver. Specifically, we use the results obtained by \lapsolver with $\epsilon=10^{-19}$ as the ground truth, and vary the parameter $\epsilon$ in \lapsolver to compare the absolute errors against those of \treeindex. The results are presented in Fig.~\ref{fig:precision-evaluation}. As illustrated, \treeindex consistently achieves an absolute error smaller than $10^{-11}$, which is negligible in practical applications.

\begin{table}[t]
	\centering
	\caption{Performance comparison between \treeindex and Shortest Path Distance method \htwoh \cite{HopLabeling2018hierarchy}}
	\label{tab:comparison-shortest-path}
	\resizebox{\columnwidth}{!}{
	\begin{tabular}{llccc}
	\toprule
	\textbf{Dataset} & \textbf{Method} & \textbf{Query Time} & \textbf{Index Size} & \textbf{Construction Time} \\
	\midrule
	\multirow{2}{*}{\newyork} & \treeindex & $2.6\times10^{-5}$secs & $0.68$GB & 6.7 secs \\
	 & \htwoh & $1.3\times10^{-6}$secs & $0.38$GB & 2.3 secs \\
	\midrule
	\multirow{2}{*}{\roadpa} & \treeindex & $1.4\times10^{-4}$secs & $9.28$GB & 452.5 secs \\
	 & \htwoh & $4.3\times10^{-6}$secs & $4.93$GB & 21.9 secs \\
	\bottomrule
	\end{tabular}
	}
	\vspace{-0.4cm}
\end{table}
\stitle{Exp IV: Comparison with Shortest Path Distance Index.} Although resistance distance computation and shortest path distance computation are fundamentally different problems, we also compare \treeindex with the state-of-the-art shortest path distance labelling approach, \htwoh, which likewise utilizes tree decomposition and vertex hierarchy. We compare these methods in terms of query time, index size, and construction time, with the results summarized in Table~\ref{tab:comparison-shortest-path}. Results for \newyork and \roadpa are presented, and consistent patterns are observed across other datasets. It can be seen that our method achieves performance comparable to \htwoh, though slightly inferior. This minor performance gap primarily arises from the more complicated index construction process in our method, which involves numerical computations. Nevertheless, our method dramatically improves the query efficiency for resistance distance, making its computation practical on large road networks. Notably, before this study, the most advanced method required approximately $5,000$ seconds per query on \fullusa, as demonstrated in our experiments. Thus, our method effectively brings resistance distance computation into a practical realm for large road networks.

\begin{figure}[t!]
	\vspace*{-0.2cm}
	\begin{center}
 \begin{tabular}[t]{c}
  \vspace{-0.1cm}\subfigure[query time, road networks]{
				\includegraphics[width=0.42\columnwidth, height=2.4cm]{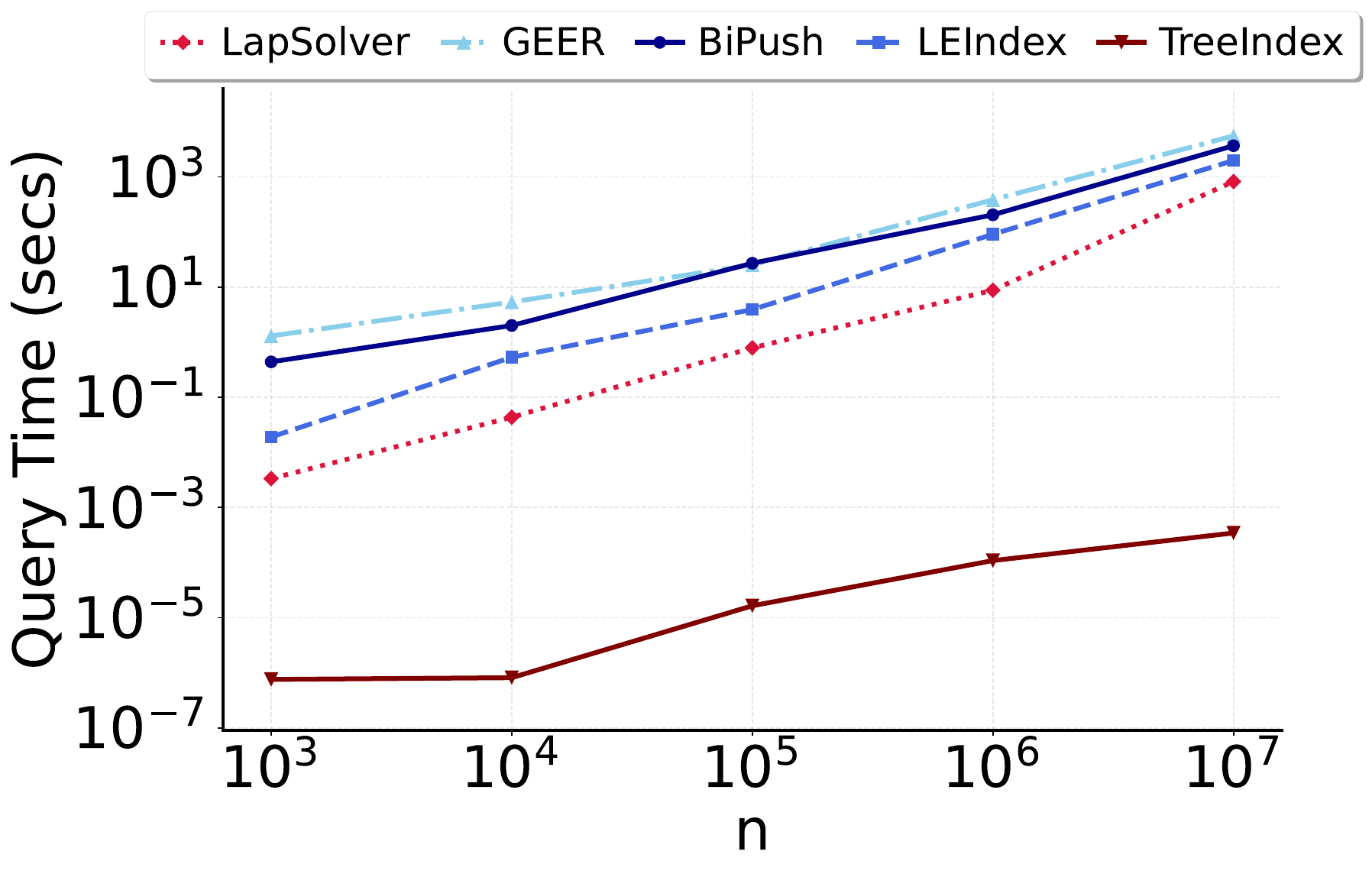}
      }\hspace{0.06\columnwidth}
      \vspace{-0.1cm}\subfigure[label construction time, road networks]{
				\includegraphics[width=0.42\columnwidth, height=2.4cm]{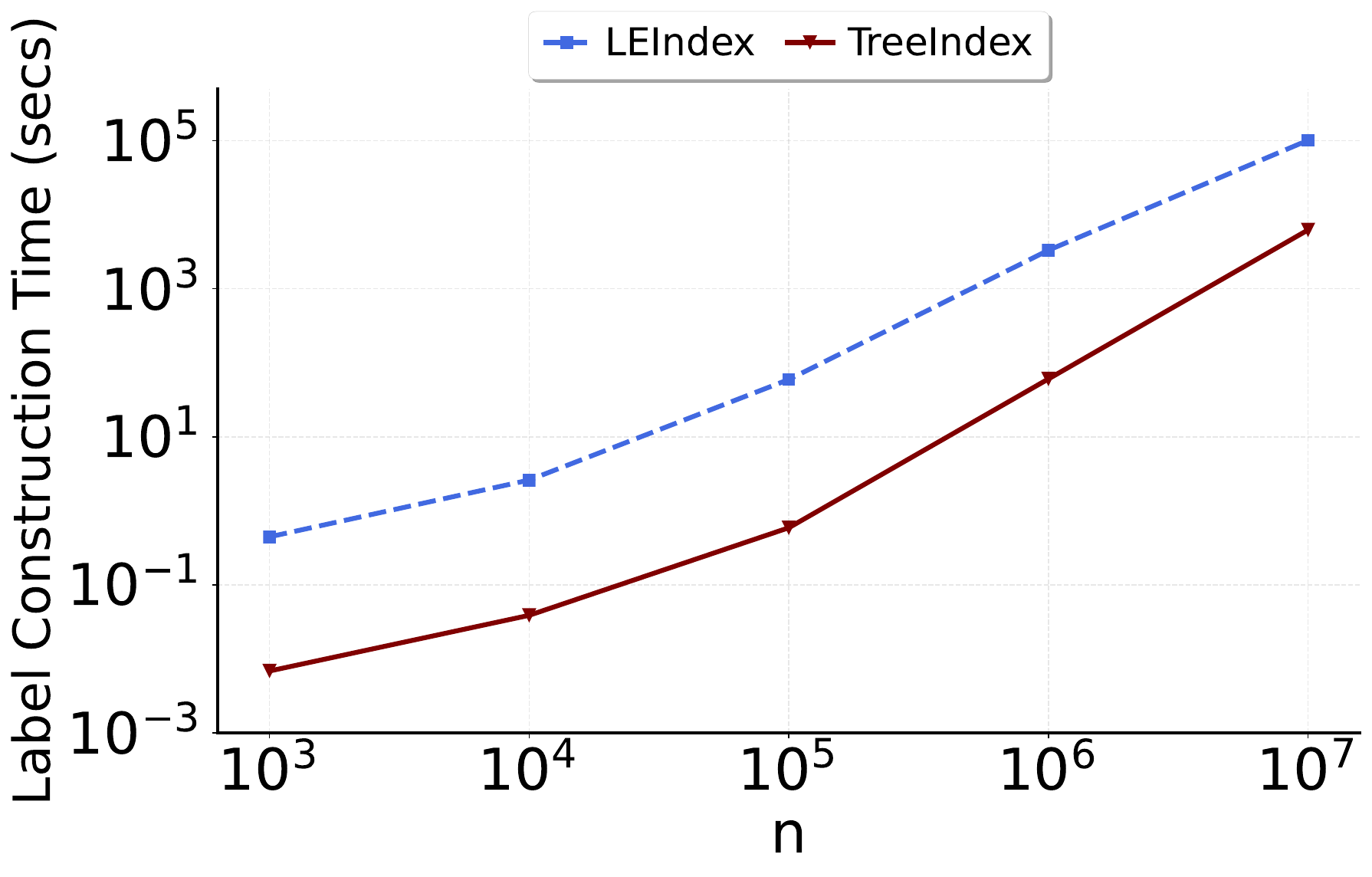}
      }\\
      \vspace{-0.1cm}\subfigure[query time, synthetic scale-free graphs]{
				\includegraphics[width=0.42\columnwidth, height=2.4cm]{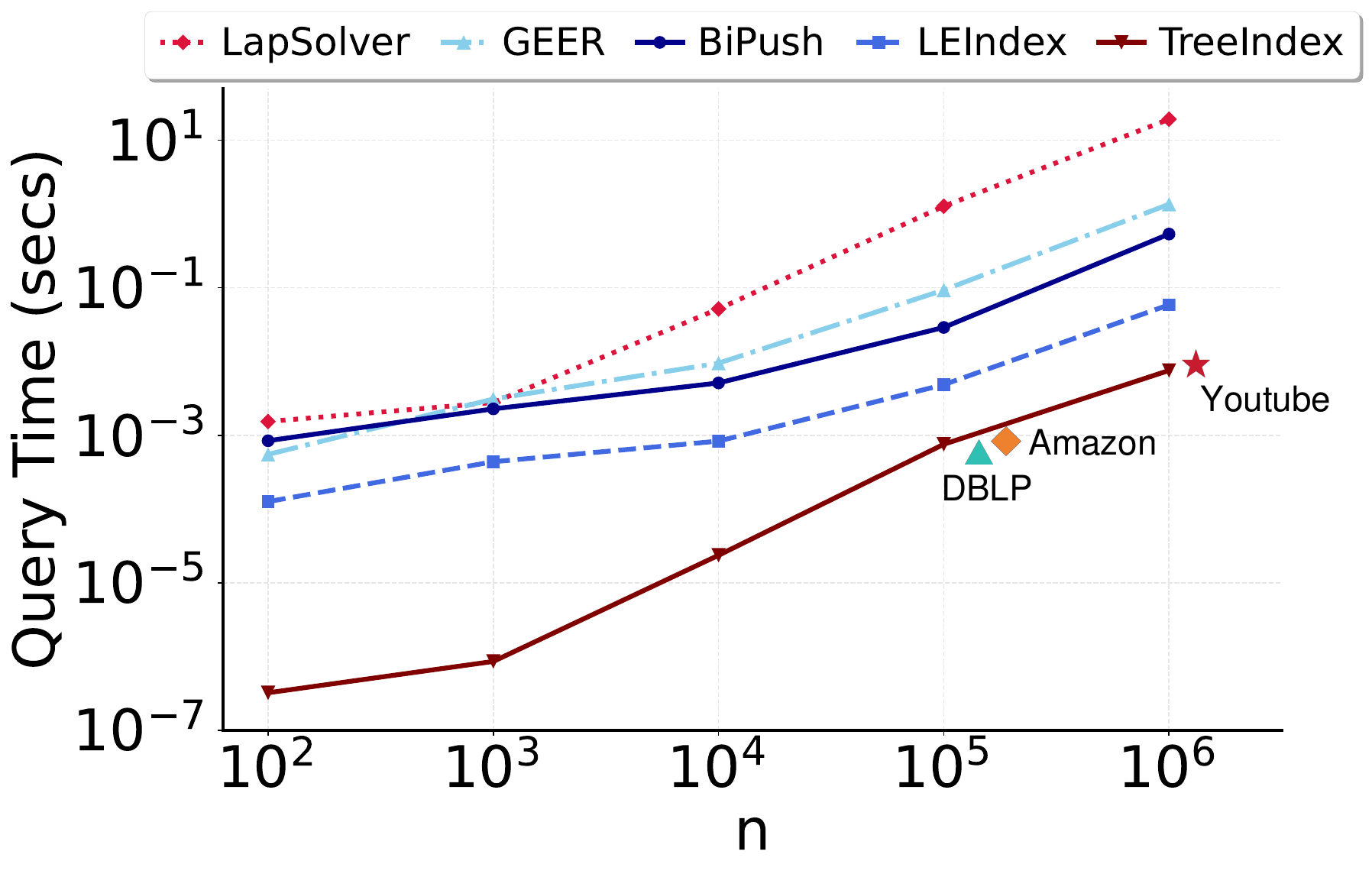}
			}\hspace{0.06\columnwidth}
      \vspace{-0.1cm}\subfigure[label construction time, synthetic scale-free graphs]{
				\includegraphics[width=0.42\columnwidth, height=2.4cm]{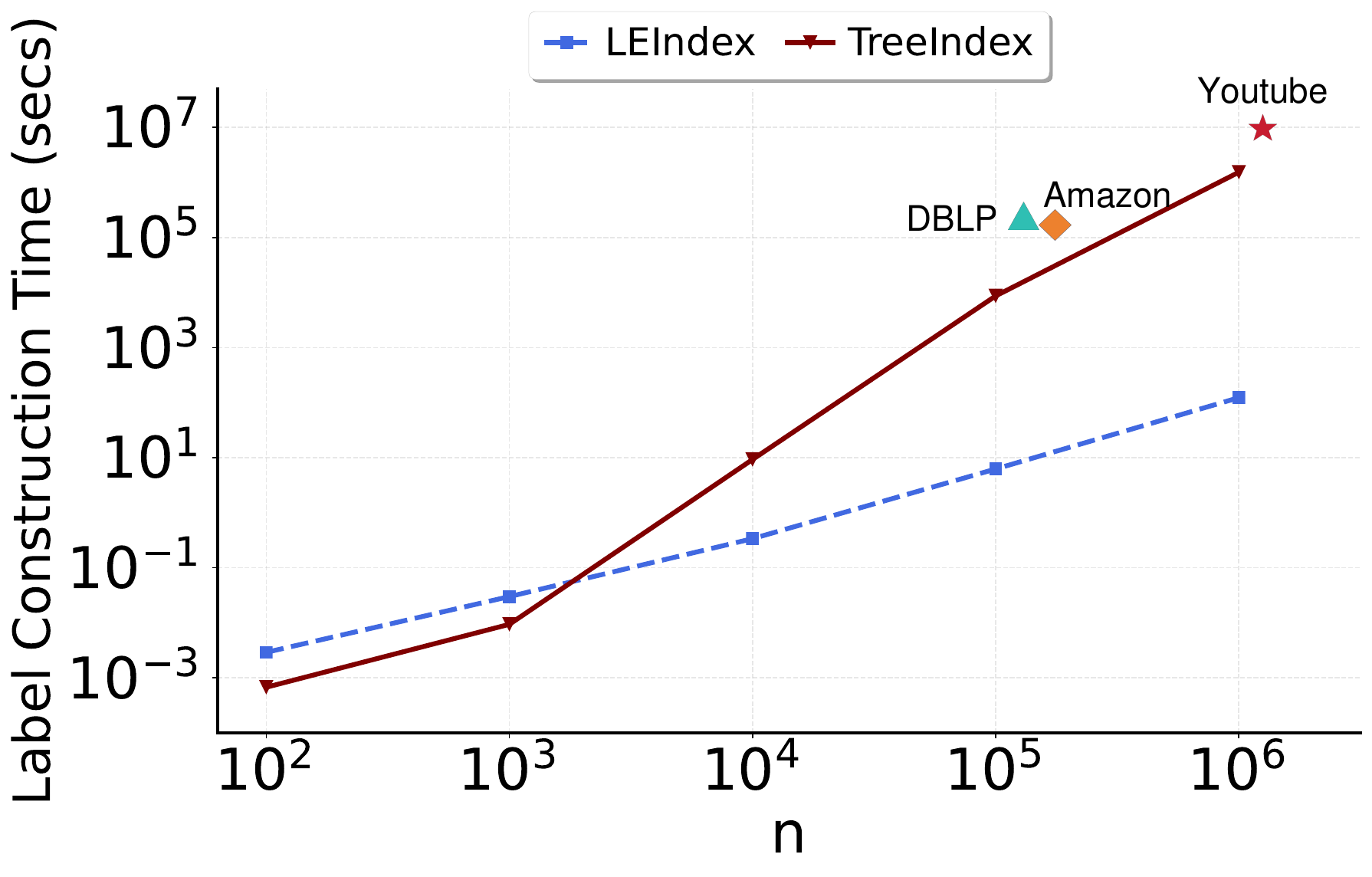}
			}
		\end{tabular}
	\end{center}
	\vspace*{-0.3cm}
	\caption{Scalability test on road networks and synthetic graphs}
	\vspace*{-0.4cm}
	\label{fig:scalability-test}
\end{figure}
\stitle{Exp-V: Scalability Test.} We conduct scalability tests on both real-world road networks and synthetic scale-free graphs. In this paper, we focus on the problem of resistance distance computation on small treewidth graphs such as road networks. Following previous studies for shortest path distance computation \cite{HopLabeling2018hierarchy,ProjectedVertexSeparator2021}, we first conduct scalability test by extracting different sizes of road networks from the OpenStreetMap dataset \cite{OpenStreetMap}. We vary the node size $n$ from $10^3$ to $10^7$ and compare the single-pair query time and label construction time of different methods. The results are shown in Fig.~\ref{fig:scalability-test} (a)-(b). It can be seen that the query time and label construction time both increase slowly when the node size increases. This validates that \treeindex is scalable on very large road networks. Then, we use the commonly-adpoted Chung-Lu model \cite{chung2002connected} to generate scale-free synthetic graphs with different sizes. We fix the power-law exponent $\gamma$ as $2.2$, and vary $n$ from $10^2$ to $10^6$. The results are shown in Fig.~\ref{fig:scalability-test} (c)-(d). We also plot the results of three real-life non-road networks. For graphs larger than \dblp (including \youtube \cite{snapnets} with $1$M nodes, $3$M edges), the results are estimated using the operation numbers obtained from tree decomposition, as stated in the time complexity analysis in Section~\ref{sec:resistance-distance-labelling}. It can be seen that although the query time is still very fast when the node size increases, the bottleneck of \treeindex is the label construction time which grows rapidly. When $n=10^5$, it costs $48$ hours ($2$ days) to construct the labels, and it is hard to generate labels for graphs with larger sizes. For \youtube, label construction will cost more than $140$ days by estimation. The same situation can be observed in the tree decomposition-based shortest path distance computation methods such as \htwoh \cite{HopLabeling2018hierarchy}, as the treeheight is relatively large. Thus, it is a promising future direction to incorporate \treeindex with other methods for a better performance on non-road graphs.

\begin{figure}[t!]
	\vspace*{-0.2cm}
	\begin{center}
 \begin{tabular}[t]{c}
      \subfigure[single-pair, query time]{
        \includegraphics[width=0.42\columnwidth, height=2.4cm]{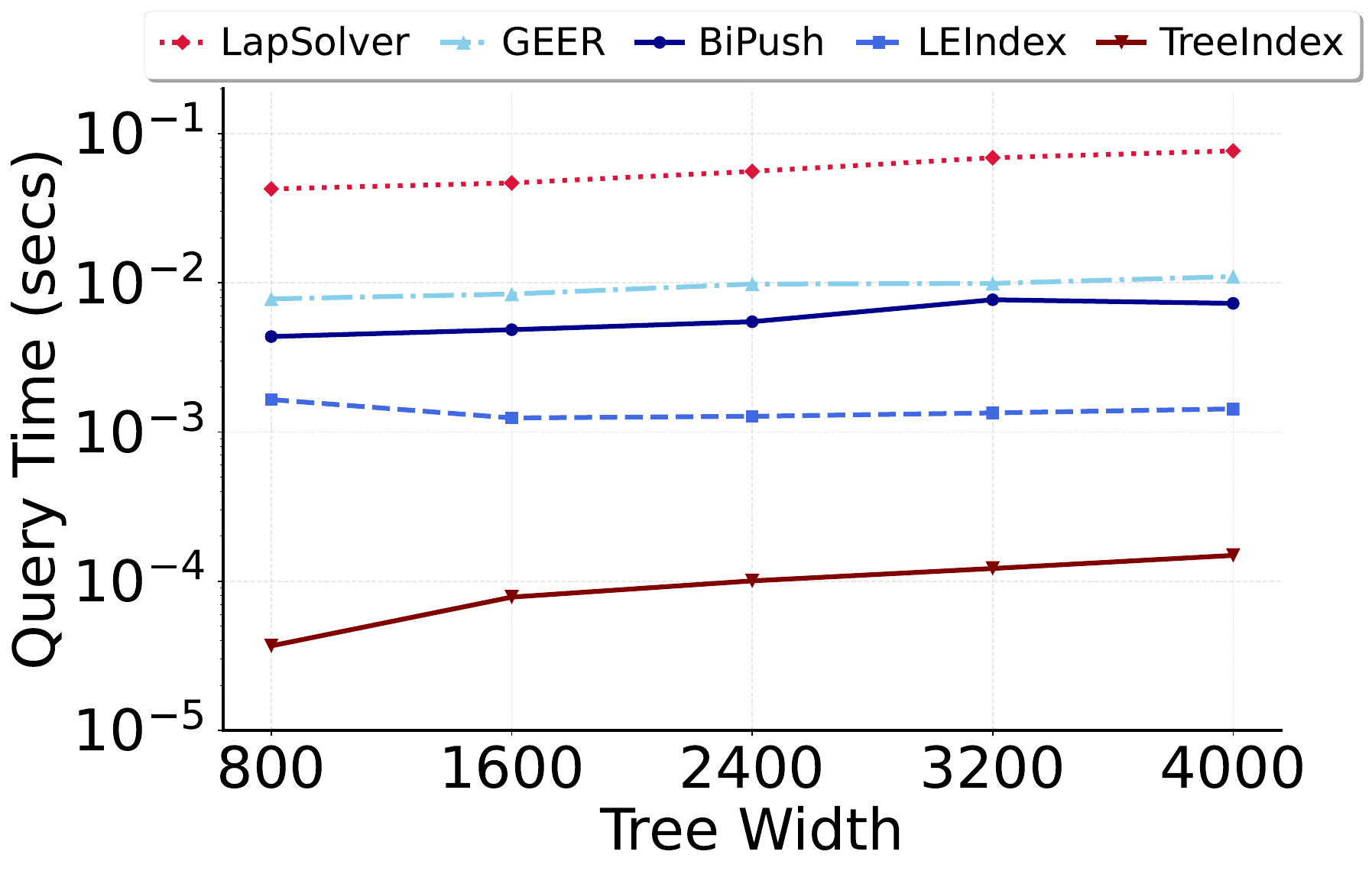}
      }\hspace{0.06\columnwidth}
      \subfigure[label construction time]{
        \includegraphics[width=0.42\columnwidth, height=2.4cm]{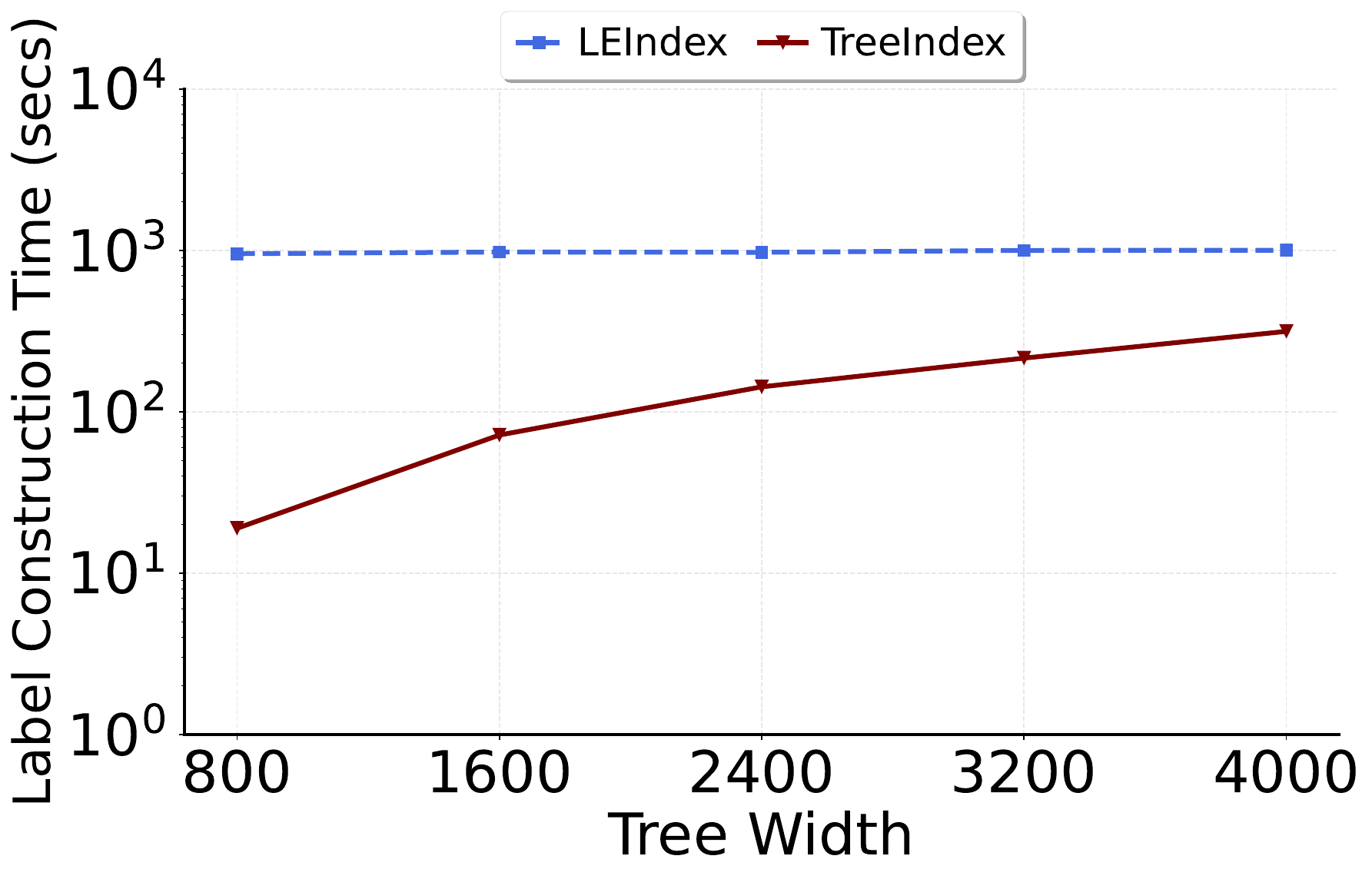}
      }
		\end{tabular}
	\end{center}
	\vspace*{-0.6cm}
	\caption{Performance of \treeindex when varying treewidth}
	\vspace*{-0.5cm}
	\label{fig:varying-treewidth}
\end{figure}
\stitle{Exp-VI: Performance when varying treewidth.} We also conduct experiments to see the performance of \treeindex when the treewidth varies. Specifically, we vary the parameter $\gamma$ of the Chung-Lu model to generate graphs with specific treewidth. We fix the node number as $10^4$, vary the treewidth from $800$ to $4000$ and compare the single-pair query time and label construction time of different methods. The results are shown in Fig.~\ref{fig:varying-treewidth}. It can be seen that the query time and label construction time grow when the treewidth increases. \treeindex is significantly faster than the existing methods while the difference becomes smaller when the treewidth is large. This validates that the proposed method \treeindex is proper for small treewidth graphs.

\begin{figure}[t!]
	\vspace*{0cm}
	\begin{center}
		\includegraphics[width=0.98\columnwidth]{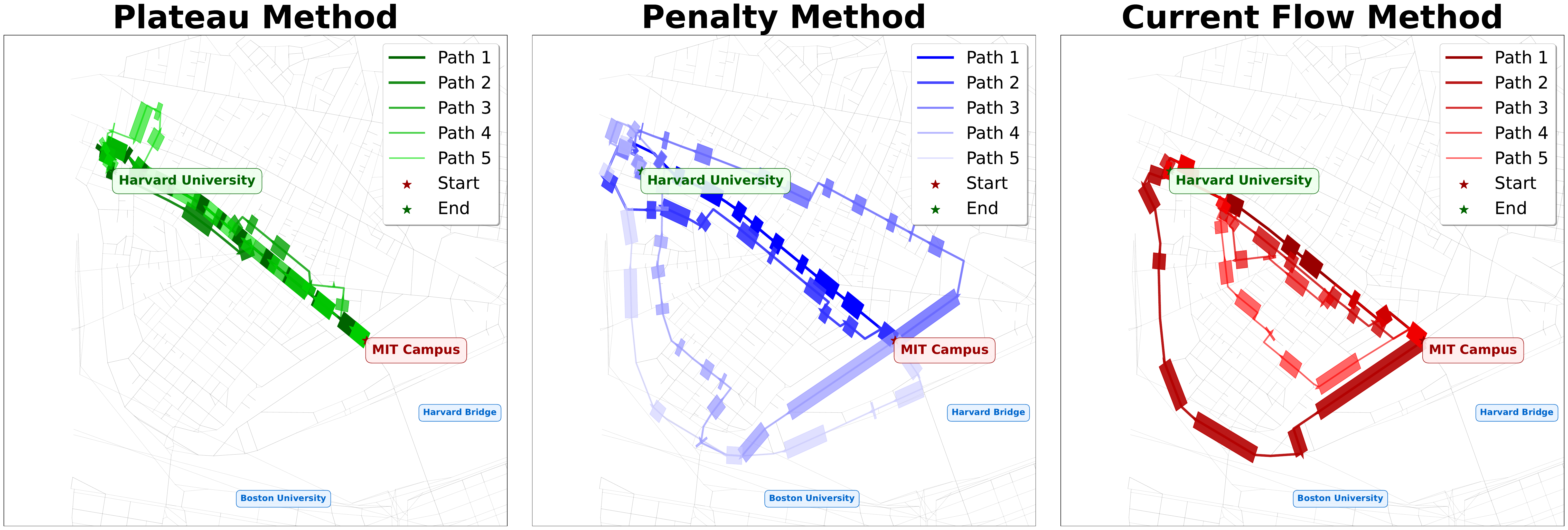}
	\end{center}
	\vspace*{-0.2cm}
	\caption{Comparison of different routing methods from MIT to Harvard}
	\vspace*{-0.3cm}
	\label{fig:case-study-robust-routing}
\end{figure}
\begin{table}[t]
	\centering
	\caption{Performance of different routing methods}
	\label{tab:detailed-comparison}
	\vspace*{-0.2cm}
	\resizebox{\columnwidth}{!}{
	\begin{tabular}{ccccc}
	\toprule
	\textbf{Method} & \textbf{Routing Time} & \length & \diversity & \robustness \\
	\midrule
	\plateaumethod \cite{PlateauMethod} & $0.002$ secs & 1.13 & 0.61 & 0.08 \\
	\penaltymethod \cite{PenaltyMethod} & $0.067$ secs & 1.49 & 0.94 & 0.87 \\
	\resistancepath & $0.010$ secs & 1.25 & 0.87 & 0.86 \\
	\bottomrule
	\end{tabular}
	}
	\vspace{-0.4cm}
\end{table}
\stitle{Exp-VII: Case study--Robust Routing on Road networks.} In this experiment, we compare the \resistancedistance-based routing method with the state-of-the-art robust routing methods \penaltymethod \cite{PenaltyMethod} and \plateaumethod \cite{PlateauMethod}. For evaluation, we use a real-world road network of Boston extracted from OpenStreetMap \cite{OpenStreetMap}, comprising $1,591$ nodes and $3,540$ edges. Here, edge weights represent the corresponding travel times. Fig.~\ref{fig:case-study-robust-routing} illustrates routing results obtained by different methods when $k=5$. It is evident that the \plateaumethod\ method generates similar paths, as it fails to avoid certain routes to reach Harvard. In contrast, \penaltymethod\ and \resistancepath\ consistently identify robust and diverse paths. Furthermore, we evaluate the quality of alternative paths using metrics such as \length, \diversity, and \robustness. Specifically, \length\ denotes the average ratio of the lengths of the alternative paths to that of the shortest path; \diversity\ represents the average pairwise Jaccard similarity among all alternative paths; and \robustness\ is the probability that $s$ and $t$ remain connected via the alternative paths after each edge is independently removed with a probability of $0.001$. Detailed results are presented in Table~\ref{tab:detailed-comparison}. It can be observed that the \penaltymethod can achieve higher \diversity and \robustness but at the cost of substantially longer routing time. Conversely, \plateaumethod method is significantly faster but produces similar paths. In comparison, \resistancepath consistently finds robust and diverse paths rapidly. These results indicate that the the \resistancepath-based routing method is ideal for robust routing applications on road networks.

\stitle{Summary of findings.} The experimental results demonstrate that the proposed \treeindex approach significantly outperforms state-of-the-art exact and approximate methods in terms of query efficiency (over $3$ orders of magnitude faster) and query accuracy (nearly exact), at the expense of increased index construction time and index size. On social networks (have large tree-widths), the index construction takes more than $145$ hours and is challenging to apply to social networks larger than \dblp. In contrast, our method exhibits significantly superior performance on small tree-width graphs, such as road networks. Specifically, it successfully constructs resistance distance labelling on the largest available road network, \fullusa, within $7$ hours, resulting in an index size of approximately $405$~GB, whereas none of the existing methods can compute exact resistance distances on this network. The improvements provided by our method make the query performance of resistance distance comparable to that of shortest path distance.

\section{Related work}\label{sec:related-work}
\stitle{Resistance Distance Computation.} Resistance distance computation is a well-established problem in graph data management. Several algorithms have been proposed for computing resistance distance by the theoretical community \cite{GraphSparsificationEff08,TowardsOptimalNIPS23,LaplacianSolver,SimpleParallelLaplacianSolver23,UltrasparseUltrasparsifiers21,SolvingSDD14,NonExpander2023effective,LiLawrence2023new,CutToggling23}. A representative set of these methods are based on the Laplacian solver \cite{GraphSparsificationEff08,LaplacianSolver,SimpleParallelLaplacianSolver23,UltrasparseUltrasparsifiers21,SolvingSDD14}, which achieves a near-linear time complexity (with respect to the number of edges). However, despite numerous attempts to efficiently implement the Laplacian solver in practice \cite{RobustandPracticalLaplacian23,RCHOL21,AnEmpiricalComparison16}, the hidden constant factors in these complexity analyses are substantial, resulting in poor practical efficiency. In contrast to these methods, we focus on algorithms that are efficient in practice. From this perspective, numerous studies have focused on approximate solutions. \cite{KDDlocal21} first proposed local algorithms for resistance distance by sampling random walks, while \cite{22resistance,ResistanceYang} subsequently reduced the variance of this approach. Other studies have explored sampling spanning trees \cite{22resistance,SpanningEdgeCentrality}. To further enhance efficiency, an index-based solution was proposed in \cite{23resistance,EfficientIndexMaintenanceACM25}, which uses spanning forest sampling to approximate several relevant matrices as indices. These algorithms are more suitable for graphs where random walks mix rapidly, such as social networks, and are inefficient for graphs with small treewidth, such as road networks.

\stitle{Shortest Path Distance Computation.} Shortest path distance computation is another fundamental problem in graph data management. Since online methods such as Dijkstra's algorithm, bidirectional search \cite{BidirectionalShortestPath77}, and A* search \cite{Asearch05} are inefficient for large-scale networks, numerous studies have focused on index-based methods. The basic idea is to find an $h$-hopset such that after adding this set, the distances can be exactly or approximately preserved while any two nodes can reach each other within $h$ hops in the new graph \cite{Cohen94}. This hop-set as well as the distance values can be stored as labellings for efficient query \cite{DistanceLabeling01}. The hop-set and labelling-based ideas have been extensively studied in theory literature \cite{SublinearDistanceLabeling16,HubLabelingHardness17,Hopsets19,HighwayDimension10,SkeletonLabel17,DistanceOracles11,HopsetsFOCS16,ExactDistanceOracleFOCS17}. $2$-hop labelling is a special case when $h=2$. Cohen conjectured that for any graph, the optimal $2$-hop cover has size $O(\sqrt{m}\cdot n)$ \cite{Cohen2hopLabel2002}. On graphs with special structure, such as bounded treewidth \cite{CorePeriphery20,DistanceOracles11}, bounded highway dimension \cite{HighwayDimension10} or bounded skeleton dimensions \cite{SkeletonLabel17}, the theoretical bounds can be improved. Such ideas have also been utilized to design practically efficient algorithms under different computational environments \cite{FastExact13ShortestPath,ISLABEL13,HopDoublingLabel14,ShortestPathIndexMaintenance20,ReinforcementLearningTreeDecomposition23,StableTreeLabellingEDBT25,DualHierarchyLabellingACM25,RelativeSubboundednessSIGMOD22,TEDISIGMOD10,LiJunChang2012exact,HopLabeling2018hierarchy,HierarchicalCutLabelling23,ProjectedVertexSeparator2021}. However, all these methods depend on the cut property of shortest path distance. As the cut property for resistance distance is unclear, resistance distance computation presents an entirely different challenge. To the best of our knowledge, none of the techniques used in these shortest path methods had been successfully adapted for resistance distance computation prior to our study. Among these approaches, the most relevant to our work are the tree decomposition-based methods. Tree decomposition was first applied to shortest path distance computation in \tedi \cite{TEDISIGMOD10}, which leverages the tree decomposition structure to construct distance labelling for efficient shortest path distance queries. Subsequent research enhanced this approach by introducing multi-hop queries (\multihop) \cite{LiJunChang2012exact}, hierarchical distance labelling (\htwoh) \cite{HopLabeling2018hierarchy}, and pruned vertex separators \cite{ProjectedVertexSeparator2021}. Recently, balanced vertex hierarchy has been employed to further reduce index sizes \cite{HierarchicalCutLabelling23}. In this paper, we adapt the concepts of tree decomposition and vertex hierarchy from these studies and develop non-trivial extensions specifically tailored for resistance distance computation.

\section{Conclusion}
In this paper, we propose \treeindex, a novel exact method for computing resistance distances by leveraging tree decomposition to construct resistance distance labelling. Our approach specifically addresses the computational limitations of existing random walk-based methods, which perform poorly on graphs with small treewidth. To overcome these limitations, we establish the cut property of resistance distance derived from the Cholesky decomposition of the inverse Laplacian matrix and efficiently extend it to the entire graph by exploiting the hierarchical structure obtained from tree decomposition. The resulting labelling achieves a compact size of $O(n \cdot h_{\mathcal{G}})$ and can be computed in $O(n \cdot h_{\mathcal{G}}^2 \cdot d_{max})$ time, where the tree height $h_{\mathcal{G}}$ and maximum degree $d_{max}$ are typically small in practical graphs with low treewidth, such as road networks. Utilizing this labelling, single-pair resistance distance queries can be answered in $O(h_{\mathcal{G}})$ time, whereas single-source queries require $O(n \cdot h_{\mathcal{G}})$ time. Extensive experiments demonstrate that our method achieves substantial improvements in query efficiency compared to state-of-the-art exact and approximate methods, while incurring only modest increases in index size and construction time.

\balance
\bibliographystyle{ACM-Reference-Format}
\bibliography{CholWilson}


\begin{thebibliography}{71}


\ifx \showCODEN    \undefined \def \showCODEN     #1{\unskip}     \fi
\ifx \showDOI      \undefined \def \showDOI       #1{#1}\fi
\ifx \showISBNx    \undefined \def \showISBNx     #1{\unskip}     \fi
\ifx \showISBNxiii \undefined \def \showISBNxiii  #1{\unskip}     \fi
\ifx \showISSN     \undefined \def \showISSN      #1{\unskip}     \fi
\ifx \showLCCN     \undefined \def \showLCCN      #1{\unskip}     \fi
\ifx \shownote     \undefined \def \shownote      #1{#1}          \fi
\ifx \showarticletitle \undefined \def \showarticletitle #1{#1}   \fi
\ifx \showURL      \undefined \def \showURL       {\relax}        \fi
\providecommand\bibfield[2]{#2}
\providecommand\bibinfo[2]{#2}
\providecommand\natexlab[1]{#1}
\providecommand\showeprint[2][]{arXiv:#2}

\bibitem[Pla(2009)]%
        {PlateauMethod}
 \bibinfo{year}{2009}\natexlab{}.
\newblock \bibinfo{title}{CAMVIT: Choice routing}.
\newblock \bibinfo{howpublished}{\url{http://www.camvit.com}}.
\newblock
\newblock
\shownote{Accessed: 2024-06-23}.


\bibitem[Abraham et~al\mbox{.}(2010)]%
        {HighwayDimension10}
\bibfield{author}{\bibinfo{person}{Ittai Abraham}, \bibinfo{person}{Amos Fiat}, \bibinfo{person}{Andrew~V. Goldberg}, {and} \bibinfo{person}{Renato Fonseca~F. Werneck}.} \bibinfo{year}{2010}\natexlab{}.
\newblock \showarticletitle{Highway Dimension, Shortest Paths, and Provably Efficient Algorithms}. In \bibinfo{booktitle}{\emph{SODA}}. \bibinfo{pages}{782--793}.
\newblock


\bibitem[Akiba et~al\mbox{.}(2013)]%
        {FastExact13ShortestPath}
\bibfield{author}{\bibinfo{person}{Takuya Akiba}, \bibinfo{person}{Yoichi Iwata}, {and} \bibinfo{person}{Yuichi Yoshida}.} \bibinfo{year}{2013}\natexlab{}.
\newblock \showarticletitle{Fast exact shortest-path distance queries on large networks by pruned landmark labeling}. In \bibinfo{booktitle}{\emph{SIGMOD}}. \bibinfo{pages}{349--360}.
\newblock


\bibitem[Alev et~al\mbox{.}(2018)]%
        {GraphClusteringITCS18}
\bibfield{author}{\bibinfo{person}{Vedat~Levi Alev}, \bibinfo{person}{Nima Anari}, \bibinfo{person}{Lap~Chi Lau}, {and} \bibinfo{person}{Shayan~Oveis Gharan}.} \bibinfo{year}{2018}\natexlab{}.
\newblock \showarticletitle{Graph Clustering using Effective Resistance}. In \bibinfo{booktitle}{\emph{{ITCS}}} \emph{(\bibinfo{series}{LIPIcs}, Vol.~\bibinfo{volume}{94})}. \bibinfo{pages}{41:1--41:16}.
\newblock


\bibitem[Alstrup et~al\mbox{.}(2016)]%
        {SublinearDistanceLabeling16}
\bibfield{author}{\bibinfo{person}{Stephen Alstrup}, \bibinfo{person}{S{\o}ren Dahlgaard}, \bibinfo{person}{Mathias B{\ae}k~Tejs Knudsen}, {and} \bibinfo{person}{Ely Porat}.} \bibinfo{year}{2016}\natexlab{}.
\newblock \showarticletitle{Sublinear Distance Labeling}. In \bibinfo{booktitle}{\emph{ESA}} \emph{(\bibinfo{series}{LIPIcs}, Vol.~\bibinfo{volume}{57})}. \bibinfo{pages}{5:1--5:15}.
\newblock


\bibitem[Angelidakis et~al\mbox{.}(2017)]%
        {HubLabelingHardness17}
\bibfield{author}{\bibinfo{person}{Haris Angelidakis}, \bibinfo{person}{Yury Makarychev}, {and} \bibinfo{person}{Vsevolod Oparin}.} \bibinfo{year}{2017}\natexlab{}.
\newblock \showarticletitle{Algorithmic and Hardness Results for the Hub Labeling Problem}. In \bibinfo{booktitle}{\emph{SODA}}. \bibinfo{pages}{1442--1461}.
\newblock


\bibitem[Authors(2025)]%
        {fullversion}
\bibfield{author}{\bibinfo{person}{Anonymous Authors}.} \bibinfo{year}{2025}\natexlab{}.
\newblock \showarticletitle{Efficient Exact Resistance Distance Computation on Small-Treewidth Graphs: a Labelling Approach}.
\newblock \bibinfo{journal}{\emph{Full version: https://anonymous.4open.science/r/TreeIndex-32E9}} (\bibinfo{year}{2025}).
\newblock


\bibitem[Bader et~al\mbox{.}(2011)]%
        {PenaltyMethod}
\bibfield{author}{\bibinfo{person}{Roland Bader}, \bibinfo{person}{Jonathan Dees}, \bibinfo{person}{Robert Geisberger}, {and} \bibinfo{person}{Peter Sanders}.} \bibinfo{year}{2011}\natexlab{}.
\newblock \showarticletitle{Alternative Route Graphs in Road Networks}. In \bibinfo{booktitle}{\emph{{ICST} Conference, {TAPAS}}} \emph{(\bibinfo{series}{Lecture Notes in Computer Science}, Vol.~\bibinfo{volume}{6595})}. \bibinfo{pages}{21--32}.
\newblock


\bibitem[Bodlaender(2006)]%
        {Treewidth06}
\bibfield{author}{\bibinfo{person}{Hans~L. Bodlaender}.} \bibinfo{year}{2006}\natexlab{}.
\newblock \showarticletitle{Treewidth: Characterizations, Applications, and Computations}. In \bibinfo{booktitle}{\emph{Graph-Theoretic Concepts in Computer Science, 32nd International Workshop}}. \bibinfo{pages}{1--14}.
\newblock


\bibitem[Bollob{\'a}s(1998)]%
        {bollobas1998modern}
\bibfield{author}{\bibinfo{person}{B{\'e}la Bollob{\'a}s}.} \bibinfo{year}{1998}\natexlab{}.
\newblock \bibinfo{booktitle}{\emph{Modern graph theory}}. Vol.~\bibinfo{volume}{184}.
\newblock


\bibitem[Boman et~al\mbox{.}(2016)]%
        {AnEmpiricalComparison16}
\bibfield{author}{\bibinfo{person}{Erik~G. Boman}, \bibinfo{person}{Kevin Deweese}, {and} \bibinfo{person}{John~R. Gilbert}.} \bibinfo{year}{2016}\natexlab{}.
\newblock \showarticletitle{An Empirical Comparison of Graph Laplacian Solvers}. In \bibinfo{booktitle}{\emph{ALENEX}}. \bibinfo{pages}{174--188}.
\newblock


\bibitem[Cai et~al\mbox{.}(2023)]%
        {NonExpander2023effective}
\bibfield{author}{\bibinfo{person}{Dongrun Cai}, \bibinfo{person}{Xue Chen}, {and} \bibinfo{person}{Pan Peng}.} \bibinfo{year}{2023}\natexlab{}.
\newblock \showarticletitle{Effective Resistances in Non-Expander Graphs}. In \bibinfo{booktitle}{\emph{ESA}}, Vol.~\bibinfo{volume}{274}. \bibinfo{pages}{29:1--29:18}.
\newblock


\bibitem[Chang et~al\mbox{.}(2012)]%
        {LiJunChang2012exact}
\bibfield{author}{\bibinfo{person}{Lijun Chang}, \bibinfo{person}{Jeffrey~Xu Yu}, \bibinfo{person}{Lu Qin}, \bibinfo{person}{Hong Cheng}, {and} \bibinfo{person}{Miao Qiao}.} \bibinfo{year}{2012}\natexlab{}.
\newblock \showarticletitle{The exact distance to destination in undirected world}.
\newblock \bibinfo{journal}{\emph{{VLDB} J.}} \bibinfo{volume}{21}, \bibinfo{number}{6} (\bibinfo{year}{2012}), \bibinfo{pages}{869--888}.
\newblock


\bibitem[Chen et~al\mbox{.}(2021b)]%
        {RCHOL21}
\bibfield{author}{\bibinfo{person}{Chao Chen}, \bibinfo{person}{Tianyu Liang}, {and} \bibinfo{person}{George Biros}.} \bibinfo{year}{2021}\natexlab{b}.
\newblock \showarticletitle{{RCHOL:} Randomized Cholesky Factorization for Solving {SDD} Linear Systems}.
\newblock \bibinfo{journal}{\emph{{SIAM} J. Sci. Comput.}} \bibinfo{volume}{43}, \bibinfo{number}{6} (\bibinfo{year}{2021}), \bibinfo{pages}{C411--C438}.
\newblock


\bibitem[Chen et~al\mbox{.}(2021a)]%
        {ProjectedVertexSeparator2021}
\bibfield{author}{\bibinfo{person}{Zitong Chen}, \bibinfo{person}{Ada~Wai{-}Chee Fu}, \bibinfo{person}{Minhao Jiang}, \bibinfo{person}{Eric Lo}, {and} \bibinfo{person}{Pengfei Zhang}.} \bibinfo{year}{2021}\natexlab{a}.
\newblock \showarticletitle{{P2H:} Efficient Distance Querying on Road Networks by Projected Vertex Separators}. In \bibinfo{booktitle}{\emph{SIGMOD}}. \bibinfo{pages}{313--325}.
\newblock


\bibitem[Chung and Lu(2002)]%
        {chung2002connected}
\bibfield{author}{\bibinfo{person}{Fan Chung} {and} \bibinfo{person}{Linyuan Lu}.} \bibinfo{year}{2002}\natexlab{}.
\newblock \showarticletitle{Connected components in random graphs with given expected degree sequences}.
\newblock \bibinfo{journal}{\emph{Annals of combinatorics}} \bibinfo{volume}{6}, \bibinfo{number}{2} (\bibinfo{year}{2002}), \bibinfo{pages}{125--145}.
\newblock


\bibitem[Chung(1997)]%
        {chung1997spectral}
\bibfield{author}{\bibinfo{person}{Fan~RK Chung}.} \bibinfo{year}{1997}\natexlab{}.
\newblock \bibinfo{booktitle}{\emph{Spectral graph theory}}. Vol.~\bibinfo{volume}{92}.
\newblock \bibinfo{publisher}{American Mathematical Soc.}
\newblock


\bibitem[Cohen(1994)]%
        {Cohen94}
\bibfield{author}{\bibinfo{person}{Edith Cohen}.} \bibinfo{year}{1994}\natexlab{}.
\newblock \showarticletitle{Polylog-time and near-linear work approximation scheme for undirected shortest paths}. In \bibinfo{booktitle}{\emph{STOC}}. \bibinfo{pages}{16--26}.
\newblock


\bibitem[Cohen et~al\mbox{.}(2002)]%
        {Cohen2hopLabel2002}
\bibfield{author}{\bibinfo{person}{Edith Cohen}, \bibinfo{person}{Eran Halperin}, \bibinfo{person}{Haim Kaplan}, {and} \bibinfo{person}{Uri Zwick}.} \bibinfo{year}{2002}\natexlab{}.
\newblock \showarticletitle{Reachability and distance queries via 2-hop labels}. In \bibinfo{booktitle}{\emph{SODA}}. \bibinfo{pages}{937--946}.
\newblock


\bibitem[Cohen et~al\mbox{.}(2014)]%
        {SolvingSDD14}
\bibfield{author}{\bibinfo{person}{Michael~B. Cohen}, \bibinfo{person}{Rasmus Kyng}, \bibinfo{person}{Gary~L. Miller}, \bibinfo{person}{Jakub~W. Pachocki}, \bibinfo{person}{Richard Peng}, \bibinfo{person}{Anup~B. Rao}, {and} \bibinfo{person}{Shen~Chen Xu}.} \bibinfo{year}{2014}\natexlab{}.
\newblock \showarticletitle{Solving {SDD} linear systems in nearly \emph{m}log\({}^{\mbox{1/2}}\)\emph{n} time}. In \bibinfo{booktitle}{\emph{STOC}}. \bibinfo{pages}{343--352}.
\newblock


\bibitem[Cohen{-}Addad et~al\mbox{.}(2017)]%
        {ExactDistanceOracleFOCS17}
\bibfield{author}{\bibinfo{person}{Vincent Cohen{-}Addad}, \bibinfo{person}{S{\o}ren Dahlgaard}, {and} \bibinfo{person}{Christian Wulff{-}Nilsen}.} \bibinfo{year}{2017}\natexlab{}.
\newblock \showarticletitle{Fast and Compact Exact Distance Oracle for Planar Graphs}. In \bibinfo{booktitle}{\emph{FOCS}}. \bibinfo{pages}{962--973}.
\newblock


\bibitem[de~Champeaux and Sint(1977)]%
        {BidirectionalShortestPath77}
\bibfield{author}{\bibinfo{person}{Dennis de Champeaux} {and} \bibinfo{person}{Lenie Sint}.} \bibinfo{year}{1977}\natexlab{}.
\newblock \showarticletitle{An Optimality Theorem for a Bi-Directional Heuristic Search Algorithm}.
\newblock \bibinfo{journal}{\emph{Comput. J.}} \bibinfo{volume}{20}, \bibinfo{number}{2} (\bibinfo{year}{1977}), \bibinfo{pages}{148--150}.
\newblock


\bibitem[Demetrescu et~al\mbox{.}(2009)]%
        {dimacs}
\bibfield{author}{\bibinfo{person}{Camil Demetrescu}, \bibinfo{person}{Andrew Goldberg}, {and} \bibinfo{person}{David Johnson}.} \bibinfo{year}{2009}\natexlab{}.
\newblock \bibinfo{title}{The shortest path problem: Ninth DIMACS implementation challenge}.
\newblock \bibinfo{howpublished}{\url{https://www.diag.uniroma1.it/challenge9}}.
\newblock


\bibitem[Devriendt et~al\mbox{.}(2022)]%
        {GraphVariance}
\bibfield{author}{\bibinfo{person}{Karel Devriendt}, \bibinfo{person}{Samuel Martin{-}Gutierrez}, {and} \bibinfo{person}{Renaud Lambiotte}.} \bibinfo{year}{2022}\natexlab{}.
\newblock \showarticletitle{Variance and Covariance of Distributions on Graphs}.
\newblock \bibinfo{journal}{\emph{{SIAM} Rev.}} \bibinfo{volume}{64}, \bibinfo{number}{2} (\bibinfo{year}{2022}), \bibinfo{pages}{343--359}.
\newblock


\bibitem[Devriendt et~al\mbox{.}(2024)]%
        {GraphCurvature}
\bibfield{author}{\bibinfo{person}{Karel Devriendt}, \bibinfo{person}{Andrea Ottolini}, {and} \bibinfo{person}{Stefan Steinerberger}.} \bibinfo{year}{2024}\natexlab{}.
\newblock \showarticletitle{Graph curvature via resistance distance}.
\newblock \bibinfo{journal}{\emph{Discret. Appl. Math.}}  \bibinfo{volume}{348} (\bibinfo{year}{2024}), \bibinfo{pages}{68--78}.
\newblock


\bibitem[Dwaraknath et~al\mbox{.}(2023)]%
        {TowardsOptimalNIPS23}
\bibfield{author}{\bibinfo{person}{Rajat~Vadiraj Dwaraknath}, \bibinfo{person}{Ishani Karmarkar}, {and} \bibinfo{person}{Aaron Sidford}.} \bibinfo{year}{2023}\natexlab{}.
\newblock \showarticletitle{Towards Optimal Effective Resistance Estimation}. In \bibinfo{booktitle}{\emph{NIPS}}.
\newblock


\bibitem[Elkin and Neiman(2016)]%
        {HopsetsFOCS16}
\bibfield{author}{\bibinfo{person}{Michael Elkin} {and} \bibinfo{person}{Ofer Neiman}.} \bibinfo{year}{2016}\natexlab{}.
\newblock \showarticletitle{Hopsets with Constant Hopbound, and Applications to Approximate Shortest Paths}. In \bibinfo{booktitle}{\emph{FOCS}}. \bibinfo{pages}{128--137}.
\newblock


\bibitem[Farhan et~al\mbox{.}(2023)]%
        {HierarchicalCutLabelling23}
\bibfield{author}{\bibinfo{person}{Muhammad Farhan}, \bibinfo{person}{Henning Koehler}, \bibinfo{person}{Robert Ohms}, {and} \bibinfo{person}{Qing Wang}.} \bibinfo{year}{2023}\natexlab{}.
\newblock \showarticletitle{Hierarchical Cut Labelling - Scaling Up Distance Queries on Road Networks}.
\newblock \bibinfo{journal}{\emph{Proc. {ACM} Manag. Data}} \bibinfo{volume}{1}, \bibinfo{number}{4} (\bibinfo{year}{2023}), \bibinfo{pages}{244:1--244:25}.
\newblock


\bibitem[Farhan et~al\mbox{.}(2025)]%
        {DualHierarchyLabellingACM25}
\bibfield{author}{\bibinfo{person}{Muhammad Farhan}, \bibinfo{person}{Henning Koehler}, {and} \bibinfo{person}{Qing Wang}.} \bibinfo{year}{2025}\natexlab{}.
\newblock \showarticletitle{Dual-Hierarchy Labelling: Scaling Up Distance Queries on Dynamic Road Networks}.
\newblock \bibinfo{journal}{\emph{Proc. {ACM} Manag. Data}} \bibinfo{volume}{3}, \bibinfo{number}{1} (\bibinfo{year}{2025}), \bibinfo{pages}{35:1--35:25}.
\newblock


\bibitem[Farzan and Kamali(2011)]%
        {DistanceOracles11}
\bibfield{author}{\bibinfo{person}{Arash Farzan} {and} \bibinfo{person}{Shahin Kamali}.} \bibinfo{year}{2011}\natexlab{}.
\newblock \showarticletitle{Compact Navigation and Distance Oracles for Graphs with Small Treewidth}. In \bibinfo{booktitle}{\emph{ICALP}}. \bibinfo{pages}{268--280}.
\newblock


\bibitem[Fu et~al\mbox{.}(2013)]%
        {ISLABEL13}
\bibfield{author}{\bibinfo{person}{Ada~Wai{-}Chee Fu}, \bibinfo{person}{Huanhuan Wu}, \bibinfo{person}{James Cheng}, {and} \bibinfo{person}{Raymond~Chi{-}Wing Wong}.} \bibinfo{year}{2013}\natexlab{}.
\newblock \showarticletitle{{IS-LABEL:} an Independent-Set based Labeling Scheme for Point-to-Point Distance Querying}.
\newblock \bibinfo{journal}{\emph{VLDB}} \bibinfo{volume}{6}, \bibinfo{number}{6} (\bibinfo{year}{2013}), \bibinfo{pages}{457--468}.
\newblock


\bibitem[Gao et~al\mbox{.}(2023)]%
        {RobustandPracticalLaplacian23}
\bibfield{author}{\bibinfo{person}{Yuan Gao}, \bibinfo{person}{Rasmus Kyng}, {and} \bibinfo{person}{Daniel~A. Spielman}.} \bibinfo{year}{2023}\natexlab{}.
\newblock \showarticletitle{Robust and Practical Solution of Laplacian Equations by Approximate Elimination}.
\newblock \bibinfo{journal}{\emph{CoRR}}  \bibinfo{volume}{abs/2303.00709} (\bibinfo{year}{2023}).
\newblock


\bibitem[Gavoille et~al\mbox{.}(2001)]%
        {DistanceLabeling01}
\bibfield{author}{\bibinfo{person}{Cyril Gavoille}, \bibinfo{person}{David Peleg}, \bibinfo{person}{Stephane Perennes}, {and} \bibinfo{person}{Ran Raz}.} \bibinfo{year}{2001}\natexlab{}.
\newblock \showarticletitle{Distance labeling in graphs}. In \bibinfo{booktitle}{\emph{SODA}}. \bibinfo{pages}{210--219}.
\newblock


\bibitem[Goldberg and Harrelson(2005)]%
        {Asearch05}
\bibfield{author}{\bibinfo{person}{Andrew~V. Goldberg} {and} \bibinfo{person}{Chris Harrelson}.} \bibinfo{year}{2005}\natexlab{}.
\newblock \showarticletitle{Computing the shortest path: \emph{A} search meets graph theory}. In \bibinfo{booktitle}{\emph{SODA}}. \bibinfo{pages}{156--165}.
\newblock


\bibitem[Golub and Van~Loan(2013)]%
        {golub2013matrix}
\bibfield{author}{\bibinfo{person}{Gene~H Golub} {and} \bibinfo{person}{Charles~F Van~Loan}.} \bibinfo{year}{2013}\natexlab{}.
\newblock \bibinfo{booktitle}{\emph{Matrix computations}}.
\newblock \bibinfo{publisher}{JHU press}.
\newblock


\bibitem[Gupta et~al\mbox{.}(2019)]%
        {Hopsets19}
\bibfield{author}{\bibinfo{person}{Siddharth Gupta}, \bibinfo{person}{Adrian Kosowski}, {and} \bibinfo{person}{Laurent Viennot}.} \bibinfo{year}{2019}\natexlab{}.
\newblock \showarticletitle{Exploiting Hopsets: Improved Distance Oracles for Graphs of Constant Highway Dimension and Beyond}. In \bibinfo{booktitle}{\emph{ICALP}}, Vol.~\bibinfo{volume}{132}. \bibinfo{pages}{143:1--143:15}.
\newblock


\bibitem[Hayashi et~al\mbox{.}(2016)]%
        {SpanningEdgeCentrality}
\bibfield{author}{\bibinfo{person}{Takanori Hayashi}, \bibinfo{person}{Takuya Akiba}, {and} \bibinfo{person}{Yuichi Yoshida}.} \bibinfo{year}{2016}\natexlab{}.
\newblock \showarticletitle{Efficient Algorithms for Spanning Tree Centrality}. In \bibinfo{booktitle}{\emph{IJCAI}}. \bibinfo{pages}{3733--3739}.
\newblock


\bibitem[Henzinger et~al\mbox{.}(2023)]%
        {CutToggling23}
\bibfield{author}{\bibinfo{person}{Monika Henzinger}, \bibinfo{person}{Billy Jin}, \bibinfo{person}{Richard Peng}, {and} \bibinfo{person}{David~P. Williamson}.} \bibinfo{year}{2023}\natexlab{}.
\newblock \showarticletitle{A Combinatorial Cut-Toggling Algorithm for Solving Laplacian Linear Systems}.
\newblock \bibinfo{journal}{\emph{Algorithmica}} \bibinfo{volume}{85}, \bibinfo{number}{12} (\bibinfo{year}{2023}), \bibinfo{pages}{3680--3716}.
\newblock


\bibitem[Jambulapati and Sidford(2021)]%
        {UltrasparseUltrasparsifiers21}
\bibfield{author}{\bibinfo{person}{Arun Jambulapati} {and} \bibinfo{person}{Aaron Sidford}.} \bibinfo{year}{2021}\natexlab{}.
\newblock \showarticletitle{Ultrasparse Ultrasparsifiers and Faster Laplacian System Solvers}. In \bibinfo{booktitle}{\emph{SODA}}. \bibinfo{pages}{540--559}.
\newblock


\bibitem[Jiang et~al\mbox{.}(2014)]%
        {HopDoublingLabel14}
\bibfield{author}{\bibinfo{person}{Minhao Jiang}, \bibinfo{person}{Ada~Wai{-}Chee Fu}, \bibinfo{person}{Raymond~Chi{-}Wing Wong}, {and} \bibinfo{person}{Yanyan Xu}.} \bibinfo{year}{2014}\natexlab{}.
\newblock \showarticletitle{Hop Doubling Label Indexing for Point-to-Point Distance Querying on Scale-Free Networks}.
\newblock \bibinfo{journal}{\emph{VLDB}} \bibinfo{volume}{7}, \bibinfo{number}{12} (\bibinfo{year}{2014}), \bibinfo{pages}{1203--1214}.
\newblock


\bibitem[Koehler et~al\mbox{.}(2025)]%
        {StableTreeLabellingEDBT25}
\bibfield{author}{\bibinfo{person}{Henning Koehler}, \bibinfo{person}{Muhammad Farhan}, {and} \bibinfo{person}{Qing Wang}.} \bibinfo{year}{2025}\natexlab{}.
\newblock \showarticletitle{Stable Tree Labelling for Accelerating Distance Queries on Dynamic Road Networks}. In \bibinfo{booktitle}{\emph{EDBT}}. \bibinfo{pages}{477--489}.
\newblock


\bibitem[Kosowski and Viennot(2017)]%
        {SkeletonLabel17}
\bibfield{author}{\bibinfo{person}{Adrian Kosowski} {and} \bibinfo{person}{Laurent Viennot}.} \bibinfo{year}{2017}\natexlab{}.
\newblock \showarticletitle{Beyond Highway Dimension: Small Distance Labels Using Tree Skeletons}. In \bibinfo{booktitle}{\emph{SODA}}. \bibinfo{pages}{1462--1478}.
\newblock


\bibitem[Kyng and Sachdeva(2016)]%
        {LaplacianSolver}
\bibfield{author}{\bibinfo{person}{Rasmus Kyng} {and} \bibinfo{person}{Sushant Sachdeva}.} \bibinfo{year}{2016}\natexlab{}.
\newblock \showarticletitle{Approximate Gaussian Elimination for Laplacians - Fast, Sparse, and Simple}. In \bibinfo{booktitle}{\emph{FOCS}}. \bibinfo{pages}{573--582}.
\newblock


\bibitem[Leskovec and Krevl(2014)]%
        {snapnets}
\bibfield{author}{\bibinfo{person}{Jure Leskovec} {and} \bibinfo{person}{Andrej Krevl}.} \bibinfo{year}{2014}\natexlab{}.
\newblock \bibinfo{title}{{SNAP Datasets}: {Stanford} Large Network Dataset Collection}.
\newblock \bibinfo{howpublished}{\url{http://snap.stanford.edu/data}}.
\newblock


\bibitem[Li and Sachdeva(2023)]%
        {LiLawrence2023new}
\bibfield{author}{\bibinfo{person}{Lawrence Li} {and} \bibinfo{person}{Sushant Sachdeva}.} \bibinfo{year}{2023}\natexlab{}.
\newblock \showarticletitle{A New Approach to Estimating Effective Resistances and Counting Spanning Trees in Expander Graphs}. In \bibinfo{booktitle}{\emph{SODA}}. \bibinfo{pages}{2728--2745}.
\newblock


\bibitem[Li et~al\mbox{.}(2020)]%
        {CorePeriphery20}
\bibfield{author}{\bibinfo{person}{Wentao Li}, \bibinfo{person}{Miao Qiao}, \bibinfo{person}{Lu Qin}, \bibinfo{person}{Ying Zhang}, \bibinfo{person}{Lijun Chang}, {and} \bibinfo{person}{Xuemin Lin}.} \bibinfo{year}{2020}\natexlab{}.
\newblock \showarticletitle{Scaling Up Distance Labeling on Graphs with Core-Periphery Properties}. In \bibinfo{booktitle}{\emph{SIGMOD}}. \bibinfo{pages}{1367--1381}.
\newblock


\bibitem[Liao et~al\mbox{.}(2025)]%
        {EfficientIndexMaintenanceACM25}
\bibfield{author}{\bibinfo{person}{Meihao Liao}, \bibinfo{person}{Cheng Li}, \bibinfo{person}{Rong{-}Hua Li}, {and} \bibinfo{person}{Guoren Wang}.} \bibinfo{year}{2025}\natexlab{}.
\newblock \showarticletitle{Efficient Index Maintenance for Effective Resistance Computation on Evolving Graphs}.
\newblock \bibinfo{journal}{\emph{Proc. {ACM} Manag. Data}} \bibinfo{volume}{3}, \bibinfo{number}{1} (\bibinfo{year}{2025}), \bibinfo{pages}{36:1--36:27}.
\newblock


\bibitem[Liao et~al\mbox{.}(2023)]%
        {22resistance}
\bibfield{author}{\bibinfo{person}{Meihao Liao}, \bibinfo{person}{Rong{-}Hua Li}, \bibinfo{person}{Qiangqiang Dai}, \bibinfo{person}{Hongyang Chen}, \bibinfo{person}{Hongchao Qin}, {and} \bibinfo{person}{Guoren Wang}.} \bibinfo{year}{2023}\natexlab{}.
\newblock \showarticletitle{Efficient Resistance Distance Computation: The Power of Landmark-based Approaches}.
\newblock \bibinfo{journal}{\emph{Proc. {ACM} Manag. Data}} \bibinfo{volume}{1}, \bibinfo{number}{1} (\bibinfo{year}{2023}), \bibinfo{pages}{68:1--68:27}.
\newblock


\bibitem[Liao et~al\mbox{.}(2024)]%
        {23resistance}
\bibfield{author}{\bibinfo{person}{Meihao Liao}, \bibinfo{person}{Junjie Zhou}, \bibinfo{person}{Rong{-}Hua Li}, \bibinfo{person}{Qiangqiang Dai}, \bibinfo{person}{Hongyang Chen}, {and} \bibinfo{person}{Guoren Wang}.} \bibinfo{year}{2024}\natexlab{}.
\newblock \showarticletitle{Efficient and Provable Effective Resistance Computation on Large Graphs: An Index-based Approach}.
\newblock \bibinfo{journal}{\emph{Proc. {ACM} Manag. Data}} \bibinfo{volume}{2}, \bibinfo{number}{3} (\bibinfo{year}{2024}), \bibinfo{pages}{133}.
\newblock


\bibitem[Liu et~al\mbox{.}(2023)]%
        {OversquashingWWW23}
\bibfield{author}{\bibinfo{person}{Yang Liu}, \bibinfo{person}{Chuan Zhou}, \bibinfo{person}{Shirui Pan}, \bibinfo{person}{Jia Wu}, \bibinfo{person}{Zhao Li}, \bibinfo{person}{Hongyang Chen}, {and} \bibinfo{person}{Peng Zhang}.} \bibinfo{year}{2023}\natexlab{}.
\newblock \showarticletitle{CurvDrop: {A} Ricci Curvature Based Approach to Prevent Graph Neural Networks from Over-Smoothing and Over-Squashing}. In \bibinfo{booktitle}{\emph{WWW}}. \bibinfo{pages}{221--230}.
\newblock


\bibitem[Maehara et~al\mbox{.}(2014)]%
        {CoreTree14}
\bibfield{author}{\bibinfo{person}{Takanori Maehara}, \bibinfo{person}{Takuya Akiba}, \bibinfo{person}{Yoichi Iwata}, {and} \bibinfo{person}{Ken{-}ichi Kawarabayashi}.} \bibinfo{year}{2014}\natexlab{}.
\newblock \showarticletitle{Computing Personalized PageRank Quickly by Exploiting Graph Structures}.
\newblock \bibinfo{journal}{\emph{VLDB}} \bibinfo{volume}{7}, \bibinfo{number}{12} (\bibinfo{year}{2014}), \bibinfo{pages}{1023--1034}.
\newblock


\bibitem[Mohaisen et~al\mbox{.}(2010)]%
        {mixing2010measuring}
\bibfield{author}{\bibinfo{person}{Abedelaziz Mohaisen}, \bibinfo{person}{Aaram Yun}, {and} \bibinfo{person}{Yongdae Kim}.} \bibinfo{year}{2010}\natexlab{}.
\newblock \showarticletitle{Measuring the mixing time of social graphs}. In \bibinfo{booktitle}{\emph{SIGCOMM}}. \bibinfo{pages}{383--389}.
\newblock


\bibitem[{OpenStreetMap contributors}(2017)]%
        {OpenStreetMap}
\bibfield{author}{\bibinfo{person}{{OpenStreetMap contributors}}.} \bibinfo{year}{2017}\natexlab{}.
\newblock \bibinfo{title}{{Planet dump retrieved from https://planet.osm.org }}.
\newblock \bibinfo{howpublished}{\url{ https://www.openstreetmap.org }}.
\newblock


\bibitem[Ouyang et~al\mbox{.}(2018)]%
        {HopLabeling2018hierarchy}
\bibfield{author}{\bibinfo{person}{Dian Ouyang}, \bibinfo{person}{Lu Qin}, \bibinfo{person}{Lijun Chang}, \bibinfo{person}{Xuemin Lin}, \bibinfo{person}{Ying Zhang}, {and} \bibinfo{person}{Qing Zhu}.} \bibinfo{year}{2018}\natexlab{}.
\newblock \showarticletitle{When Hierarchy Meets 2-Hop-Labeling: Efficient Shortest Distance Queries on Road Networks}. In \bibinfo{booktitle}{\emph{SIGMOD}}. \bibinfo{pages}{709--724}.
\newblock


\bibitem[Ouyang et~al\mbox{.}(2020)]%
        {ShortestPathIndexMaintenance20}
\bibfield{author}{\bibinfo{person}{Dian Ouyang}, \bibinfo{person}{Long Yuan}, \bibinfo{person}{Lu Qin}, \bibinfo{person}{Lijun Chang}, \bibinfo{person}{Ying Zhang}, {and} \bibinfo{person}{Xuemin Lin}.} \bibinfo{year}{2020}\natexlab{}.
\newblock \showarticletitle{Efficient Shortest Path Index Maintenance on Dynamic Road Networks with Theoretical Guarantees}.
\newblock \bibinfo{journal}{\emph{Proc. {VLDB} Endow.}} \bibinfo{volume}{13}, \bibinfo{number}{5} (\bibinfo{year}{2020}), \bibinfo{pages}{602--615}.
\newblock


\bibitem[Pachev and Webb(2018)]%
        {linkprediction2018spectralembedding}
\bibfield{author}{\bibinfo{person}{Benjamin Pachev} {and} \bibinfo{person}{Benjamin Webb}.} \bibinfo{year}{2018}\natexlab{}.
\newblock \showarticletitle{Fast link prediction for large networks using spectral embedding}.
\newblock \bibinfo{journal}{\emph{Journal of Complex Networks}} \bibinfo{volume}{6}, \bibinfo{number}{1} (\bibinfo{year}{2018}), \bibinfo{pages}{79--94}.
\newblock


\bibitem[Peng et~al\mbox{.}(2021)]%
        {KDDlocal21}
\bibfield{author}{\bibinfo{person}{Pan Peng}, \bibinfo{person}{Daniel Lopatta}, \bibinfo{person}{Yuichi Yoshida}, {and} \bibinfo{person}{Gramoz Goranci}.} \bibinfo{year}{2021}\natexlab{}.
\newblock \showarticletitle{Local Algorithms for Estimating Effective Resistance}. In \bibinfo{booktitle}{\emph{KDD}}. \bibinfo{pages}{1329--1338}.
\newblock


\bibitem[Qi et~al\mbox{.}(2021)]%
        {mixing18}
\bibfield{author}{\bibinfo{person}{Yi Qi}, \bibinfo{person}{Wanyue Xu}, \bibinfo{person}{Liwang Zhu}, {and} \bibinfo{person}{Zhongzhi Zhang}.} \bibinfo{year}{2021}\natexlab{}.
\newblock \showarticletitle{Real-World Networks Are Not Always Fast Mixing}.
\newblock \bibinfo{journal}{\emph{Comput. J.}} \bibinfo{volume}{64}, \bibinfo{number}{2} (\bibinfo{year}{2021}), \bibinfo{pages}{236--244}.
\newblock


\bibitem[Robertson and Seymour(1984)]%
        {TreeWidth84}
\bibfield{author}{\bibinfo{person}{Neil Robertson} {and} \bibinfo{person}{Paul~D. Seymour}.} \bibinfo{year}{1984}\natexlab{}.
\newblock \showarticletitle{Graph minors. {III.} Planar tree-width}.
\newblock \bibinfo{journal}{\emph{J. Comb. Theory {B}}} \bibinfo{volume}{36}, \bibinfo{number}{1} (\bibinfo{year}{1984}), \bibinfo{pages}{49--64}.
\newblock


\bibitem[Sachdeva and Zhao(2023)]%
        {SimpleParallelLaplacianSolver23}
\bibfield{author}{\bibinfo{person}{Sushant Sachdeva} {and} \bibinfo{person}{Yibin Zhao}.} \bibinfo{year}{2023}\natexlab{}.
\newblock \showarticletitle{A Simple and Efficient Parallel Laplacian Solver}. In \bibinfo{booktitle}{\emph{SPAA}}. \bibinfo{pages}{315--325}.
\newblock


\bibitem[Shi et~al\mbox{.}(2014)]%
        {density-clustering-sigmod14}
\bibfield{author}{\bibinfo{person}{Jieming Shi}, \bibinfo{person}{Nikos Mamoulis}, \bibinfo{person}{Dingming Wu}, {and} \bibinfo{person}{David~W. Cheung}.} \bibinfo{year}{2014}\natexlab{}.
\newblock \showarticletitle{Density-based place clustering in geo-social networks}. In \bibinfo{booktitle}{\emph{SIGMOD}}. \bibinfo{pages}{99--110}.
\newblock


\bibitem[Sinop et~al\mbox{.}(2021)]%
        {RobustRouting21}
\bibfield{author}{\bibinfo{person}{Ali~Kemal Sinop}, \bibinfo{person}{Lisa Fawcett}, \bibinfo{person}{Sreenivas Gollapudi}, {and} \bibinfo{person}{Kostas Kollias}.} \bibinfo{year}{2021}\natexlab{}.
\newblock \showarticletitle{Robust Routing Using Electrical Flows}. In \bibinfo{booktitle}{\emph{SIGSPATIAL}}. \bibinfo{pages}{282--292}.
\newblock


\bibitem[Spielman and Srivastava(2008)]%
        {GraphSparsificationEff08}
\bibfield{author}{\bibinfo{person}{Daniel~A. Spielman} {and} \bibinfo{person}{Nikhil Srivastava}.} \bibinfo{year}{2008}\natexlab{}.
\newblock \showarticletitle{Graph sparsification by effective resistances}. In \bibinfo{booktitle}{\emph{STOC}}. \bibinfo{pages}{563--568}.
\newblock


\bibitem[Tetali(1991)]%
        {tetali1991random}
\bibfield{author}{\bibinfo{person}{Prasad Tetali}.} \bibinfo{year}{1991}\natexlab{}.
\newblock \showarticletitle{Random walks and the effective resistance of networks}.
\newblock \bibinfo{journal}{\emph{Journal of Theoretical Probability}} \bibinfo{volume}{4}, \bibinfo{number}{1} (\bibinfo{year}{1991}), \bibinfo{pages}{101--109}.
\newblock


\bibitem[Topping et~al\mbox{.}(2022)]%
        {OversquashingICLR22}
\bibfield{author}{\bibinfo{person}{Jake Topping}, \bibinfo{person}{Francesco~Di Giovanni}, \bibinfo{person}{Benjamin~Paul Chamberlain}, \bibinfo{person}{Xiaowen Dong}, {and} \bibinfo{person}{Michael~M. Bronstein}.} \bibinfo{year}{2022}\natexlab{}.
\newblock \showarticletitle{Understanding over-squashing and bottlenecks on graphs via curvature}. In \bibinfo{booktitle}{\emph{ICLR}}.
\newblock


\bibitem[Wei(2010)]%
        {TEDISIGMOD10}
\bibfield{author}{\bibinfo{person}{Fang Wei}.} \bibinfo{year}{2010}\natexlab{}.
\newblock \showarticletitle{{TEDI:} efficient shortest path query answering on graphs}. In \bibinfo{booktitle}{\emph{SIGMOD}}. \bibinfo{pages}{99--110}.
\newblock


\bibitem[Yang and Tang(2023)]%
        {ResistanceYang}
\bibfield{author}{\bibinfo{person}{Renchi Yang} {and} \bibinfo{person}{Jing Tang}.} \bibinfo{year}{2023}\natexlab{}.
\newblock \showarticletitle{Efficient Estimation of Pairwise Effective Resistance}.
\newblock \bibinfo{journal}{\emph{Proc. {ACM} Manag. Data}} \bibinfo{volume}{1}, \bibinfo{number}{1} (\bibinfo{year}{2023}), \bibinfo{pages}{16:1--16:27}.
\newblock


\bibitem[Yin et~al\mbox{.}(2012)]%
        {long-tail-recommendation-vldb12}
\bibfield{author}{\bibinfo{person}{Hongzhi Yin}, \bibinfo{person}{Bin Cui}, \bibinfo{person}{Jing Li}, \bibinfo{person}{Junjie Yao}, {and} \bibinfo{person}{Chen Chen}.} \bibinfo{year}{2012}\natexlab{}.
\newblock \showarticletitle{Challenging the Long Tail Recommendation}.
\newblock \bibinfo{journal}{\emph{VLDB}} \bibinfo{volume}{5}, \bibinfo{number}{9} (\bibinfo{year}{2012}), \bibinfo{pages}{896--907}.
\newblock


\bibitem[Zhang et~al\mbox{.}(2021)]%
        {DynamicHubLabelingICDE21}
\bibfield{author}{\bibinfo{person}{Mengxuan Zhang}, \bibinfo{person}{Lei Li}, \bibinfo{person}{Wen Hua}, \bibinfo{person}{Rui Mao}, \bibinfo{person}{Pingfu Chao}, {and} \bibinfo{person}{Xiaofang Zhou}.} \bibinfo{year}{2021}\natexlab{}.
\newblock \showarticletitle{Dynamic Hub Labeling for Road Networks}. In \bibinfo{booktitle}{\emph{ICDE}}. \bibinfo{pages}{336--347}.
\newblock


\bibitem[Zhang and Yu(2022)]%
        {RelativeSubboundednessSIGMOD22}
\bibfield{author}{\bibinfo{person}{Yikai Zhang} {and} \bibinfo{person}{Jeffrey~Xu Yu}.} \bibinfo{year}{2022}\natexlab{}.
\newblock \showarticletitle{Relative Subboundedness of Contraction Hierarchy and Hierarchical 2-Hop Index in Dynamic Road Networks}. In \bibinfo{booktitle}{\emph{{SIGMOD}}}. \bibinfo{pages}{1992--2005}.
\newblock


\bibitem[Zheng et~al\mbox{.}(2023)]%
        {ReinforcementLearningTreeDecomposition23}
\bibfield{author}{\bibinfo{person}{Bolong Zheng}, \bibinfo{person}{Yong Ma}, \bibinfo{person}{Jingyi Wan}, \bibinfo{person}{Yongyong Gao}, \bibinfo{person}{Kai Huang}, \bibinfo{person}{Xiaofang Zhou}, {and} \bibinfo{person}{Christian~S. Jensen}.} \bibinfo{year}{2023}\natexlab{}.
\newblock \showarticletitle{Reinforcement Learning based Tree Decomposition for Distance Querying in Road Networks}. In \bibinfo{booktitle}{\emph{ICDE}}. \bibinfo{pages}{1678--1690}.
\newblock


\end{thebibliography}
\end{document}